\documentclass[a4paper,11pt]{article}

\usepackage{color}
\usepackage{bbm}
\usepackage{graphicx}
\usepackage{epstopdf}
\usepackage{amsmath}
\usepackage{appendix}
\usepackage{amsthm}
\usepackage[boxruled,algo2e]{algorithm2e}
\usepackage{amsfonts}
\usepackage{amssymb}
\usepackage{enumitem}
\usepackage[margin=1cm, font={small,it}]{caption}
\usepackage[utf8]{inputenc}
\usepackage{multirow}
\usepackage{color}
\usepackage{a4wide}
\usepackage[english]{babel}
\usepackage{hyperref}
\hypersetup{colorlinks=true,citecolor=blue, urlcolor=blue, linkcolor=black, filecolor=black}
\usepackage[draft]{fixme}

\newtheorem{theorem}{Theorem}[section]
\newtheorem{lemma}[theorem]{Lemma}
\newtheorem{proposition}[theorem]{Proposition}
\newtheorem{corollary}[theorem]{Corollary}
\newtheorem{definition}{Definition}

\newtheorem{remark}{Remark}

\newtheorem{assumption}{Assumption}

\numberwithin{equation}{section}

\everymath{\displaystyle}

\graphicspath{{image/}}

\def\1{\mathbbm{1}}
\def\rme{\mathrm{e}}
\def\rmd{\mathrm{d}}

\def\IdVect{\mathbf{1}}
\def\IdMat{\mathrm{Id}}
\def\HR{\mathrm{Hr}}
\def\IR{\mathrm{Ir}}

\title{Modelling Bid and Ask prices using constrained Hawkes processes\\
Ergodicity and scaling limit}

\newcommand{\footnoteremember}[2]{
  \footnote{#2}
  \newcounter{#1}
  \setcounter{#1}{\value{footnote}}
}
\newcommand{\footnoterecall}[1]{
  \footnotemark[\value{#1}]
}
\author{Ban Zheng\footnoteremember{myfootnote1}{Natixis, Equity Markets. E-mail: ban.zheng@melix.net. The authors would like to thank the members of Natixis quantitative research team for fruitful discussions.}
\footnoteremember{myfootnote0}{T\'el\'ecom ParisTech (CNRS LTCI, Insitut Mines-Télécom). E-mail: francois.roueff@telecom-paristech.fr}
\footnoteremember{myfootnote}{Corresponding author. Chair of Quantitative Finance, Ecole Centrale Paris, MAS Laboratory. E-mail: frederic.abergel@ecp.fr}
\and
Fran\c{c}ois Roueff\footnoterecall{myfootnote0}
\and
Fr\'ed\'eric Abergel\footnoterecall{myfootnote}
}

\begin{document}
\date{\today}
\maketitle
\abstract {We introduce a multivariate point process describing the dynamics of
  the Bid and Ask price of a financial asset. The point process is similar to a
  Hawkes process, with additional constraints on its intensity corresponding to
  the natural ordering of the best Bid and Ask prices. We study this process in
  the special case where the fertility function is exponential, so that the
  process is entirely described by an underlying Markov chain including the
  constraint variable. Natural, explicit conditions on the parameters are
  established that ensure the ergodicity of the chain. Moreover, scaling limits
  are derived for the integrated point process.}

{\bf Keywords:} Point processes. Hawkes processes. Limit order book. Microstructure noise. Bid-Ask Spread. Ergodicity. Scaling limit. Markov model.

{\bf Mathematics Subject Classification:} 37A25, 60F05, 60G55, 60J05.

{\bf JEL Classification:} C62.

\newpage

\section{Introduction}

Modelling the mechanisms of liquidity taking and providing is key to
understanding the finer scale, microscopic dynamics of the price of a financial
asset. In the context of electronic, order-driven markets, the modelling effort
naturally occurs at the order-book level, and the study of order-driven markets
has received much attention over the past two decades. On the one hand,
extensive statistical studies of the limit order book dynamics and information
content have been performed, see e.g. \cite{biais-hillion-spatt-1995,
  gourieroux-jasiak-lefol-1999, smith-farmer-gillemot-krishnamurthy-2003,
  farmer-gillemot-lillo-mike-sen-2004, hollifield-miller-sandras-2004,
  bouchaud-farmer-lillo-2008, fletcher-hussain-shawe-taylor-2010,
  chakraborti-munitoke-patriarca-abergel-2011-I, zheng-moulines-abergel-2012};
on the other hand, stochastic models, including equilibrium models, agent-based
models and Markov models, have emerged as a mathematical representation of the
limit order book, see \cite{parlour-1998, rosu-2009, cont-stoikov-talreja-2010,
  chakraborti-munitoke-patriarca-abergel-2011-II, abergel-jedidi-2011,
  cont-delarrard-2011}.

It is quite clear from the references above, or better, from a direct
inspection of limit order book high frequency data, that a full-fledged
order-book model rapidly becomes cumbersome and may sometimes hide simple yet
essential mechanisms. One may therefore question the level of complexity
required to understand the essential features of the evolution of a financial
asset in the high frequency realm. Obviously, a minimal description should
model the Bid/Ask spread, since this quantity reveals essential information on
liquidity. Such a description should also endogenously account for the
interplay between the spread and the mid-price.
It is our goal in this paper to introduce a phenomenological,
point-process-based description of the joint dynamics of the mid-price and the
spread, or equivalently, of the respective best Bid and Ask prices. This
description can be viewed as a simplified form of an order-book model with only
two limits and infinite liquidity available at each best bid and ask prices, so
that its aim is to represent the joint dynamics of the price and the spread in
the absence of liquidity costs. Therefore, it can be viewed as an upgrade of
the classical modelling approach that addresses only the mid-price
dynamics. Actually, we introduce a more general model that can deal with
several constraints, a situation that can happen when studying for instance the
joint dynamics of the Bid and Ask prices of several tradable assets.

The theory of point processes provides a natural tool to model the dynamics of
the price of a financial asset at the level of individual changes, or ``tick"
level. Bauwens and Hautsch \cite{bauwens-hautsch-2006} is a comprehensive
introduction to the application of point processes in financial time
series. Such processes allow one to directly model the arrival of events
affecting the price dynamics. Recently, Hawkes processes have been introduced
in market microstructure. Hawkes processes belong to the class of self-exciting point
processes, where the intensity is driven by a weighted function of the time
distance to previous points of the process. These processes originate from the
literature in seismology, where they were introduced to model the arrival of replicas in the aftermath of an earthquake. Later, they have been successfully applied to financial markets in the context of default risk modelling \cite{errais2010}, contagion across equity markets \cite{sahalia-cacho-diaz-laeven-2010}, or foreign exchange modelling \cite{embrechts-liniger-lin-2011}. To the best of our knowledge, the first application of Hawkes processes in
financial time series is \cite{bowsher-2002}. Since then, there
has been a growing interest in using Hawkes processes to model
high frequency financial data, see \cite{bauwens-2003, hewlett-2006,
  bowsher-2007, large-2007, bacry-delattre-hoffmann-muzy-2010, munitoke-pomponio-2011, bacry-delattre-hoffmann-muzy-2012}.
There are some very convincing reasons to this growing interest: first, Hawkes
processes offer a very natural way to model the dependence structures of the
arrival of changes in the price of one or several assets, see
e.g. \cite{bacry-delattre-hoffmann-muzy-2010, bacry-dayri-muzy-2011},
leading to nice representations for high frequency volatility and
correlation. Moreover, the mutual excitation mechanism built in multivariate Hawkes processes is consistant with phenomena observed in market microstructure, where the interplay between liquidity providing and taking is crucial, see \cite{munitoke-pomponio-2011} for empirical evidence.
Obviously, other important phenomena also affect the price dynamics and liquidity supply: macro-economic news, idiosyncratic news, private information... but they are exogenous and cannot be built directly into a price model. In this work, we do not address that aspect of the modelling, but rather concentrate on the endogenous interplay between orders of various types.

It is noteworthy that previous studies using point processes focus on the description of the evolution of the mid-price of one or several assets, without incorporating the spread-price interaction.
In this contribution, we wish to model the dynamics of the best Bid and Ask prices in a stylized
limit order book; and use Hawkes processes to that aim. There is however a main difficulty in using classical Hawkes processes: there is a natural ordering between the best Bid and Ask prices, and they must evolve according to specific constraints. In particular, the Bid/Ask spread is always nonnegative, and is therefore bounded away from zero by its lower bound, the tick size. This means that, as soon as the minimal bid-ask spread is attained, the best bid price cannot increase
until the best ask price goes up, and conversely, the best Ask price cannot go down until the best Bid price does. Hence, we add a
constraint variable to the classical multivariate Hawkes process to model the
dynamics of the best limit prices: in a situation where one of the prices is constrained, the intensity of the arrival process of events that would violate the constraint becomes zero, and stays there until the constraint becomes inactive.

Below is an example of a concrete problem that motivates our theoretical approach:
Figure~\ref{fig:FTE_BBBA} shows a snapshot of the price evolution of a tradable asset, both in
physical time and event time.
\begin{figure}[htbp]
\centering
\begin{tabular}{cc}
\includegraphics[trim=0mm 0mm 0mm 0mm, clip, height=6.8cm, width=6.8cm, angle=-90]{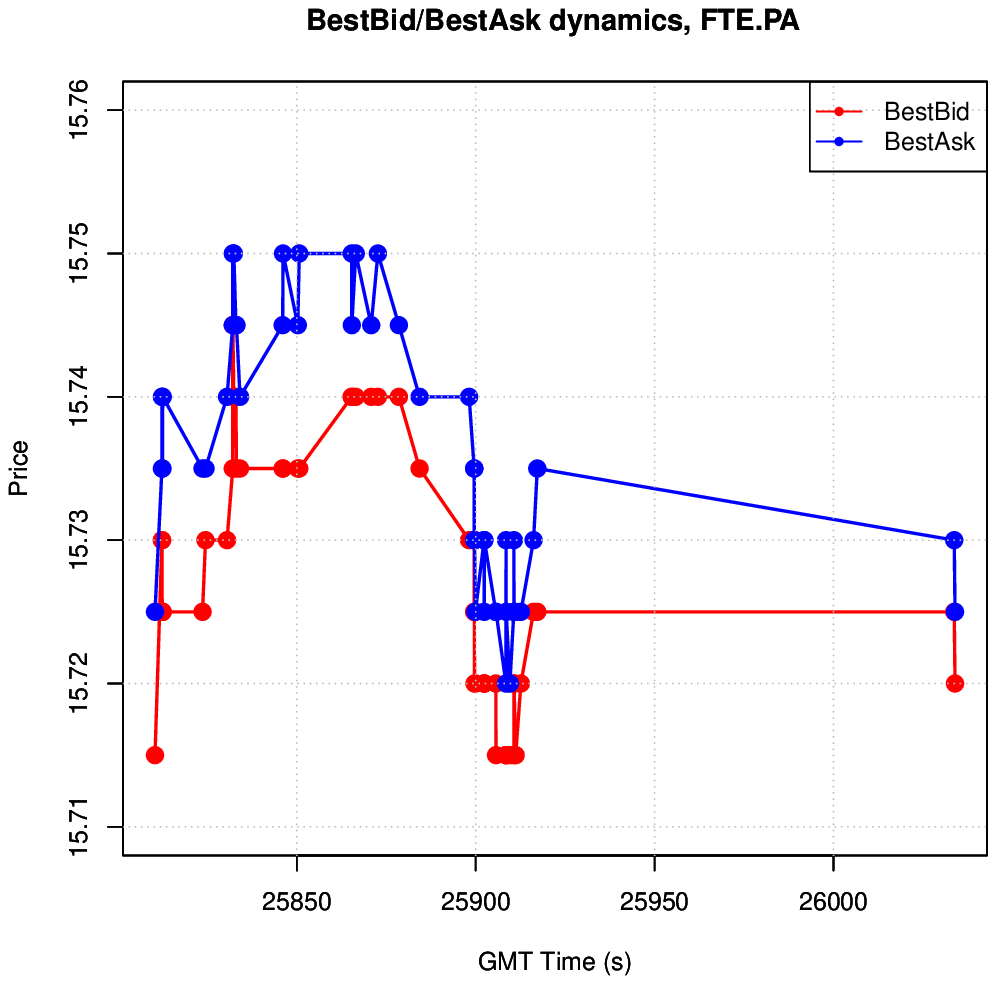}
\includegraphics[trim=0mm 0mm 0mm 0mm, clip, height=6.8cm, width=6.8cm, angle=-90]{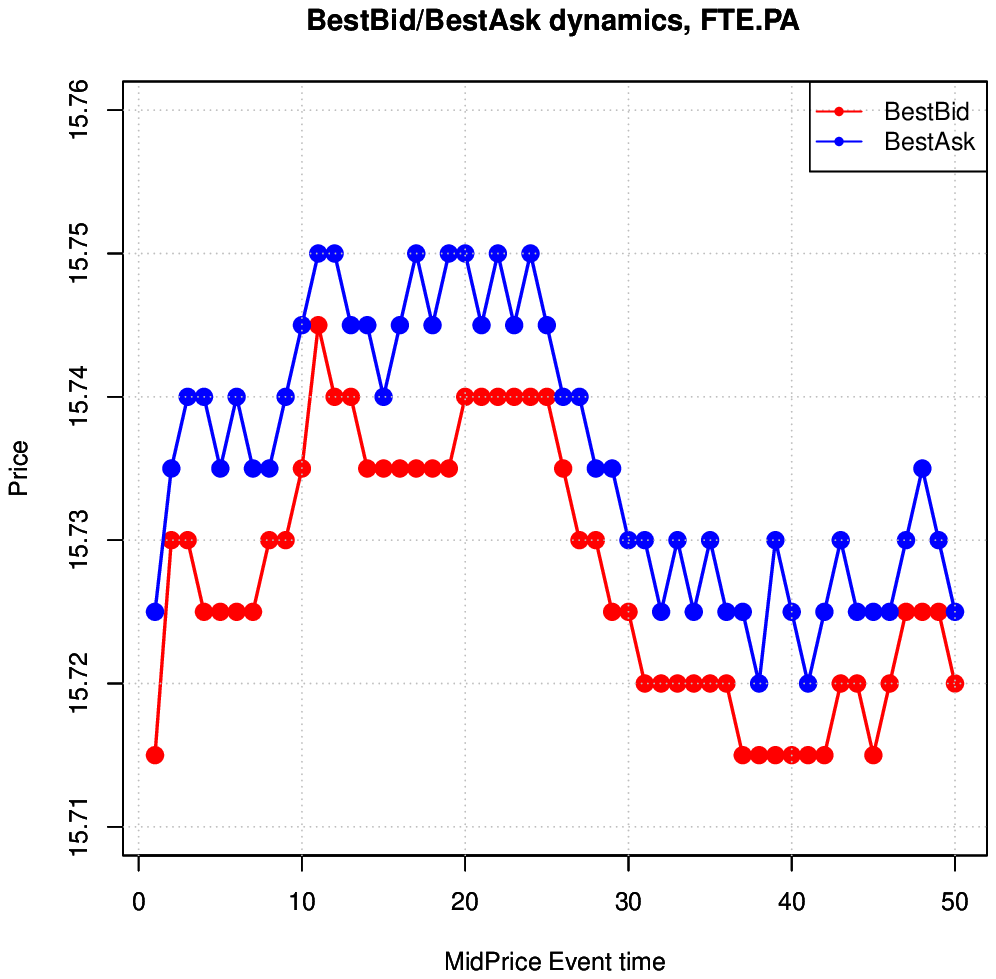}
\end{tabular}
\caption{Limit order book of Orange~:~Evolution of Best Bid and Best Ask prices over a short time window of
  April 1, 2011. The evolution is displayed in physical time (left) and in event time (right).}
\label{fig:FTE_BBBA}
\end{figure}
The simplest description of the dynamics of the best Bid and Ask prices is based on four types of events:
\begin{itemize}
  \item \emph{Event 1}: Best Ask price moves upward one tick,
  \item \emph{Event 2}: Best Ask price moves downward one tick,
  \item \emph{Event 3}: Best Bid price moves upward one tick,
  \item \emph{Event 4}: Best Bid price moves downward one tick.
\end{itemize}
The arrival of these events is described by a marked point process $\mathbf{N}$.

As argued above, Hawkes processes offer a sensible framework to model the arrival times of
these events. However, since the Best Bid has to remain at least one tick below
the Best Ask, we need to introduce a constraint. Namely, defining the spread as
$S=BestAsk-BestBid$ (in ticks), events 2 and 3 cannot happen whenever $S=1$.

Our results show that, under some explicit drift conditions, the spread
variable is stationary and the rescaled price behaves like a random walk. More
generally, we study Hawkes processes with multivariate constraints: in a
Markovian setting with exponential fertility functions, we provide sufficient
conditions that ensure the ergodicity of the embedded chain, and study the
scaling limits of the integrated point process. The main methodology is based
on establishing geometric drift conditions as detailed in
\cite[Chapter~15]{meyn-tweedie-2009}.

The paper is structured as follows: in
Section~\ref{sec:main-assump-and-notation}, we introduce the constrained multivariate Hawkes
processes and present the main assumptions and
notation. Section~\ref{sec:main-results} gives the main results on the
geometrical ergodicity and the large scale behavior. The
application to a limit order book is presented in
Section~\ref{sec:sclobchp-application}. We have gathered the detailed proofs in
Section~\ref{sec:proofs}. Some useful technical results are postponed to
Appendix~\ref{app:technical}.

%
%

\section{Main assumptions and notation}
\label{sec:main-assump-and-notation}

\subsection{Multivariate Hawkes process with constraints}
\label{sec:constr-hawk-proc}

Let $p,q$ be two positive integers. We introduce a constrained multivariate Hawkes process
defined by a couple $(\mathbf{S},\mathbf{N})$, where $\{\mathbf{S}(t),\,t\geq0\}$ is a multivariate \emph{spread} process taking its values in
$\mathbb{Z}_+^q=\{1,2,\dots\}^q$ and $\mathbf{N}$ is a $p$-dimensional
multivariate point process. We shall describe the dynamics of this model by
the conditional intensity of $\mathbf{N}$ as defined
in \cite{bremaud-massoulie-1996} and an evolution equation for $\mathbf{S}$.

We assume that $\mathbf{S}$ is right-continuous with left limits and denote by
$\mathbf{S}(t -)$ its left limit at time $t$. The conditional intensity of the
multivariate point process $\mathbf{N}$ at time $t$ shall depend on the
spread variable $\mathbf{S}(t-)$. For notational compactness we will
see $\mathbf{N}$ as a marked point process having arrivals in $\mathbb{R}$ with
marks in $\{1,\dots,p\}$. We will denote by $\mathbf{N}_j$, $j=1,\dots,p$ the
$p$ point processes corresponding to each mark $j$, namely
$$
\mathbf{N}=\sum_{j=1}^p \mathbf{N}_j\otimes\delta_{j}\;.
$$
In other words, each thinned unmarked process $\mathbf{N}_j$ is given from
$\mathbf{N}$ by the formula
$$
\mathbf{N}_j(g)=\mathbf{N}(g\otimes\1_{\{j\}})\;,
$$
for all non-negative measurable function $g$ defined on $\mathbb{R}$. Here and
all along the paper, point processes are seen as random point measures and we
use the notation $\xi(g)$ for the integral of a function $g$ with respect to a
measure $\xi$. We further use the symbol $\otimes$ for the product of measures
(that is, $\xi\otimes\zeta(A\times B)=\xi(A)\zeta(B)$) and for the tensor
product of functions (that is, $g\otimes h(x,y)=g(x)\,h(y)$).  The spread
variable $\mathbf{S}(t-)$ acts on the point process by constraining the set of
possible marks for the first event arriving after time $t$ in $\mathbf{N}$. The
constraint is added in the usual expression of the conditional intensity of a
multivariate Hawkes process. More precisely, for all time $t\geq0$ and any mark
$i=1,\dots,p$, the conditional intensity of $\mathbf{N}_i$ given
$\mathbf{N}_{|(-\infty,t)}$ and $\mathbf{S}(t-)$ is set as
\begin{align}\label{eq:cond-intensity-cons}
  \boldsymbol{\mu}(t,i) =
  \begin{cases}
    0 &\text{ if $\mathbf{S}(t-)\in \mathbf{A}_i$}\\
\boldsymbol{\mu}_0(i) +
    \sum_{j=1}^{p}{\int_{(-\infty,t)}{\phi_{i, j}(t-u)\;
        \mathbf{N}_{j}(\mathrm{d} u)}}&\text{ otherwise}\;,
  \end{cases}
\end{align}
where $\boldsymbol{\mu}_0(i)$ is the immigrant
rate of mark $i$, $\phi_{i,j}$ is the fertility rate for producing a mark $i$
event from a mark $j$ event, and $\mathbf{A}_1,\dots,\mathbf{A}_p$ are products
of finite subsets of
$\mathbb{Z}_+^q$, $\mathbf{A}_i=A_i(1)\times\dots\times A_i(q)$.

In turn, an arrival in $\mathbf{N}$ increases the value of the spread by a value
only depending on the mark. More precisely, the spread process $\mathbf{S}$
satisfies the evolution equation, for all $t>u\geq0$,
\begin{align}\label{eq:state-evolution}
\mathbf{S}(t) = \mathbf{S}(u) + \mathbf{N}(\1_{(u, t]}\otimes \mathbf{J})\;,
\end{align}
where $\mathbf{J}$ is defined on $\{1,\dots,p\}$ with values in $\mathbb{Z}^q$. In other words, at each arrival in $\mathbf{N}$ with mark equal to $i$, $\mathbf{S}$ jump by an increment given by $\mathbf{J}(i)$.

Note that equations~(\ref{eq:cond-intensity-cons}) and~(\ref{eq:state-evolution}) are valid for $t\geq0$ and $t>u\geq0$ respectively so that the distribution of the process
$$
\{\left(\mathbf{S}(t),\mathbf{N}_j((0,t])\right),\,t\geq0,\,j=1,\dots,p\}
$$
is defined conditionally to the initial conditions given by $\mathbf{S}(0-)$ and
$\mathbf{N}$ restricted to $(-\infty,0)$. To summarize, this distribution is
parameterized by the constrained sets $\mathbf{A}_1,\dots,\mathbf{A}_p$, the immigrants rates
$\boldsymbol{\mu}_0(i)$, $i=1,\dots,p$, the cross--fertility rate functions
$\phi_{i,j}:[0,\infty)\to[0,\infty)$, $i,j=1,\dots,p$, and the jump
increments $\mathbf{J}(i)\in\mathbb{Z}^q$ for $i\in\{1,\dots,p\}$.

Meanwhile, we define an unconstrained multivariate Hawkes process $\mathbf{N}'$
with the same fertility rate and immigrant intensity (but without constraints),
that is, with conditional intensity on $[0,\infty)$ satisfying, for $i=1,
\dots, p$,

\begin{align*}
\boldsymbol{\mu}'(t,i) = \boldsymbol{\mu}_0(i) + \sum_{j=1}^{p}{\int_{(-\infty,t)}{\phi_{i, j}(t-u)\;
     \mathbf{N}'_{j}( \mathrm{d} u)}}, \quad t\geq0\;,
\end{align*}
where $\mathbf{N}'_j=\mathbf{N}'(\cdot\otimes\1_{\{j\}})$ for $j=1,\dots,p$.

The \emph{average fertility} matrix is defined by
\begin{equation}
  \label{eq:aleph-def}
\aleph = [\alpha_{i,j}]_{i,j=1,\dots,p}\quad\text{with}\quad
\alpha_{i, j} = \int_0^{\infty} \phi_{i, j}(t) \mathrm{d}t,\quad 1\leq i,j \leq p\;.
\end{equation}

The following assumption ensures the existence of a stationary version of the
(unconstrained) multivariate Hawkes process $\mathbf{N}'$, see \cite[Example~8.3(c)]{daley-vere-jones-2003}.

\begin{assumption}
  \label{ass:Phi}
The spectral radius of $\aleph$ is strictly less than $1$.
\end{assumption}

Under Assumption~\ref{ass:Phi}, we will use the following (well defined) vector
\begin{equation}
  \label{eq:defu}
\mathbf{u}=(\IdMat_p-\aleph^T)^{-1}\IdVect_p=\sum_{k\geq0}\aleph^k\IdVect_p\;,
\end{equation}
where $\IdVect_p$ denotes $p$-dimensional ones vector,
$$
\IdVect_p = {\underbrace{[1\;\; \cdots\;\; 1]}_{p}}^T\;,
$$
and $\IdMat_p$ denotes the $p\times p$ identity matrix.

In the following, by convention, we use
bold faces symbols only when $q$ may be larger than one.
In the case where $q$ is set to 1, we shall write $J$
and $S_n$ in unbold faces since they are scalar valued and, similarly, we
write $A_i$ in unbold face for the constraint sets.

For a real valued function $w$ defined on $\{1,\dots,p\}$, we shall use the notation
$$
\overrightarrow{w}=[w(1),\dots,w(p)]^T\;.
$$
If $w$ is vector-valued, we use the same notation to obtain a matrix. For instance,
$\overrightarrow{\mathbf{J}}$ is a $p\times q$ matrix.

\subsection{A very special case}
\label{sec:very-special-case}
When we study the ergodicity of the process, we investigate the joint
ergodicity of the spread variable and of the point process.
A very special case is obtained by setting
$$
\phi_{i,j}\equiv0,\quad i,j=1,\dots p\;,
$$
and $q = 1$.  In this case, we only need to study the stationarity of the
spread variable since the intensities are constant. Suppose moreover that $J$
only takes values $+1$ and $-1$ and the \emph{constraint} sets $A_i$ are all
equal to $\{1\}$ if $J(i)=-1$ and to $\emptyset$ if $J(i)=1$. In this case,
$\{S(t),\,t\geq0\}$ is a birth-death process on $\mathbb{Z}_+$ with constant
birth rate $\mu_+=\sum_{i\in I_+}\boldsymbol{\mu}_0(i)$ and constant death rate
$\mu_-=\sum_{i\in I_-}\boldsymbol{\mu}_0(i)$, where $I_+,I_-$ is the partition
of $\{1,\dots,p\}$ corresponding to the indices where $J$ takes values $+1$ and
$-1$, respectively. By \cite[Corollary~2.5] {asmussen-2003}, $(S(t))_{t\geq0}$
is ergodic if and only if
\begin{equation}
  \label{eq:ergodicity-birthdeath}
\mu_+ - \mu_- = \overrightarrow{J}^T \boldsymbol{\mu}_0 = \sum_{i=1}^p J(i) \boldsymbol{\mu}_0(i)
<0 \;.
\end{equation}
Here and in the following, $A^T$ denotes
the transpose of matrix $A$.
\subsection{A simple order book}
\label{sec:sclobchp-application-desc}

In the introduction, we presented a simple order book with only two limit
prices, namely the best Bid/Ask prices. We now detail how the constrained
multivariate Hawkes process can be used to model this situation.
Recall that the four possible events are: Best Ask price moves upward one tick
($i=1$), Best Ask price moves downward one tick ($i=2$), Best Bid price moves
upward one tick ($i=3$), Best Bid price moves downward one tick
($i=4$). Moreover, the spread variable $S$ is the standard (univariate) Bid-Ask
spread.  In the formalism of Section~\ref{sec:constr-hawk-proc}, this
case corresponds to $p=4$, $q=1$ and the constrained sets $A_i$ are defined by
\begin{equation}
\label{eq:sclob-constraints}
A_i =
\begin{cases}
\emptyset\;, & i = 1 \text{ or } 4 \\
{\{1\}}\;, & i = 2 \text{ or } 3
\end{cases}
\end{equation}
The spread variable $S$ evolves according to~(\ref{eq:state-evolution}) with a
scalar $J$ defined by
\begin{equation}
\label{eq:sclob-J}
J(i) =
\begin{cases}
1 \;, & \quad i=1 \text{ or } 4 \\
-1 \;, & \quad i=2 \text{ or } 3
\end{cases}
\end{equation}
Hence the constrained multivariate Hawkes process provides a general
framework for modelling this simple order book. If $\aleph=0$ as in
Section~\ref{sec:very-special-case}, Condition~(\ref{eq:ergodicity-birthdeath})
for the ergodicity of the model becomes
\begin{equation}
  \label{eq:ergodicity-lob-birth-death}
 \boldsymbol{\mu}_0(2)+\boldsymbol{\mu}_0(3)< \boldsymbol{\mu}_0(1)+\boldsymbol{\mu}_0(4)  \;.
\end{equation}
This condition can be interpreted as a drift that forces the mean intensity of
decreases of the spread variable $S$ to be larger than the mean intensity of
increases of the spread variable $S$.  However, in such a special case, given
$\mathbf{N}$ and $S$ up to time $t$, the conditional distribution of
$\mathbf{N}$ restricted to $(t,\infty)$ only depends on $S(t)$ and not on the
past events. Hence this model has very limited interest for modelling the
dynamics of the limit order book. Therefore we wish to investigate the case
where $\aleph$ is non-zero and, in particular, to determine how the ergodicity
condition~(\ref{eq:ergodicity-lob-birth-death}) should be adapted to this
case. This will be answered in Section~\ref{sec:sclobchp-application}.

\subsection{The Markov assumption}
\label{sec:markov-assumption}

We shall consider a particular shape of the fertility functions $\phi_{i,j}$ which implies a Markov
property for the model. This allows us to obtain very precise results on the
ergodicity of $(\mathbf{S},\mathbf{N})$ and is also of interest in
applications since the model is parameterized by a restricted set of well
understood parameters.

From now on, we suppose that there exists $\beta>0$ such that
$$
\phi_{i, j}(u) = \alpha_{i, j}\beta \mathrm e^{-\beta u}\;,\quad
1\leq i,j\leq p\;.
$$
This parameter is the reciprocal of a time~: the larger $\beta$ is, the shorter
the dependence persists along the time between successive events.
We denote a new process defined
in the state space $\mathsf{X}=\mathbb{Z}_+^q\times{\mathbb{R}_{+}^{p}}$ by
$\mathbf{X}(t)=(\mathbf{S}(t), \boldsymbol{\boldsymbol{\lambda}}(t))$ where
$\boldsymbol{\boldsymbol{\lambda}}(t) = [\boldsymbol{\lambda}(t,1), \dots,
\boldsymbol{\lambda}(t,p)]^T$ with, for all $i\in\{1, \dots, p\}$,
$$
\boldsymbol{\lambda}(t,i) = \sum_{j=1}^{p}{\int_{(-\infty,t)}{\phi_{i,
      j}(t-u) \mathbf{N}_{j}(\mathrm{d} u)}}\;.
$$
\begin{proposition}
\label{prop:markov}
The process $\{\mathbf{X}(t),\,t\geq0\}$ is a Markov process.
\end{proposition}
This property directly follows from the exponential form of the fertility
functions $\phi_{i,j}$ (see \cite{oakes-1975,errais2010,liniger-2009} in the
unconstrained case). We omit the proof of Proposition~\ref{prop:markov} which
is similar to the unconstrained case.

In the following, we shall essentially rely on a Markov chain
$(\mathbf{Z}_n)_{n\geq0}$ for describing the dynamics of the process
$\{\mathbf{X}(t),\,t\geq0\}$. We now detail the construction of this Markov
chain. Let $\mu^{0}_{0}$ denote an arbitrary positive constant (say
$\mu_0^0=1$). An embedded Markov chain $(\mathbf{X}_n)_{n\geq0}$ shall be
defined by sampling the continuous time process $\{\mathbf{X}(t),\,t\geq0\}$ at
increasing discrete times $(T_n)_{n\geq 0}$ gathering the arrivals of
$\mathbf{N}'$ (the unconstrained Hawkes process defined in
Section~\ref{sec:constr-hawk-proc}) and the arrivals of an independent
homogeneous Poission point process (PPP) $\mathrm{n}_0$ with intensity
$\mu^0_0$. Observe that $\mathbf{S}(t)$ is constant between two consecutive
arrivals of the point process $\mathbf{N}$. Other sampling points are those
generated by $\mathrm{n}_0$ and those present in $\mathbf{N}'$ but not in
$\mathbf{N}$ (which correspond to the arrivals of $\mathbf{N}'$ that violate
the constraint). To these points we shall assign a mark of new type $i=0$, see
the second line in~(\ref{eq:hazard-rate}) below).  Hence, almost surely,
$\mathbf{S}$ does not change at arrivals corresponding to marks of this
type. The sample time instants with marks of type $i=0$ are artificially
added to avoid the periodicity of the embedded chain, see the end of the proof
of Proposition~\ref{prop:general-smallset}.

Each jump of $\boldsymbol{\lambda}(t)$ corresponds to an arrival in the point
process $\mathbf{N}$. We denote its mark by $I_n$, which takes values in
$\{1,\dots,p\}$, where $n$ is the positive integer such that the corresponding
jump instant $t$ equals $T_n$. For other sample times $T_n$, which thus
correspond to the arrivals of $\mathrm{n}_0$, we set $I_n=0$. Hence the marks of
the embedded chain will take values in the set $\{0,1,\dots,p\}$.

Finally, we define $\Delta_1=T_1$ and, for each $n \geq 2, \Delta_{n} = T_{n} -
T_{n-1}$, $\mathbf{S}_{n} = \mathbf{S}(T_n)$, $\boldsymbol{\lambda}_{n} =
\boldsymbol{\lambda}(T_n)$ and $\mathbf{X}_n=\mathbf{X}(T_n)=(\mathbf{S}_{n},\boldsymbol{\lambda}_{n})$.
Then, given $(\mathbf{S}_{0},\boldsymbol{\lambda}_{0})$ and $T_0=0$, we can
generate the sequence $(\Delta_n,\mathbf{S}_{n},\boldsymbol{\lambda}_{n})$,
$n=0,1,\dots$ iteratively as follows.

\begin{align*}
\Delta_{n+1} =
\min(\Delta_{n+1}^{0}, \Delta_{n+1}^{1}, \cdots, \Delta_{n+1}^{p})\;,
\end{align*}
where, given $\mathbf{X}_{0},\ldots, \mathbf{X}_{n}$, $\Delta_{1}, \ldots,
\Delta_{n}$ and $I_1,\dots,I_n$, the conditional distribution of
$\Delta_{n+1}^{0}$, $\Delta_{n+1}^{1}$, ... and $\Delta_{n+1}^{p}$ is that of
positive independent variables whose marginal distributions are determined by
hazard rates $\HR_{0},\HR_{1},\dots,\HR_p$ defined for each $i$ and $t\geq0$ by
\begin{align}
\label{eq:hazard-rate}
\HR_i(t) & =
\begin{cases}
(\boldsymbol{\mu}_{0}(i) + \boldsymbol{\lambda}_n (i) \rme^{-\beta t})\1_{\mathbf{A}_i^c}(\mathbf{S}_n), & i > 0 \;,\\
\mu^{0}_{0} + \sum_{j=1}^p \1_{\mathbf{A}_j}(\mathbf{S}_n) [\boldsymbol{\mu}_{0}(j) + \boldsymbol{\lambda}_n (j) \rme^{-\beta t}], & i = 0\;.
\end{cases}
\end{align}
Then, the event type at time $T_{n+1}$ is
$$
I_{n+1} =
\arg\min_{i\in\{0,\cdots,p\}}(\Delta_{n+1}^{i})\;,
$$
the spread variable at time $T_{n+1}$ is given by
\begin{align}
\label{eq:recurrence-of-S}
\mathbf{S}_{n+1} =\mathbf{S}(T_{n+1})= \mathbf{S}_{n} + \mathbf{J}_o(I_{n+1})\;,
\end{align}
where $\mathbf{J}_o$ is the extension of $\mathbf{J}$ to $\{0,1,\dots,p\}$ defined by
\begin{align}
\label{eq:conv-Jzero}
\mathbf{J}_o(i)=
\begin{cases}
\mathbf{J}(i)&\text{ if $i\geq1$}\\
\mathbf{0}_q&\text{ if $i=0$}\;.
\end{cases}
\end{align}
and, for all $i\in\{1,\dots,p\}$, the self-excitation intensities at time
$T_{n+1}$ are given by
\begin{align}
\label{eq:intensity-dynamic}
\boldsymbol{\lambda}_{n+1}(i)=\boldsymbol{\lambda}(T_{n+1}, i) = \boldsymbol{\lambda}_{n}(i)\rme^{-\beta \Delta_{n+1}} + \beta \alpha_{i,I_{n+1}} \;.
\end{align}
Observe that, for all $T_n\leq t< T_{n+1}$, the continuous time process is
interpolated as
\begin{align*}
\mathbf{S}(t) & = \mathbf{S}_{n}\\
  \boldsymbol{\lambda}(t,i) &=\boldsymbol{\lambda}_n(i) \rme^{-\beta t}\;.
\end{align*}

For any probability distribution
$\nu$ on $\mathsf{X}=\mathbb{Z}_+^q \times \mathbb{R}_+^p$, we denote by
$\mathbb{P}^{\nu}$ the probability corresponding to the initial distribution
$\nu$ at time $t=0$. Notation $\mathbb{E}^{\nu}$ corresponds to the
corresponding expectation. If $\nu$ is a Dirac distribution at point
$\mathbf{x}\in\mathsf{X}$, we will simply write $\mathbb{P}^{\mathbf{x}}$ and
$\mathbb{E}^{\mathbf{x}}$.

The above facts thus imply that
$\mathbf{X}_n=(\mathbf{S}_{n},\boldsymbol{\lambda}_{n})$, $\mathbf{Y}_n =
({I}_n, \mathbf{S}_{n},\boldsymbol{\lambda}_{n})$ and
$\mathbf{Z}_n=(\Delta_n,I_n, \mathbf{S}_n,{\boldsymbol{\lambda}}_n)$, defined
for $n=0,1,\dots$ are Markov chains respectively valued in
$\mathsf{X}=\mathbb{Z}_+^q \times \mathbb{R}_+^p$, $\mathsf{Y} = \{0,1, \dots,
p\} \times\mathbb{Z}_+^q \times \mathbb{R}_+^p$ and $\mathsf{Z}=\mathbb{R}_+
\times \{0,\dots,p\} \times \mathbb{Z}_+^q \times \mathbb{R}_+^p$. We shall
denote by $Q$, $\tilde{Q}$ and $\overline{Q}$ the transition kernels of these
Markov processes, respectively.  We shall use these three transition kernels,
depending on the context. As we shall see, the ergodicity of $\tilde{Q}$ and
$\overline{Q}$ essentially follow from that of $Q$. However, the kernel
$\tilde{Q}$ turned out to be easier to handle for proving the ergodicity when
$q\geq2$, see Theorem~\ref{thm:geometrical-ergodicqGeQ2} below. This is
the reason why we introduced the chain $(\mathbf{Y}_n)_{n\geq0}$. The chain
$(\mathbf{Z}_n)_{n\geq0}$ was introduced because the processes
$\mathbf{N}([0,t])$ and $\mathbf{S}(t)$ can be deterministically expressed for
all $t\in\mathbb{R}_+$ using this chain. This fact is used in
Theorem~\ref{thm:fclt-cont-time}.

\begin{remark}
\label{rem:IvanishesInInitialCondition}
Observe that, from the
description above, the conditional distribution of $(\Delta_1,{I}_1,
\mathbf{S}_{1},\boldsymbol{\lambda}_{1})$ given $(\Delta_0,{I}_0,
\mathbf{S}_{0},\boldsymbol{\lambda}_{0})$ does not depend on $(\Delta_0,I_0)$. In
particular, the whole path $\{\mathbf{Z}_n,\,n\geq1\}$ only depends on the
initial condition set on $\mathbf{X}_0$ and we may write, for all
$\mathbf{z}=(\delta,i,\mathbf{x})\in \mathsf{Z}$ and all $A \in \mathcal{B}(\mathsf{Z})$,
\begin{gather}
  \label{eq:Qtilde}
\overline{Q}(\mathbf{z},A)= \mathbb{P}^{\mathbf{x}}((\Delta_1,I_1,\mathbf{X}_{1}) \in A)\;.
\end{gather}
Clearly the same remark holds for $\tilde{Q}$, namely, $\tilde{Q}(\mathbf{y},A)=
\mathbb{P}^{\mathbf{x}}((I_1,\mathbf{X}_{1}) \in A)$ for all
$\mathbf{y}=(i,\mathbf{x})\in\mathsf{Y}$ and $A \in \mathcal{B}(\mathsf{Y})$.
\end{remark}
Because the ergodicity of $Q$ will be proved using an induction on the number
of constraints $q$, we need to introduce further Markov chains and transition kernels.
To initiate the induction, we define the Markov chain $\{(\check{I}_n, \check{\boldsymbol{\lambda}}_n),\,n\geq0\}$
as the chain valued in  $\{0,1, \dots, p\}\times\mathbb{R}_+^q$ that starts at the same state as $(I_0,\boldsymbol{\lambda}_0)$ but with transition kernel $\check{Q}$, which is defined as $\tilde{Q}$ but with all the $\mathbf{A}_i$'s
replaced by the empty set and without the spread variable. We will call this chain the \emph{unconstrained Hawkes embedded chain}, since it is associated to a classical (unconstrained) multivariate Hawkes point process. This chain corresponds to the case $q=0$.
\begin{remark}
\label{rem:Icheck}
As in the case $q\geq1$, the conditional distribution of $\{(\check{I}_n,
\check{\boldsymbol{\lambda}}_n),\,n\geq1\}$ given the initial condition $(\check{I}_0,
\check{\boldsymbol{\lambda}}_0)$ does not depend on $\check{I}_0$. Hence we
will use the notation  $\mathbb{E}^{\boldsymbol{\ell}}$ or
$\mathbb{P}^{\boldsymbol{\ell}}$ to underline this fact, for instance, for
any $\mathbf{y}=(i,\boldsymbol{\ell})\in\{0,1, \dots,
p\}\times\mathbb{R}_+^p$, $j\in\{0,1, \dots,
p\}$ and Borel subset $A\subset\mathbb{R}_+^p$,
$$
\check{Q}(\mathbf{y},\{j\}\times A)= \mathbb{P}^{\boldsymbol{\ell}}(\check{I}_1=j, \check{\boldsymbol{\lambda}}_1 \in A) \;.
$$
\end{remark}
Furthermore, for any $\mathcal{J}\subseteq\{1,\dots,q\}$, we denote by
$\tilde{Q}^{(-\mathcal{J})}$ the transition kernel defined on
$\{0,1,\dots,p\}\times\mathbb{Z}_+^{q-\#\mathcal{J}}\times\mathbb{R}_+^p$
defined as the transition kernel $Q$ but without the spread variable
$\mathbf{S}_j$, $j\in\mathcal{J}$ and their corresponding constraint sets
$\mathbf{A}_1(j),\dots,\mathbf{A}_p(j)$. In particular, we have
$\tilde{Q}^{(-\emptyset)}=\tilde{Q}$ and
$\tilde{Q}^{(-\{1,\dots,q\})}=\check{Q}$.
Meanwhile, for any $\mathcal{J}\subseteq\{1,\dots,q\}$, we denote by
$\tilde{Q}^{(+\mathcal{J})}$ the transition kernel defined on
$\{0,1,\dots,p\}\times\mathbb{Z}_+^{\#\mathcal{J}}\times\mathbb{R}_+^p$
defined as the transition kernel $Q$ with the spread variable
$\mathbf{S}_j$, $j\in\mathcal{J}$ and their corresponding constraints
$\mathbf{A}_1(j),\dots,\mathbf{A}_p(j)$. In particular
$\tilde{Q}^{(+\{1,\dots,q\})}=\tilde{Q}$ and
$\tilde{Q}^{(+\emptyset)}=\check{Q}$.

\section{Main results}
\label{sec:main-results}
We now present the main results of this work. All the results will rely on the
Markov assumption introduced in Section~\ref{sec:markov-assumption}. Under this
Markov framework, in Section~\ref{sec:irreduc}, we show that the kernel $Q$ and
$\tilde Q$ defined in Section~\ref{sec:markov-assumption} are
$\psi$-irreducible and aperiodic. We provide a partial drift condition in
Section~\ref{sec:partial-drift-condition} and then prove that $\check{Q}$ is
$V$-geometrically ergodic in Section \ref{sec:case-q-0}, where $V$ is unbounded
off petite sets. Then we study the ergodicity in the case $q=1$ in
Section~\ref{sec:case-q-1}. The general case is presented in
Section~\ref{sec:case-q-lager-1}. Finally, we determine the scaling limit of the
point process in \emph{physical time} in Section \ref{sec:scaling-limit}.

\subsection{Irreducibility}
\label{sec:irreduc}

Because of the constraints sets $\mathbf{A}_i$ and the function $\mathbf{J}$, the path of
the process $\mathbf{S}$ cannot evolve arbitrarily. The following definitions
will be useful.

\begin{definition}
\label{def:accessibility}
Let $m$ be a positive integer and $\mathbf{s} \in\mathbb{Z}_+^q$, the set of admissible paths $\mathcal{A}_m (\mathbf{s})$ is defined
as the set of $(j_1, \cdots, j_m)\in \{0,\dots,p\}^m$ such that
\begin{align}
\label{eq:def:accessibility}
 \mathbf{s} \in \mathbf{A}_{j_1}^c, \mathbf{s}+\mathbf{J}_o(j_1) \in \mathbf{A}_{j_2}^c, \cdots, \mathbf{s}+\sum_{n=1}^{m-1} \mathbf{J}_o(j_n) \in
\mathbf{A}_{j_m}^c\;.
\end{align}
\end{definition}
An admissible path $(j_1, \cdots, j_m)\in\mathcal{A}_m (s)$ implies that, for
any $i=0,1,\dots,p$ and $\boldsymbol{\ell}\in(0,\infty)^p$, given
$\mathbf{Y}_0=(i,\mathbf{s},\boldsymbol{\ell})$ the conditional probability to
have $I_1=j_1,\dots,I_m=j_m$ (and thus, by~(\ref{eq:recurrence-of-S})
$\mathbf{S}_k=\mathbf{s}+\sum_{n=1}^{k} \mathbf{J}_o(j_n)$ for all
$k=1,2,\dots,m$) is positive.

The main assumption of this section is the following one.

\begin{assumption}
\label{ass:access}
The fertility matrix $\aleph$ defined in~(\ref{eq:aleph-def}) is invertible.
Moreover, if $q\geq1$, there exists $\mathbf{s}_o \in\mathbb{Z}_+^q$ such that
the following assertion holds. For all integer $K \geq 1$, there exists an
integer $m\geq p+1$ such that, for all $\mathbf{s} \in \{1,\dots, K\}^q$, there
exists an admissible path $(j_1, \cdots, j_m) \in \mathcal{A}_m (\mathbf{s})$
such that $\mathbf{s} +\sum_{n=1}^m \mathbf{J}_o(j_n) = \mathbf{s}_o$ and
$\{j_{m-p+1}, \cdots, j_{m}\} = \{1, \cdots, p\}$ (possibly in a different
order).
\end{assumption}

We say that $C$ is an
$(m,\epsilon,\nu)$-small set for the kernel $Q$ if, for all
$\mathbf{x}\in C$ and $A \in \mathcal{B}(\mathsf{X})$,
$$
Q^m(\mathbf{x},A)\geq\epsilon\,\nu(A)\;,
$$
where $m$ is a positive integer, $\epsilon>0$ and $\nu$ a probability measure.
Following \cite{meyn-tweedie-2009}, the existence of small sets is related to
$\psi$-irreducibility.  The assumption on the invertibility of $\aleph$ will be
useful for proving the $\psi$--irreducibility of the embedded Markov chain.
The second assumption says that there exists an admissible path (see
Definition~\ref{def:accessibility}) with the last $p$ steps containing the $p$
different marks. We will check this condition for the application of the
constrained Hawkes process to a limit order book in
Section~\ref{sec:sclobchp-application}. It will be used to establish the
existence of small sets for $Q$, see Proposition~\ref{prop:general-smallset}.
Most of the time, we will rather refer to the notion of petite sets, which is a
slight extension of small sets. Following \cite{meyn-tweedie-2009}, we say that
$A$ is a petite set for the kernel $Q$ if, for some probability measure $a$ on $\{0,1,2,\dots\}$, $\epsilon>0$ and some
probability measure $\nu$ on $\mathsf{X}$, it is a
$(1,\epsilon,\nu)$-small set for the $a$-sampled chain which has kernel
$\sum_{m\geq0} a(m) Q^m$.

We can now state the main result of this section.
\begin{theorem}
\label{thm:aper-irred}
Let $Q$ and $\tilde{Q}$ be the transition kernels defined in
Section~\ref{sec:markov-assumption} on the spaces
$\mathbb{Z}_+^q\times\mathbb{R}_+^p$ and
$\{0,1,\dots,p\}\times\mathbb{Z}_+^q\times\mathbb{R}_+^p$, respectively, with
$p\geq1$ and $q\geq0$. Suppose moreover that
Assumption~\ref{ass:access} holds.  Then then kernels $Q$ and $\tilde Q$ are
aperiodic and $\psi$-irreducible. Moreover for all $K\geq1$ and $M>0$,
$\{1,\dots,K\}^q\times(0,M]^p$ and
$\{0,\dots,p\}\times\{1,\dots,K\}^q\times(0,M]^p$ are petite sets for $Q$ and
$\tilde Q$, respectively.
\end{theorem}

The proof is postponed to Section~\ref{sec:small-sets}.

\subsection{Partial drift condition}
\label{sec:partial-drift-condition}

We now derive a partial drift condition that will be useful to obtain a
``complete'' drift condition on the process $\{\mathbf{X}(t),\,t\geq0\}$. In
the following result, the drift condition is partial in the sense that it does
not control $\mathbf{S}_n$. It only says that, under Assumption~\ref{ass:Phi},
independently of the process $\mathbf{S}_n$, $\boldsymbol{\lambda}_n$ should not
have large excursions away of a compact set.

\begin{proposition}
\label{prop:drift-step-1}
Let $Q$ be the kernel defined in Section~\ref{sec:markov-assumption}  on the space
$\mathbb{Z}_+^q\times\mathbb{R}_+^p$  with $p\geq1$ and $q\geq0$.
Suppose that Assumption~\ref{ass:Phi} holds. Then there exists $\gamma>0$,
$\theta \in (0,1)$, $M > 0$ and $b > 0$ such that, for all $\mathbf{s} \in \mathbb{Z}_+^q$ and
$\boldsymbol{\ell} \in \mathbb{R}_+^p$, we have
\begin{align}
\label{eq:drift-step-1}
[Q(\1_{\mathbb{Z}_+^q}\otimes V_{1,\gamma})](\mathbf{s},\boldsymbol{\ell})
\leq \theta \; V_{1,\gamma}(\boldsymbol{\ell}) +
b \1_{(0, M]^p}(\boldsymbol{\ell}) \;,
\end{align}
where
\begin{equation}
  \label{eq:V1}
V_{1,\gamma}(\boldsymbol{\ell}) = \mathrm{e}^{\gamma \mathbf{u}^T
  \boldsymbol{\ell}}\;,
\end{equation}
with $\mathbf{u}$ defined as in~(\ref{eq:defu}).
\end{proposition}

The proof is postponed to Section~\ref{sec:proof-partial-drift-condition}.

\subsection{The case $q=0$}
\label{sec:case-q-0}

If $q=0$, the process $\mathbf{N}$ is a standard (unconstrained) multivariate
Hawkes process. In this case the process $\{\boldsymbol{\lambda}_n,\,n\geq1\}$
is a Markov chain (the spread variable $\mathbf{S}_n$ vanishes) and the
``partial'' drift condition of Proposition~\ref{prop:drift-step-1} becomes a
``complete'' drift condition for this chain. Moreover, if $q=0$ in
Theorem~\ref{thm:aper-irred}, Assumption~\ref{ass:access} boils down to
assuming that the fertility matrix $\aleph$ defined in~(\ref{eq:aleph-def}) is
invertible. We also note that the sublevel sets of $V_{1,\gamma}$ are petite
sets by Theorem~\ref{thm:aper-irred}.  Hence we obtain that
$\{\check{\boldsymbol{\lambda}}_n,\,n\geq0\}$ is $V_{1,\gamma}$-geometrically
ergodic, see~\cite[Chapter~15]{meyn-tweedie-2009} (and also
\cite[Theorem~9.1.8]{meyn-tweedie-2009} to get the Harris recurrence).  In fact
the same result holds on the extended chain
$\{(\check{I}_n,\check{\boldsymbol{\lambda}}_n),\,n\geq0\}$, whose transition
kernel has been denoted by $\check{Q}$ in Section~\ref{sec:markov-assumption}.
\begin{proposition}
  \label{prop:embeddedunconstrained-geom-erg}
  Let $\check{Q}$ be the transition kernel defined in
  Section~\ref{sec:markov-assumption} on the space
  $\{0,1,\dots,p\}\times\mathbb{R}_+^p$. Suppose that $\aleph$ is invertible
  and that Assumption~\ref{ass:Phi} holds.  Then there exists $\gamma>0$ such
  that $\check{Q}$ is $(\1_{\{0,\dots,p\}} \otimes V_{1,\gamma})$-geometrically
  ergodic, where $V_{1,\gamma}$ is defined in~(\ref{eq:V1}).
\end{proposition}
\begin{proof}
  As explained above, Theorem~\ref{thm:aper-irred} and
  Proposition~\ref{prop:drift-step-1} apply to the case $q=0$ and,
  as a consequence,
  $\{\check{\boldsymbol{\lambda}}_n,\,n\geq0\}$ is a $V_{1,\gamma}$-geometrically
  ergodic. This conclusion also applies for the kernel $\check{Q}$ of the chain
  $\{(\check{I}_n,\check{\boldsymbol{\lambda}}_n),\,n\geq0\}$ with
  $V_{1,\gamma}$ replaced by $(\1_{\{0,\dots,p\}} \otimes V_{1,\gamma})$
  because the conditional probability of $(\check{I}_1,
  \check{\boldsymbol{\ell}}_1)$ given $(\check{I}_0,
  \check{\boldsymbol{\ell}}_0)$ does not depend on $\check{I}_0$.
\end{proof}
Having proved the ergodicity of $\check{Q}$, we can establish the following
result on its stationary distribution. It will be used to obtain the case
$q=1$.
\begin{corollary}
    \label{cor:moments-checkQvsN}  Under the same assumptions as
  Proposition~\ref{prop:embeddedunconstrained-geom-erg}, denote by
  $\check{\pi}$ the stationary distribution of  $\check{Q}$. Let
  $\mathbf{w}:\{1,\dots,p\}\to\mathbb{R}$. Let $\mathbf{w}_o$ be the extension
  of $\mathbf{w}$ to $\{0,1,\dots,p\}$ obtained by setting
  $\mathbf{w}_o(0)=0$. Then we have
  \begin{equation}
    \label{eq:moments-checkQvsN}
   \check{\pi}(\mathbf{w}_o\otimes\1_{\mathbb{R}_+^p})=\frac{\overrightarrow{\mathbf{w}}^T(\IdMat_p-\aleph)^{-1}\boldsymbol{\mu}_0}{\mu_0^0+\IdVect_p^T(\IdMat_p-\aleph)^{-1}\boldsymbol{\mu}_0}\;.
  \end{equation}
\end{corollary}
The proof is postponed to Section~\ref{sec:proof-coroll-refc}.

\subsection{The case $q=1$}
\label{sec:case-q-1}
We shall further denote
\begin{equation}
  \label{eq:def-sstar}
s^*=1+\max\left(\bigcup_{j=1,\dots,p}A_j\right)\;.
\end{equation}

In the case $q=1$, we need a drift function that applies to
$\mathbf{X}_n=({S}_n,\boldsymbol{\lambda}_n)$ which is unbounded off petite
sets, that is, the drift function must diverge as at least one component of
$\mathbf{X}_n$ goes to infinity, while, in Proposition~\ref{prop:drift-step-1}
the drift function goes to infinity only when one of the components of
$\boldsymbol{\lambda}_n$ goes to infinity. To this end we define
\begin{equation}
  \label{eq:V0}
V_{0,\gamma}(s) = \mathrm{e}^{\gamma s}\;.
\end{equation}
We introduce the following assumption.
\begin{assumption}
\label{ass:qequalsone}
The following inequality holds
\begin{align}
\label{eq:qequalsone}
\overrightarrow{J}^T(\IdMat_p-\aleph)^{-1}\boldsymbol{\mu}_0<0\;.
\end{align}
\end{assumption}
Note that Condition~(\ref{eq:ergodicity-birthdeath}) corresponds to
Assumption~\ref{ass:qequalsone} in the special case $\aleph=0$.

\begin{theorem}
\label{thm:geometrical-ergodicqIs1}
Let $q=1$ and $p\geq1$ and suppose that
Assumption~\ref{ass:Phi},~\ref{ass:access} and~\ref{ass:qequalsone} hold.  Let
$\tilde{Q}$ be the kernel defined in Section~\ref{sec:markov-assumption} on the
space $\{0,1,\dots,p\}\times\mathbb{Z}_+\times\mathbb{R}_+^p$ with
$p\geq1$. Then, for all $\gamma_1>0$ small enough, there exists $\gamma_0^*>0$
such that, for all $\gamma_0\in(0,\gamma_0^*]$, $\tilde{Q}$ is
$(\1_{\{0,\dots,p\}} \otimes V_{0,\gamma_0}\otimes
V_{1,\gamma_1})$-geometrically ergodic, where $V_{0,\gamma_0}$ and
$V_{1,\gamma_1}$ are defined in~(\ref{eq:V0}) and~(\ref{eq:V1}).
\end{theorem}
The proof of Theorem~\ref{thm:geometrical-ergodicqIs1} is omitted as it is a particular case of
Theorem~\ref{thm:geometrical-ergodicqGeQ2}, which is proved in
Section~\ref{sec:induction-q}.

Assumptions~\ref{ass:Phi} and~\ref{ass:access} appear to be very mild.
Assumption~\ref{ass:Phi} is related to the
stability of the underlying unconstrained Hawkes process.
Assumption~\ref{ass:access} is used to obtain small sets
for the chains $Q$ and $\tilde{Q}$.
A natural question is to ask whether Assumption~\ref{ass:qequalsone}
$\overrightarrow{J}^T(\IdMat_p-\aleph)^{-1}\boldsymbol{\mu}_0<0$ is sharp.
The following theorem partially answers to this question.

\begin{theorem}
\label{thm:nonergodic}
Let $\{(S_n, {\boldsymbol{\lambda}}_n),\,n\geq0\}$ be the Markov chain
on the space $\mathbb{Z}_+\times \mathbb{R}_+^p$ with transition kernel $Q$
defined in Section~\ref{sec:markov-assumption}. Suppose that
Assumptions~\ref{ass:Phi} and~\ref{ass:access} hold.  Then $Q$ is
$\psi$-irreducible. Moreover the two following assertions holds.
\begin{enumerate}[label=(\roman*)]
\item\label{item:cas-toujours-recurrent} If $\psi(\{s^*,s^*+1,\dots\}\times(0,\infty)^p)=0$, then $Q$ is $(\1_{\{1,\dots,s^*-1\}}\otimes
V_{1,\gamma_1})$-geometrically ergodic, where $s^*$ is defined by~(\ref{eq:def-sstar}).
\item\label{item:pas-toujours-recurrent}  Otherwise, if
  \begin{equation}
    \label{eq:transiant-cond}
\overrightarrow{J}^T(\IdMat_p-\aleph)^{-1}\boldsymbol{\mu}_0>0\;,
  \end{equation}
then $Q$ is transient.
\end{enumerate}
\end{theorem}

The proof is postponed to Section~\ref{sec:proof-theor-refthm:n}.
\begin{remark} The fact that $Q$ is $\psi$-irreducible follows from
  Theorem~\ref{thm:aper-irred}.  Then, by definition of $\nu$ in the proof of
  Proposition~\ref{prop:general-smallset}, Assumption~\ref{ass:access} implies
  that $\psi(\{s_0\}\times(0,\infty)^p)>0$. Thus if $s_0$ in
  Assumption~\ref{ass:access} satisfies $s_0\geq s^*$, then the
  case~\ref{item:cas-toujours-recurrent} does not happen.
\end{remark}

\subsection{General case}
\label{sec:case-q-lager-1}

To obtain a drift function in the general case $q\geq1$, we shall use the
function $V_{0,\gamma}$ defined by (\ref{eq:V0}) applied multiplicatively to each
component of $\mathbf{S}_n$ and the function $V_{1,\gamma_1}$ defined by (\ref{eq:V1}) applied to intensity $\boldsymbol{\lambda}_n$.

We shall see that the proof of the ergodicity for $q\geq1$ relies on an
induction on $q$, see the details in Section~\ref{sec:induction-q}. For $q=1$
the induction applies under the simple Assumption~\ref{ass:qequalsone} because,
using Corollary~\ref{cor:moments-checkQvsN}, it allows us to compute some
moment under the stationary distribution of the unconstrained chain (with
kernel $\check{Q}$). In general the induction is more involved. We shall need
some additional notation.  For any set $\mathcal{J} \subseteq \{1,\dots,q\}$,
we define ${J}^{\mathcal{J}}:\{1,\dots,p\}\to\mathbb{Z}$ by
\begin{equation}
  \label{eq:defJJ}
{J}^{\mathcal{J}}(i)=
\sum_{j\in\mathcal{J}}\mathbf{J}_j(i)\;.
\end{equation}
We further denote by ${J}^{\mathcal{J}}_o$ the extension of ${J}^{\mathcal{J}}$ on  $\{0,\dots,p\}$ defined by
${J}_o^{\mathcal{J}}(0)=0$.

\begin{theorem}
\label{thm:geometrical-ergodicqGeQ2}
Let $q\geq1$ and $p\geq1$, and suppose that
Assumption~\ref{ass:Phi} and~\ref{ass:access} hold.  Let $\tilde{Q}$ be the kernel
defined in Section~\ref{sec:markov-assumption} on the space
$\{0,1,\dots,p\}\times\mathbb{Z}_+^q\times\mathbb{R}_+^p$. Define, for all
$\mathcal{J} \subset \{1, \dots, q\}$, the kernel $\tilde{Q}^{(+\mathcal{J})}$
as in Section~\ref{sec:markov-assumption}.
Check the following conditions in this order~:\\
\For{$k=1,2,3,\dots,q$} {\For{$\mathcal{J} \subset \{1, \dots, q\}$ such that
    $\#\mathcal{J} = k$} {Check that
\begin{align}
\label{eq:condition-geo-ergo-induction-emptycase}
{\left({\overrightarrow{J}}^{\mathcal{J}}\right)}^T(\IdMat_p-\aleph)^{-1}\boldsymbol{\mu}_0<0\;.
\end{align}
\For{$\mathcal{J}' \subset\mathcal{J}$ with $\mathcal{J}' \notin\{\emptyset,\mathcal{J}\}$}
{Check that
\begin{align}
\label{eq:condition-geo-ergo-induction}
\tilde{\pi}^{(+\mathcal{J}')}\left[{J_o^{\mathcal{J}}}\otimes\1_{\mathbb{Z}_+^q\times\mathbb{R}_+^p}\right]<0\;.
\end{align}
}
Then, for all $\gamma_1>0$ small enough, there exists $\gamma_0^*>0$  such that
for all $\gamma_0\in(0,\gamma_0^*]$, $\tilde{Q}^{(+\mathcal{J})}$ is
$(\1_{\{0,\dots,p\}} \otimes V_{0,\gamma_0}^{\otimes (\#\mathcal{J})}\otimes
V_{1,\gamma_1})$-geometrically ergodic. We denote by
$\tilde{\pi}^{(+\mathcal{J})}$ its stationary distribution.
}
}
Then, for all $\gamma_1>0$ small enough, there exists $\gamma_0^*>0$ such that
for all $\gamma_0\in(0,\gamma_0^*]$,  $\tilde{Q}$ is
$(\1_{\{0,\dots,p\}} \otimes V_{0,\gamma_0}^q\otimes
V_{1,\gamma_1})$-geometrically ergodic.
\end{theorem}
The proof is postponed to the end of Section~\ref{sec:induction-q}.
\begin{remark}
By Corollary~\ref{cor:moments-checkQvsN},
Condition~(\ref{eq:condition-geo-ergo-induction-emptycase}) is equivalent
to~(\ref{eq:condition-geo-ergo-induction}) with $\mathcal{J}'=\emptyset$.
\end{remark}
\begin{remark}
  Observe that Condition~(\ref{eq:condition-geo-ergo-induction}) makes sense at this
  step. Indeed, since $\#\mathcal{J} = k$, and $\#\mathcal{J}' < k$, we have
  already checked that $\tilde{Q}^{(+\mathcal{J}')}$ is geometrically
  ergodic. Thus the stationary distribution $\tilde{\pi}^{(+\mathcal{J}')}$ is
  well defined. In practice, one may check this assumption using Monte Carlo
  simulations, since the law of large number holds in this case
  (see~\cite[Theorem~17.1.7]{meyn-tweedie-2009}).
\end{remark}

\subsection{Scaling limit}
\label{sec:scaling-limit}

In the application of this model to the limit order book, an increment of the best bid price
or the best ask price at time $T_k$ is specified by the mark $I_k\in\{0, \dots,
p\}$, see the details in Section~\ref{sec:sclobchp-application}.
It is interesting to investigate the microscopic behavior of the mid-price (mean of the best bid price and the best ask price) at
large scales, see~\cite{bacry-delattre-hoffmann-muzy-2010} for the non-constrained case.

Suppose that the Assumptions of Theorem~\ref{thm:geometrical-ergodicqGeQ2}
hold, then $\tilde{Q}$ is $(\1_{\{0,\dots,p\}} \otimes
V_{0,\gamma_0}^{q}\otimes V_{1,\gamma_1})$-geometrically ergodic for some
$\gamma_0, \gamma_1 > 0$ and $V_{0,\gamma_0}$ and $V_{1,\gamma_1}$ defined
in~(\ref{eq:V0}) and~(\ref{eq:V1}).  Applying the results of
\cite[Theorem~17.4.2]{meyn-tweedie-2009}, we obtain a functional central limit
theorem (FCLT) in \emph{discrete time} for the chain $\tilde{Q}$. To obtain a
result in (physical) \emph{continuous time}, we actually need a functional CLT for the kernel
$\overline{Q}$ defined in Section~\ref{sec:markov-assumption}. Hence we must
first check that $\overline{Q}$ is geometrically ergodic.

Let us define $V_{2,\gamma} : \mathbb{R}_+ \to \mathbb{R}_+$ by
\begin{equation}
  \label{eq:V2}
V_{2,\gamma}(\delta) = \mathrm{e}^{\gamma \delta}\;.
\end{equation}

We provide the ergodicity property of $\overline{Q}$ in the following corollary.
\begin{corollary}
\label{cor:geo-ergo-overlineQ}
Let $q,p\geq1$ and $\overline{Q}$ be the kernel defined in
Section~\ref{sec:markov-assumption} on the space
$\mathbb{R}_+\times\{0,1,\dots,p\}\times\mathbb{Z}_+^q\times\mathbb{R}_+^p$. Suppose
that the Assumptions of Theorem~\ref{thm:geometrical-ergodicqGeQ2} hold. Then, for
all $\gamma_1>0$ small enough, there exists $\gamma_0^*>0$ such that for all
$\gamma_0\in(0,\gamma_0^*]$ and all $\gamma_2>0$ small enough, $\overline{Q}$
is $(V_{2,\gamma_2} \otimes \1_{\{0,\dots,p\}} \otimes
V_{0,\gamma_0}^{q}\otimes V_{1,\gamma_1})$-geometrically ergodic.
\end{corollary}
The proof is postponed to Section~\ref{sec:proof-coroll-refc-1}.

By Corollary~\ref{cor:geo-ergo-overlineQ}, $\overline{Q}$ admits a stationary
distribution $\overline{\pi}$. We shall use the notation
$\overline{\mathbb{E}}$ for the expectation under the stationary distribution.

Let now $g:\mathsf{Z} \to \mathbb{R}$ be such that for all $\gamma_0, \gamma_1,
\gamma_2 > 0$,
\begin{align}
\label{eq:g-sub-expo}
\sup_{\mathbf{z} \in \mathsf{Z}} {\frac{g(\mathbf{z})}{[V_{2, \gamma_2} \otimes \1_{\{0,\dots,p\}} \otimes V_{0,\gamma_0}^{q} \otimes
V_{1,\gamma_1}](\mathbf{z})}} < \infty \;.
\end{align}

By Corollary~\ref{cor:geo-ergo-overlineQ}, under this condition, we have
$\overline{\pi}(|g|) < \infty$.  Let $\bar{g} = g - \overline{\pi}(g)$. Using
Corollary~\ref{cor:geo-ergo-overlineQ} and
\cite[Theorem~17.4.4]{meyn-tweedie-2009}, for all $\gamma_0, \gamma_1, \gamma_2
> 0$, there exists $R > 0$ such that the Poisson equation $\hat{g}-\overline{Q}(\hat{g})= \bar{g}$
(see \cite[Chapiter~17]{meyn-tweedie-2009}) admits a solution $\hat{g}$ satisfying
the bound $|\hat{g}| \leq R\left((V_{2,\gamma_2} \otimes \1_{\{0,\dots,p\}}
  \otimes V_{0,\gamma_0}^{q}\otimes V_{1,\gamma_1}) + 1\right)$. Moreover, we have, for all $\mathbf{z} \in \mathsf{Z}$,
$$
\hat{g}(\mathbf{z}) =\sum_{k=0}^{\infty} \overline{Q}^k
(\mathbf{z},\bar{g})\;.
$$
It also follows that $\overline{\pi}(|\hat{g}|^2) < \infty$.
Let us denote the linearly interpolated partial sums of $g(\mathbf{Z}_n)$ by
\begin{equation*}
s_n (t, g) = \Sigma_{\lfloor nt\rfloor} (\bar{g}) + (nt - \lfloor nt\rfloor) \left[\Sigma_{\lfloor nt\rfloor + 1} (\bar{g}) - \Sigma_{\lfloor nt\rfloor} (\bar{g}) \right]\;,
\end{equation*}
where $\Sigma_k(h)=\sum_{j=1}^kh(\mathbf{Z}_j)$. Applying \cite[Theorem~17.4.4]{meyn-tweedie-2009}, if the (nonnegative) constant
\begin{equation}
\label{eq:def-sigma-g}
\sigma_{g}^2 := \overline{\pi}\left(\hat{g}^2 - \{\overline{Q} \hat{g}\}^2\right)
\end{equation}
is strictly positive, we have, for any initial distribution, as $n \to \infty$,
\begin{align}
\label{eq:fclt}
(n\sigma_g^2)^{-1/2} s_n (t, g) \overset{d}{\to} B_t \;,
\end{align}
where $\overset{d}{\to}$ denotes the weak convergence. Here, the weak convergence holds in
$C([0,1])$ (the space of continuous functions defined on $[0,1]$) and $B$
denotes a standard Brownian motion on $[0, 1]$.

\begin{remark}\label{rem:spread-notscaling}
Note that  $\sigma_{g}^2$ may not be strictly positive. An interesting example
is given by $g = \1_{\mathbb{R}_+}\otimes\mathbf{J}_o \otimes
\1_{\mathbb{Z}_+^q} \otimes \1_{\mathbb{R}_+^p}$. In this case, we have
$\Sigma_n ({g}) = \mathbf{S}_n-\mathbf{S}_0$ and thus
$$
\Sigma_n (\overline{g}) = \mathbf{S}_n-\mathbf{S}_0 - \overline{\mathbb{E}}(\mathbf{S}_n) + \overline{\mathbb{E}}(\mathbf{S}_0) + \sum_{k=1}^{n} {\overline{\mathbb{E}}(\mathbf{J}_o(I_k)) - \overline{\pi}(g)} \;,
$$
Hence, we get that $\Sigma_n (\overline{g})  = 0_{P_*}(1)$, which implies $\sigma_g^2 = 0$.
\end{remark}

We now derive the main result of this section, which determines the scaling
limit of the constrained Hawkes process. We need some additional notation
before stating the result.  Let $\mathbf{w}:\{1,\dots,p\}\to\mathbb{R}$.
Define
\begin{equation}
\label{eq:mathcalE-w}
\mathcal{E}(\mathbf{w})=\frac{\overline{\mathbb{E}}[\mathbf{w}(I_1)]}{\overline{\mathbb{E}}[\Delta_1]}\;.
\end{equation}
Further define
\begin{equation}
\label{eq:v-w}
v(\mathbf{w})=\sigma_{g}^2\;,
\end{equation}
where  $\sigma_{g}^2$ is defined as in~(\ref{eq:def-sigma-g}) with
$g:\mathsf{Z} \to \mathbb{R}$ defined by $g(\mathbf{z}) = \mathbf{w}(i) -
\mathcal{E}(\mathbf{w}) \delta$  for all $\mathbf{z} = (\delta, i, \mathbf{s}, \boldsymbol{\ell})$.
\begin{theorem}
\label{thm:fclt-cont-time}
Under the assumptions of Theorem~\ref{thm:geometrical-ergodicqGeQ2}, if $v(\mathbf{w})>0$, we have,
for any initial condition, as $T\to\infty$,
\begin{equation}
  \label{eq:fclt-realtime}
T^{-1/2}\left(\mathbf{N}(\1_{[0,tT]}\otimes\mathbf{w})-t\; T\;\mathcal{E}(\mathbf{w})\right)\overset{d}{\to}
(v(\mathbf{w})/\overline{\mathbb{E}}[\Delta_1])^{1/2}\, B_t \;,
\end{equation}
where the weak convergence holds in $D([0,1])$ (the space of càdlàg functions
defined on $[0,1]$).
\end{theorem}
The proof is postponed to Section~\ref{sec:proof-theor-refeq:fc}.

\section{Application to a limit order book}
\label{sec:sclobchp-application}

Let us apply our results to the simple limit order book presented in the
introduction. The function $J$ and the sets $A_i$ of the corresponding
constrained multivariate Hawkes process are detailed in
Section~\ref{sec:sclobchp-application-desc}.  The Markov assumption of
Section~\ref{sec:markov-assumption} allows us to use a finite set of parameters
for this model, namely $\beta,\aleph$ and $\boldsymbol{\mu}_0$.

Assumptions~\ref{ass:Phi} and~\ref{ass:qequalsone}
then amount to assumptions on the parameters  $\boldsymbol{\mu}_0$ and $\aleph$.
Assumption~\ref{ass:access} says that $\aleph$ has to be invertible and
requires an additional property which only depends on the constrained sets
$A_i$ and $J$. Hence, to apply Theorem~\ref{thm:geometrical-ergodicqIs1},
we only need to check that this additional property holds.
We set $s_o=2$. For any $K\geq1$, we set $m=|K-2|+p$. Then given
$s\in\{1,\dots,K\}$ we need to show that there is an admissible path
$(j_1,\dots,j_m)\in \mathcal{A}_m (s)$ such that
$s+\sum_{k=1}^m J_o (j_k)=2(=s_o)$. To see why, let us set, for $1\leq
i\leq |s-2|$,
$$
j_i=
\begin{cases}
  2&\text{ if $s\geq2$,}\\
  1&\text{ of $s=1$.}
\end{cases}
$$
Then, by~(\ref{eq:sclob-J}), we have
$s+\sum_{k=1}^{|s-2|} J(j_k)=2$. To conclude, we set $j_i=0$ for
$|s-2|<i\leq m-4$ and $j_{m-3}=1$, $j_{m-2}=4$, $j_{m-1}=2$ and $j_{m}=3$, so
that, by~(\ref{eq:sclob-J}) and~(\ref{eq:conv-Jzero}),
$s+\sum_{k=1}^{m} J_o (j_k)=2$. Moreover, we easily check
that~(\ref{eq:def:accessibility}) holds for this choice of $(j_1,\dots,j_m)$
and that $\{j_{m-3},j_{m-2},j_{m-1},j_m\}=\{1,\dots,4\}$. Hence
Assumption~\ref{ass:access} holds and  Theorem~\ref{thm:geometrical-ergodicqIs1}
applies provided that $\aleph$ is invertible and has spectral radius smaller than
1, and $\boldsymbol{\mu}_0$ satisfies~(\ref{eq:qequalsone}) with $J$
defined by (\ref{eq:sclob-J}). The ergodicity of the underlying chain allows one
to perform meaningful statistical analysis of the data. This will be done in a
forthcoming paper.

For the moment, we focus on the scaling limit of the mid-price, defined as the middle price between the Best Bid and Best Ask prices. This value is often considered as the continuous
time price of the asset. In the framework of our model, the mid price satisfies
the following equation
\begin{align}\label{eq:mid-evolution}
P(t) = P(u) + \mathbf{N}\left(\1_{(u, t]}\otimes w\right)\;,
\end{align}
where $w(i)$ takes values $1/2, -1/2, 1/2$ and  $-1/2$ for $i=1,\dots,4$,
respectively. (Note that similar equations hold for the Best Bid and Best Ask
prices with different functions $w$s.)

Applying Theorem~\ref{thm:fclt-cont-time}, we get the scaling limit of the
mid-price in physical time is given by
\begin{equation*}
T^{-1/2}\left(P(tT)-P(0)-t\; T\; \mathcal{E}\left(w\right)\right)\overset{d}{\to}
\left(v(w\right)/\overline{\mathbb{E}}[\Delta_1])^{1/2}\, B_t \;,
\end{equation*}
where $\mathcal{E}\left(w \right)$ and $v\left(w \right)$ are respectively
defined in~(\ref{eq:mathcalE-w}) and~(\ref{eq:v-w}), provided that $v\left(w
\right)>0$. Although it does not seem easy to check that $v\left(w
\right)>0$ (except perhaps by numerical means), we do expect this to
be true. Indeed, while the spread behaves as a stationary variable,
yielding a vanishing asymptotic variance in the large scale behavior (see
Remark~\ref{rem:spread-notscaling}), the best-bid, best-ask and mid prices
behave as co-integrated random walks.

In practice, we can apply our results in several ways. We refer to the
forthcoming thesis \cite{zheng-2013} for details about the estimation methods
related to constrained Hawkes processes and their applications to limit order
books data. Let us here briefly illustrate an application on a real data
example. Using all intraday data from $1^{st}$ April 2011 to $10^{th}$ April
2011, parameters estimations for Eni SpA and Total yield the following
results. The decaying rate $\beta$ is estimated to $1.65$ for Eni SpA and to
$1.79$ for Total, which correspond to half-times respectively equal to $0.6$
and $0.55$ seconds. The spectral radii of the estimated fertility matrices
$\aleph$ are equal to $0.5723$ and $0.6221$ for Eni SpA and Total,
respectively, so Assumption \ref{ass:Phi} holds for the estimated
parameters. We also computed the corresponding constants appearing in the
left-hand side of the drift condition~(\ref{eq:qequalsone}) and obtained
$-0.0115$ and $-0.0682$ respectively for Eni SpA and Total.  We note that for
these estimated values, Assumption \ref{ass:qequalsone} holds and thus Theorem
\ref{thm:geometrical-ergodicqIs1}
applies. Following \cite{bacry-delattre-hoffmann-muzy-2010}, to illustrate the relevance of our
model, we shall briefly present empirical signature plots and compare them with
the signature plots of the fitted models. For a given asset price (best bid,
best ask or mid-price) at time $t$ $P(t)$, the empirical signature plot is
defined as an empirical estimate of the quadratic variation of $P(t)$ over $[0,
T]$ expressed as a function of the scale $\tau>0$, namely
\begin{align*}
C(\tau) = \frac{1}{T}\sum_{n=1}^{T/\tau} |{P}(n\tau) - {P}((n-1)\tau)|^2 \;.
\end{align*}
This can be compared with a theoretical signature plot based on the fitted
model, namely
\begin{align*}
\tilde C(\tau) = \frac{1}{\tau} \mathbb{E}_{\hat\theta}\left[\left(P(t+\tau) - P(t)\right)^2\right] \;,
\end{align*}
where $\mathbb{E}_\theta$ here denotes the expectation of the model with
parameter $\theta$ and $\hat\theta$ is the parameter estimated from the
data. In practice this expectation is computed via Monte Carlo simulations of
the fitted model.

In contrast with \cite{bacry-delattre-hoffmann-muzy-2010} where only the
mid-price is considered and the signature plot is always decreasing, we observe
in Figure \ref{fig:SignaturePlot_RealData_Example} that the empirical signature
plot differs depending whether the considered price is the best bid, best ask
or mid-price. In Figure \ref{fig:SP_WithConstraint_p_4_q_1}, we present the
corresponding theoretical signature plots using the fitted models. We observe
that, with a limited number of parameters, the model is able to recover shapes somewhat similar to the empirical ones.

\begin{figure}[h!]
\begin{center}
\begin{tabular}{cc}
\includegraphics[trim=0mm 0mm 0mm 0mm, clip, height=6cm, width=6cm, angle=-90]{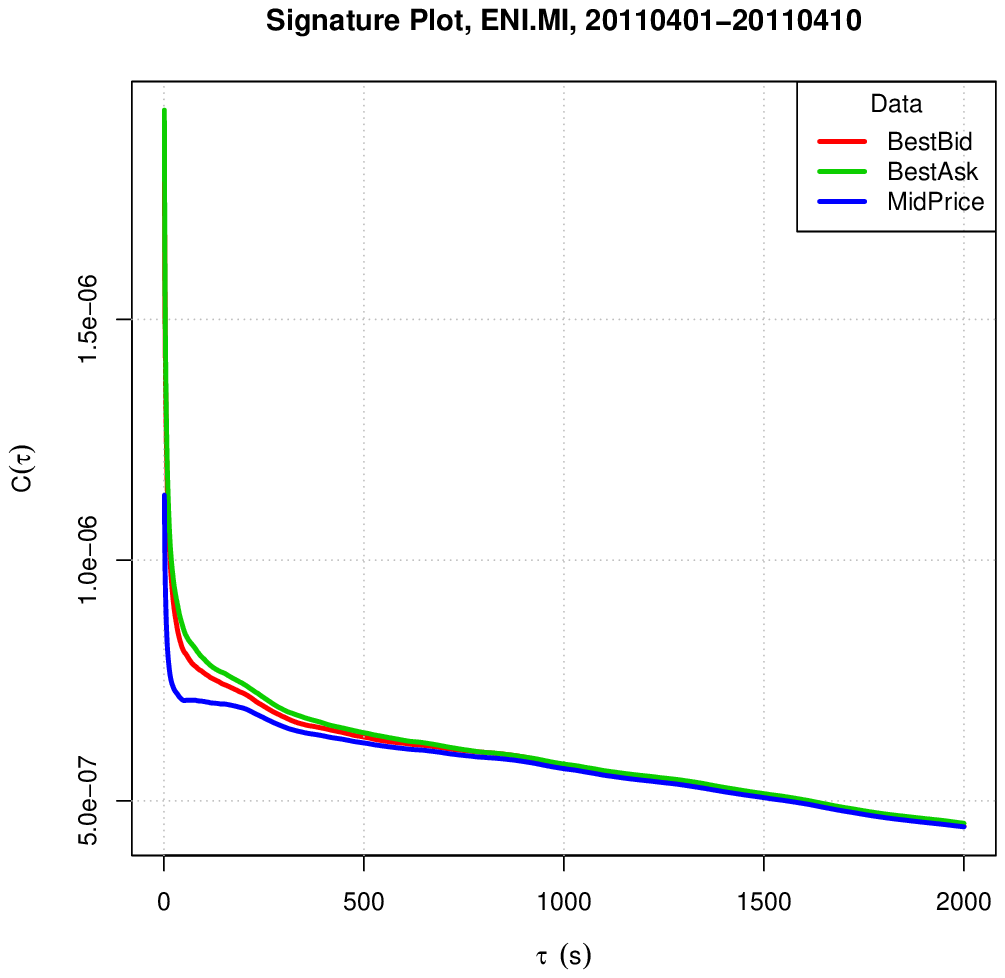}&
\includegraphics[trim=0mm 0mm 0mm 0mm, clip, height=6cm, width=6cm, angle=-90]{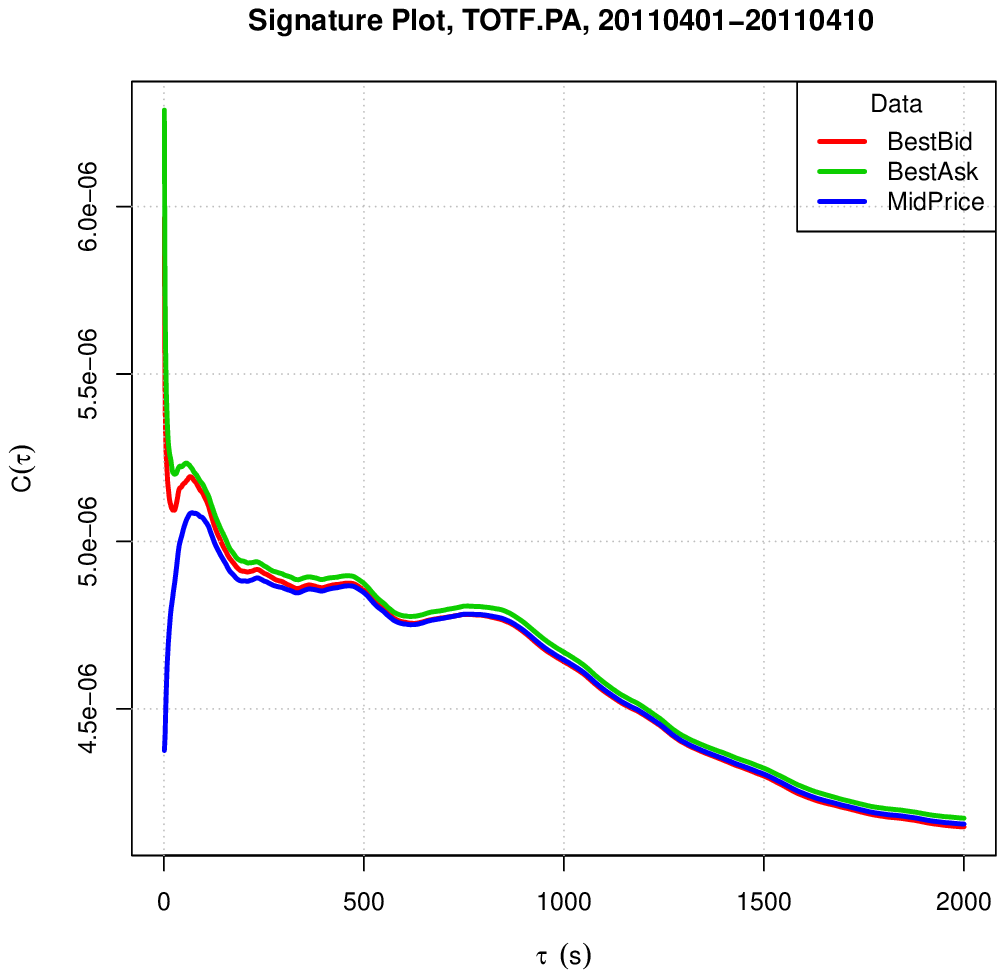}
\end{tabular}
\end{center}
\caption{Signature Plot of Eni SpA (left) and Total (right) calculated on the data of best bid, best ask and mid-price over the period from $1^{st}$ April 2011 to $10^{th}$ April 2011.}
\label{fig:SignaturePlot_RealData_Example}
\end{figure}

\begin{figure}[h!]
\begin{center}
\begin{tabular}{cc}
\includegraphics[trim=0mm 0mm 0mm 0mm, clip, height=6cm, width=6cm, angle=-90]{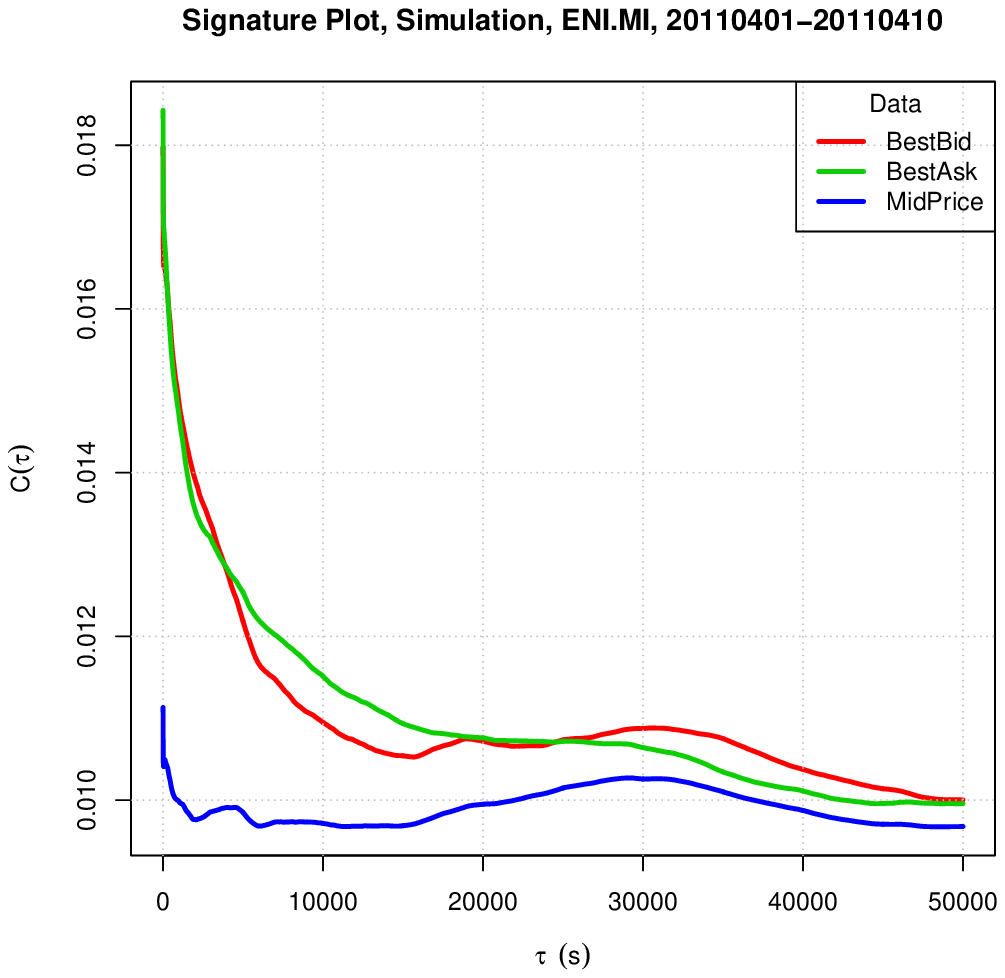}&
\includegraphics[trim=0mm 0mm 0mm 0mm, clip, height=6cm, width=6cm, angle=-90]{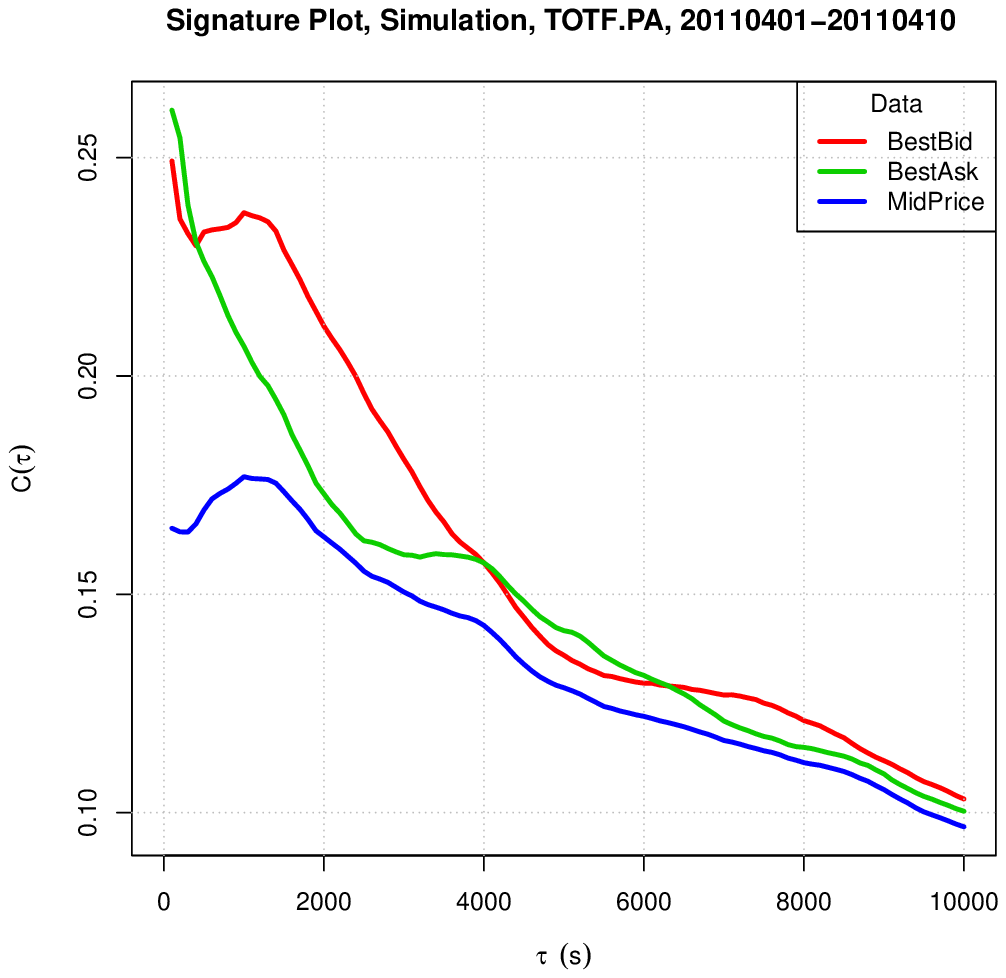}
\end{tabular}
\end{center}
\caption{Signature plots calculated with simulated data where the parameters are calibrated from the mid-quotes data of Eni SpA (left) and  Total (right) with $p=4$ and $q=1$.}
\label{fig:SP_WithConstraint_p_4_q_1}
\end{figure}

\section{Postponed proofs}
\label{sec:proofs}

\subsection{Proof of Theorem \ref{thm:aper-irred}} 
\label{sec:small-sets}

In the following we shall denote
\begin{gather}
  \label{eq:notation-majorants-minorants-mu}
\overline{\mu_0} = \max_j{\boldsymbol{\mu}_0(j)},\;\;
\underline{\mu_0} = \min_{j}{\boldsymbol{\mu}_0(j)},\\
  \label{eq:notation-majorants-minorants-alpha}
\overline{\alpha} =
\max_{i,j}{\alpha_{i, j}}, \underline{\alpha} = \min_{i,j}{\alpha_{i,
    j}}\;,\\
  \label{eq:notation-majorants-minorants-J}
\overline{J}=\max_{i=1,\dots,p}\sum_{j=1}^q|\overrightarrow{\mathbf{J}}_{i,j}|\;,
\end{gather}
and the $j$-th column vector of $\aleph$ by $\aleph_j$, so that
$\aleph=\begin{pmatrix}
\aleph_1&\dots&\aleph_{p}
\end{pmatrix}$.
To establish the existence of small sets, we need two preliminary results.
\begin{lemma}
\label{lem:lower-bound-small-set-step-1}
Let $\mathbf{x}=(\mathbf{s},\boldsymbol{\ell})\in\mathsf{X}$,
 $j\in \mathcal{A}_1 (\mathbf{s})$ with $j\neq0$ and $g: \mathsf{X}\rightarrow{\mathbb{R}_{+}}$. Then
  \begin{align}
  \label{eq:Eq1}
  \mathbb{E}^{\mathbf{x}}\left[g(\boldsymbol{\lambda}_1){\1}_{\{I_{1} = j\}}\right]
  & \geq {\underline{\mu_0}
  \int_{0}^{+\infty}g(\rme^{-t\beta}\boldsymbol{\ell}+\beta\aleph_j)\rme^{-2(\overline{\mu_0}
      + \overline{\ell})t} \mathrm{d}t} \;,
  \end{align} where
  $\overline{\ell} = \max_{i}{\boldsymbol{\ell}(i)}$.
\end{lemma}

\begin{proof}
Using (\ref{eq:intensity-dynamic}), we get that
$$
\mathbb{E}^{\mathbf{x}}\left[g(\boldsymbol{\lambda}_1){\1}_{\{I_{1} = j\}}\right] =
\mathbb{E}^{\mathbf{x}}\left[g(\rme^{-\Delta_{1}^{j}\beta}\boldsymbol{\ell}+\beta\aleph_j){\1}_{\{\Delta_{1}^{j} < U_{1}^{j}\}}\right] \;,
$$
where $U_{1}^{j} = \min_{i\neq{j}}\Delta_{1}^{i}$.

Observe that $U_{1}^{j} \overset{st}{\geq}
\mathrm{Exp}\left(\overline{\mu_0}+\overline{\ell}\right)$ where
$\mathrm{Exp}(z)$ denotes the exponential distribution with mean $\frac{1}{z}$,
thus, since $\Delta_{1}^{j}$ and $U_{1}^{j}$ are  independent, we have
$$
\mathbb{E}^{\mathbf{x}}\left[g(\boldsymbol{\lambda}_1){\1}_{\{I_{1} = j\}}\right]
\geq
\mathbb{E}^{\mathbf{x}}\left[g(\rme^{-\Delta_{1}^{j}\beta}\boldsymbol{\ell}+\beta\aleph_j)\rme^{-(\overline{\mu_0}+\overline{\ell})\Delta_{1}^{j}}\right] \;.
$$
Since $j \in \mathcal{A}_1 (\mathbf{s})$, using that the hazard rate of
$\Delta_{1}^{j}$ is bounded between $\underline{\mu_0}$ and $\overline{\mu_0} +
\overline{\ell}$, we see that the density of $\Delta_{1}^{j}$ is bounded from
below by $\underline{\mu_0}e^{-(\overline{\mu_0} + \overline{\ell})t}$ on
$t\in{\mathbb{R}_{+}}$.  Hence we get (\ref{eq:Eq1}).
\end{proof}

\begin{corollary}
\label{cor:lower-bound-small-set-step-m}
Let $m$ be an integer $\geq1$, $M \in \mathbb{R}_+$ and $g:
\mathsf{X}\rightarrow{\mathbb{R}_{+}}$. Let $\mathbf{x}=(\mathbf{s},\boldsymbol{\ell})\in\mathsf{X}$ with
$\boldsymbol{\ell} \in (0, M]^p$ and $(j_1, \cdots, j_m)\in
\mathcal{A}_m (\mathbf{s})$ with $j_i\neq0$ for $i=1,\dots,m$. Then, we have
\begin{multline*}
 \mathbb{E}^{\mathbf{x}}\left[g(\boldsymbol{\lambda}_m){\1}_{\{I_1 = j_1,
     \cdots, I_m = j_m\}}\right] \\
\geq  {\underline{\mu_0}}^{m}\int_{\mathbb{R}_{+}^m} g[e^{-\sum_{i=1}^{m}{t_i}\beta}\boldsymbol{\ell}+e^{-\sum_{i=2}^{m}{t_i}\beta} \beta\aleph_{j_1}
+\cdots+\beta\aleph_{j_m}] \\
 \times e^{- C_{m,M} \sum_{i=1}^m t_i} \; \mathrm{d}t_m\cdots \mathrm{d}t_1 \;,
\end{multline*}
where $C_{m,M} = 2(\overline{\mu_0}+ M + (m-1)\beta\overline{\alpha})$.
\end{corollary}

\begin{proof}
The case $m=1$ is given by
Lemma~\ref{lem:lower-bound-small-set-step-1}.

To get the result for any $m\geq2$, we proceed by induction. Let
$\mathbf{x}=(\mathbf{s},\boldsymbol{\ell})\in\mathsf{X}$ with
$\boldsymbol{\ell} \in (0, M]^p$ and $(j_1, \cdots, j_m)\in \{1,\dots,p\}^m$,
if $(j_1, \dots, j_m) \in \mathcal{A}_m (\mathbf{s})$. We have
\begin{multline}\label{eq:small-set-conditonning-step1}
\mathbb{E}^{\mathbf{x}}\left[g(\boldsymbol{\lambda}_m)
{\1}_{\{I_1 = j_1, \cdots, I_m = j_m\}}\right] \\=
\mathbb{E}^{\mathbf{x}}\left[{\1}_{\{I_1 = j_1\}}
\mathbb{E}\left[g(\boldsymbol{\lambda}_m)
{\1}_{\{I_2 = j_2, \cdots, I_m = j_m\}}|(I_1,\mathbf{X}_1)\right]\right]\;.
\end{multline}
Recall that given $(I_1,\mathbf{X}_1)$, the distribution of
$\{(I_k,\mathbf{X}_k),\,k\geq2\}$ does not depend on $I_1$, hence, with the
Markov property, we obtain that
\begin{align*}
  \mathbb{E}\left[g(\boldsymbol{\lambda}_m)
{\1}_{\{I_2 = j_2, \cdots, I_m = j_m\}}|(I_1,\mathbf{X}_1)\right]=
\mathbb{E}^{\mathbf{X}_1}\left[g(\boldsymbol{\lambda}_{m-1})
{\1}_{\{I_1 = j_2, \cdots, I_{m-1} = j_m\}}\right]\;.
\end{align*}
Note that, using
(\ref{eq:intensity-dynamic}) and~(\ref{eq:notation-majorants-minorants-alpha}),
$\{\mathbf{X}_0=\mathbf{x}\}$ and $\boldsymbol{\ell}\in(0,M]^p$ implies
$\boldsymbol{\lambda}_1\in(0,M+\beta\overline{\alpha}]^p$.  Thus,  in the event
$\{\mathbf{X}_0=\mathbf{x}\}$, the right-hand side of the last display can be
bounded from below by applying the induction hypothesis with $(j_1, \cdots, j_{m-1})$
replaced by $(j_2, \cdots, j_m)$, $\boldsymbol{\ell}=\boldsymbol{\lambda}_1$
and the constant
$$
C_{m-1,M+\beta\overline{\alpha}}= 2(\overline{\mu_0}+ M + \beta\overline{\alpha} +
(m-2)\beta\overline{\alpha})
=C_{m,M}\;.
$$
Plugging this bound in~(\ref{eq:small-set-conditonning-step1}),
we get that
\begin{multline}
\mathbb{E}^{\mathbf{x}}\left[g(\boldsymbol{\lambda}_m)
{\1}_{\{I_1 = j_1, \cdots, I_m = j_m\}}\right] \\
\label{eq:small-set-m-step-1}
\geq
{\underline{\mu_0}}^{m-1}\mathbb{E}^{\mathbf{x}}\left[{\1}_{\{I_1 = j_1\}}
  \int_{\mathbb{R}_{+}^{m-1}} g\left[\dots\right]
 \times e^{- C_{m,M} \sum_{i=1}^{m-1} t_i} \; \mathrm{d}t_{m-1}\cdots \mathrm{d}t_1
\right] \;,
\end{multline}
where the omitted argument of $g$ reads
$$
[\dots]=e^{-\sum_{i=1}^{m-1}{t_i}\beta}\boldsymbol{\lambda}_1+e^{-\sum_{i=2}^{m-1}{t_i}\beta} \beta\aleph_{j_2}
+\cdots+\beta\aleph_{j_{m}} \;.
$$
Using again Lemma~\ref{lem:lower-bound-small-set-step-1}, we have
\begin{align}
\mathbb{E}^{\mathbf{x}}\left[{\1}_{\{I_1 = j_1\}} g\left[\dots\right] \right] &
\geq\underline{\mu_0} \int_{0}^{+\infty}{g\left[\dots^\prime\right]\rme^{-2(\overline{\mu_0} + \overline{\ell})t} \mathrm{d}t}
\nonumber \\
\label{eq:small-set-m-step-2}
& \geq \underline{\mu_0}\int_{0}^{+\infty}{g\left[\dots^\prime\right]\rme^{-C_{m, M} t} \mathrm{d}t}
\;,
\end{align}
where the omitted argument in the right-hand side reads
\begin{align*}
[\dots^\prime] & = e^{-\sum_{i=1}^{m-1}{t_i}\beta}(e^{-{t}\beta} \boldsymbol{\ell} +
\beta\aleph_{j_1} )+e^{-\sum_{i=2}^{m-1}{t_i}\beta} \beta\aleph_{j_2} + \cdots+\beta\aleph_{j_{m}}
\\
& = e^{-\sum_{i=1}^{m-1}{t_i}\beta-t\beta}\boldsymbol{\ell}+e^{-\sum_{i=1}^{m-1}{t_i}\beta} \beta\aleph_{j_1}
+\cdots+\beta\aleph_{j_m} \;.
\end{align*}
Hence we get the result by applying the Tonelli theorem in
(\ref{eq:small-set-m-step-1}) and then using (\ref{eq:small-set-m-step-2}).
\end{proof}

\begin{proposition}
  \label{prop:general-smallset}
  Let $Q$ be the transition kernel defined in Section~\ref{sec:markov-assumption}
  on the space  $\mathbb{Z}_+^q\times\mathbb{R}_+^p$ with $p,q\geq1$. Suppose
  that Assumption~\ref{ass:access} holds.  Then there exists a probability
  measure $\nu$ on $\mathsf{X}$ such that, for all $K \geq1$ and $M>0$, there
  exists $\epsilon>0$ and a positive integer $m$ such that $\{1, \dots, K\}^q
  \times (0, M]^p$ is an $(m,\epsilon,\nu)$ and $(m+1,\epsilon,\nu)$-small set
  for $Q$.
\end{proposition}

Before providing the proof of this result, let us state the following corollary
which follows by observing that the conditional distribution of $({I}_1,
\mathbf{S}_{1},\boldsymbol{\lambda}_{1})$ given $({I}_0,
\mathbf{S}_{0},\boldsymbol{\lambda}_{0})$ does not depend on $I_0$.

\begin{corollary}
\label{cor:smallset-Qtilde}
Let $\tilde{Q}$ be the transition kernel defined in
Section~\ref{sec:markov-assumption} on the space
$\{0,1,\dots,p\}\times\mathbb{Z}_+^q\times\mathbb{R}_+^p$ with $p,q\geq1$.
Suppose that Assumption~\ref{ass:access} holds. Then there exists a probability
measure $\tilde{\nu}$ on $\mathsf{X}$ such that, for all $K \geq 1$ and $M >
0$, there exists $\epsilon>0$ and a positive integer $m$ such that $\{1,\dots,
p\} \times \{1, \dots, K\}^q \times (0, M]^p$ is an $(m,\epsilon,\tilde{\nu})$
and $(m+1,\epsilon,\tilde{\nu})$-small set for $\tilde{Q}$.
\end{corollary}

\begin{proof}[Proof of Proposition~\ref{prop:general-smallset}]
  Let $K\geq1$ and $M>0$ and denote $C = \{1, \dots, K\}^q \times (0, M]^p$.
  Let $\mathbf{x} = (\mathbf{s}, \boldsymbol{\ell})\in C$.  By
  Assumption~\ref{ass:access}, there exists an integer $m>0$, and an admissible
  path $(j_{1}, \cdots,j_{m})$ such that $\mathbf{s} + \sum_{n=1}^m
  \mathbf{J}(j_n) = \mathbf{s}_{o}$ and $\{j_{m-p+1},\cdots,j_{m}\} = \{1,
  \cdots, p\}$, where $\mathbf{s}_{o} \in \mathbb{Z}_+^q$ is independent of
  $m$, $M$ and $K$.

Now, for all subsets $A \in \mathcal{B}(\mathbb{Z}_{+}^q)$ and $B \in
\mathcal{B}(\mathbb{R}_{+}^p)$, we may write
\begin{align}\nonumber
 Q^m (\mathbf{x}, A\times B)
& = \mathbb{E}^{\mathbf{x}} \left[\1_{A}(\mathbf{S}_m) \1_{B}(\boldsymbol{\lambda}_m)\right] \nonumber \\
& \geq \mathbb{E}^{\mathbf{x}} \left[\1_{A}(\mathbf{S}_m)
  \1_{B}(\boldsymbol{\lambda}_m) \prod_{i=1}^{m}\1_{\{I_i = j_i\}}\right] \nonumber \\
\label{eq:ineq-last-m-step}
& = \1_{A}(\mathbf{s}_{o}) \mathbb{E}^{\mathbf{x}}
\left[\prod_{i=1}^{m-p}\1_{\{I_i = j_i\}}
\mathbb{E}^{\mathbf{X}_{{m-p}}}\left[\1_{B}(\boldsymbol{\lambda}_p) \prod_{i=1}^{p}\1_{\{I_i = j_{m-p+i}\}}\right]\right] \;,
\end{align}
where we used that, given $\mathbf{Y}_i$, $i=0,1,\dots,m-p$, the conditional
distribution of $\mathbf{Y}_j$, $j> m-p$ only depends on $\mathbf{X}_{{m-p}}$
and the Markov property.

Take now $K\geq1$ and $M>0$ and let $\mathbf{x} \in \{1, \cdots, K\}^q \times
(0, M]^p$.  By~(\ref{eq:intensity-dynamic})
and~(\ref{eq:notation-majorants-minorants-alpha}), under
$\{\mathbf{X}_0=\mathbf{x}\}$, we have ${\boldsymbol{\lambda}}_{m-p} \in (0,
M']^p$ where $M' = M + \overline{\alpha} \beta (m-p)$. Moreover, under
$\{\mathbf{X}_0=\mathbf{x},I_i = j_i,\,i=1,\dots,m-p\}$, the last $p$-steps
path $(j_{m-p+1}, \cdots, j_{m})$ is admissible, that is, $(j_{m-p+1}, \cdots,
j_{m}) \in \mathcal{A}_{p} (\mathbf{S}_{m-p})$. Applying
Corollary~\ref{cor:lower-bound-small-set-step-m}, we get that, under
$\{\mathbf{X}_0=\mathbf{x},I_i = j_i,\,i=1,\dots,m-p\}$,
\begin{multline}
\label{eq:ineq-last-p-step}
\mathbb{E}^{\mathbf{X}_{{m-p}}}\left[\1_{B}(\boldsymbol{\lambda}_m) \prod_{i=m-p+1}^{m}\1_{\{I_i = j_i\}}\right] \\ \geq
{\underline{\mu_0}}^{p}\int_{\mathbb{R}_+^p} \1_{B}\left(
\rme^{-\sum_{i=1}^{p}{t_i}\beta}\boldsymbol{\ell}+\rme^{-\sum_{i=2}^{p}{t_i}\beta} \beta\aleph_{j_1}+ \cdots+ \beta\aleph_{j_{p}}\right)
\\
\times \rme^{- C_{m, M'} \sum_{i=1}^{p}{t_i}} \mathrm{d}t_{p}\cdots \mathrm{d}t_1 \;,
\end{multline}
where  $C_{m, M'} = 2(\overline{\mu_0}+ M' + (m-1)\overline{\alpha})$.
Setting $u = \rme^{-\beta\sum_{i=1}^{p}{t_i}}$, $\boldsymbol{\vartheta} =
[\boldsymbol{\vartheta}(1), \cdots, \boldsymbol{\vartheta}(p-1)]^T$ where
$\boldsymbol{\vartheta}(j) = \rme^{-\sum_{i=j+1}^{p}{t_i}\beta}$ for $j = 1,
\dots, p-1$, Inequality (\ref{eq:ineq-last-p-step}) yields, under
$\{\mathbf{X}_0=\mathbf{x},I_i = j_i,\,i=1,\dots,m-p\}$,
\begin{multline}
\label{eq:ineq-last-p-step-2}
\mathbb{E}^{\mathbf{X}_{{m-p}}}\left[\1_{B}(\boldsymbol{\lambda}_m) \prod_{i=m-p+1}^{m}\1_{\{I_i = j_i\}}\right]
\\ \geq
{\underline{\mu_0}}^{p} \frac{1}{\beta^{p-1}} \int_{D} \1_{B}\left(u\boldsymbol{\ell}+\beta\aleph \boldsymbol{\vartheta}\right) \quad u^{\frac{C_{m, M'}}{\beta}} \mathrm{d}\boldsymbol{\vartheta}\mathrm{d}u\;,
\end{multline}
where $D = \{(u, \boldsymbol{\vartheta}): 0< u < \boldsymbol{\vartheta}(1) < \cdots < \boldsymbol{\vartheta}(p-1) < 1\}$.
Applying Lemma~\ref{lem:general-smallset} with $\Gamma = \beta\aleph$ and
$\gamma = \frac{C_{m, M'}}{\beta}$ in (\ref{eq:ineq-last-p-step-2}), we get that there exists $\epsilon_1 > 0$ and a
probability measure $\nu_1$ such that, under
$\{\mathbf{X}_0=\mathbf{x},I_i = j_i,\,i=1,\dots,m-p\}$,
\begin{align}
\label{eq:ineq-last-p-step-3}
\mathbb{E}^{\mathbf{X}_{{m-p}}}\left[\1_{B}\left(\boldsymbol{\lambda}_m\right) \prod_{i=m-p+1}^{m}\1_{\{I_i = j_i\}}\right] \geq \epsilon_1 \nu_1 (B)\;.
\end{align}
Moreover, it is important to note that $\nu_1$ only depends on $\beta$ and $\aleph$.

Using Inequalities (\ref{eq:ineq-last-m-step}) and
(\ref{eq:ineq-last-p-step-3}), we get that
\begin{align}
\label{eq:ineq-last-final-1}
Q^m (\mathbf{x}, A\times B) \geq \epsilon_1 \1_{A}(\mathbf{s}_{o}) \; \nu_1 (B) \mathbb{E}^{\mathbf{x}} \left[\prod_{i=1}^{m-p}\1_{\{I_i = j_i\}}\right] \;.
\end{align}
Using  that  $\mathbf{x} \in \{1, \cdots, K\}^q \times
(0, M]^p$ and Corollary~\ref{cor:lower-bound-small-set-step-m} with $g = 1$, we
further have that
$$
\mathbb{E}^{\mathbf{x}}
\left[\prod_{i=1}^{m-p}\1_{\{I_i = j_i\}}\right] \geq
{\underline{\mu_0}}^{m}\int_{\mathbb{R}_{+}^m} e^{-  C_{m, M} \sum_{i=1}^m t_i}
\; \mathrm{d}t_m\cdots \mathrm{d}t_1 = \left(\frac{\underline{\mu_0}}{C_{m, M}}\right)^{m}\;.
$$
From Inequality (\ref{eq:ineq-last-final-1}), we get that
$Q^m (\mathbf{x}, A \times B) \geq \epsilon \nu(A \times B)$
for some positive constant $\epsilon$ not depending on $\mathbf{x}$ and $\nu$
is the measure on $\mathsf{X}$ defined by $\nu(A \times B) =
\1_{A}(\mathbf{s}_{o}) \nu_1 (B)$, which does neither depend on $M$ nor $K$.
In other words, $C$ is a $(m, \epsilon, \nu)$-small
set (see \cite{meyn-tweedie-2009}).

To obtain that $C$ is also an $(m+1, \epsilon, \nu)$-small set (by possibly
decreasing $\epsilon$), we simply observe that we can carry out the same proof
as above with $m$ replaced by $m+1$ and the sequence of marks $(j_1, \cdots,
j_m)$ replaced by $(0,j_1, \cdots, j_m)$.
\end{proof}

\begin{proof}[Proof of Theorem~\ref{thm:aper-irred}]
  In Proposition~\ref{prop:general-smallset}, the measure $\nu$ does not depend
  on the set $C$, which can be chosen to contain any arbitrary point
  $\mathbf{x}$ in the state space $\mathsf{X}$. Hence the chain is
  $\phi$-irreducible with $\phi=\nu$. Moreover the set $C$ is $(m, \epsilon,
  \nu)$-small and $(m+1, \epsilon, \nu)$-small; hence, the chain is aperiodic,
  sse \cite[Section~5.4.3]{meyn-tweedie-2009}. These arguments hold for the
  Kernel $\tilde{Q}$ and thus we obtain Theorem~\ref{thm:aper-irred}.
\end{proof}

\subsection{Proof of Proposition~\ref{prop:drift-step-1}}
\label{sec:proof-partial-drift-condition}

Recall the definition of $\mathbf{u}$ in~(\ref{eq:defu}).  We shall further
denote
\begin{equation}
  \label{eq:defubar}
1\leq \underline{u}=\min_{i=1,\dots,p}\mathbf{u}(i)
\leq \overline{u}=\max_{i=1,\dots,p}\mathbf{u}(i) <\infty\;.
\end{equation}
The inequality $\underline{u} \geq 1$ follows from the fact that $\aleph^k$ has
non-negative entries for all $k\geq1$.

Observe that, given the initial condition $\mathbf{x}=(\mathbf{s},
  \boldsymbol{\ell}) \in \mathbb{Z}_+^q \times \mathbb{R}_+^q$, the hazard rate
  of $\Delta_1=\min(\Delta_1^i,\,i=0,\dots,p)$ reads
$$
\HR(u) = \sum_{i=0}^p \HR_i (u) = \mu_0^0 + \IdVect_p^T \boldsymbol{\mu}_0 + \IdVect_p^T\boldsymbol{\ell} \rme^{-\beta u} \;,
$$
where we used~(\ref{eq:hazard-rate}).
Moreover, for any $i \in \{0, 1, \cdots, p\}$, we get that, given $\Delta_1$,
the conditional probability that the first arrival's mark is $i$ is given by
\begin{align}
\label{eq:proba-mark-cond}
\mathbb{P}^{\mathbf{x}}(I_1 = i |\Delta_1)={\HR_i (\Delta_1)}/{\HR(\Delta_1)} \;.
\end{align}
Using  (\ref{eq:intensity-dynamic}), we have $\boldsymbol{\lambda}_1=\boldsymbol{\ell}\rme^{-\beta
  \Delta_1}+\beta\aleph_{I_1}\1_{\{I_1\neq0\}}$, and thus
\begin{align*}
& \quad [Q(\1_{\mathbb{Z}_+^q}\otimes V_{1,\gamma})](\mathbf{x}) \\
& = \mathbb{E}^{\mathbf{x}} \left[\exp{\left(\gamma \mathbf{u}^T \boldsymbol{\ell}\rme^{-\beta \Delta_1} + \gamma \beta \mathbf{u}^T \aleph_{I_1}\1_{\{I_1\neq0\}} \right)} \right] \\
& = \mathbb{E}^{\mathbf{x}} \left[
\exp\left(\gamma \mathbf{u}^T \boldsymbol{\ell}\rme^{-\beta \Delta_1}\right)
\sum_{i=0}^p \mathbb{P}^{\mathbf{x}} (I_1 = i | \Delta_1)
\mathrm{e}^{\gamma \beta (\mathbf{u}(i)-1)} \right] \;,
\end{align*}
by conditioning on $\Delta_1$ and using that
$\aleph^T\mathbf{u}=\mathbf{u}-\IdVect_p$ (see~(\ref{eq:defu})) and setting
$\mathbf{u}(0)=1$. Inserting Formula (\ref{eq:proba-mark-cond}), we get
$$
[Q(\1_{\mathbb{Z}_+^q}\otimes V_{1,\gamma})](\mathbf{x})
 = \mathbb{E}^{\mathbf{x}} \left[\sum_{i=0}^p \exp{\left(\gamma \mathbf{u}^T \boldsymbol{\ell}\rme^{-\beta \Delta_1} + \gamma \beta (\mathbf{u}(i)-1) \right)} {\HR_i (\Delta_1)}/{\HR(\Delta_1)}\right]\;.
$$
Using the hazard rate of $\Delta_1$ under
$\mathbb{P}^{\mathbf{x}}$, the corresponding probability density function is
given by $u\mapsto\HR(u)\mathrm{e}^{-\IR(u)}$ on $\mathbb{R}_{+}$, where
\begin{align}\nonumber
\IR_i(t)  & = \int_0^t \HR_i(u)\; \mathrm{d}u,
\quad t\geq0\;,\\
\label{eq:definition-ir}
\IR(t) & = \sum_{i=1}^p \IR^i (t) = \mu_0^0 t + \IdVect_p^T \boldsymbol{\mu}_0 t + \frac{\IdVect_p^T \boldsymbol{\ell} }{\beta} (1 - \rme^{-\beta t}) \;.
\end{align}
Hence we get
$$
[Q(\1_{\mathbb{Z}_+^q}\otimes V_{1,\gamma})](\mathbf{x}) = \int_0^{\infty}
\rme^{\gamma \mathbf{u}^T \boldsymbol{\ell} \rme^{-\beta
    u}-\IR(u)} \left[
\sum_{i=0}^p \rme^{\gamma \beta (\mathbf{u}(i) - 1)} {\HR_i (u)}\right]
\mathrm{d}u\;.
$$
Using that $\rme^{\gamma \beta (\mathbf{u}(0)-1)}\leq \rme^{\gamma \beta (\mathbf{u}(i)-1)}$ for all $i=1,\dots,p$ and the particular form of $\HR(0)$ and
$\HR(i)$ in (\ref{eq:hazard-rate}), we easily obtain that
$$
\sum_{i=0}^p \rme^{\gamma \beta (\mathbf{u}(i)-1)} {\HR_i (u)}
\leq \left[\mu_0^0+
\sum_{i=1}^p \rme^{\gamma \beta (\mathbf{u}(i)-1)}
\left(\boldsymbol{\mu}_{0}(i) + \boldsymbol{\ell}(i)\rme^{-u \beta}\right)
\right]\;.
$$
Consequently, we have
\begin{multline}
\label{eq:drift-uncond-1}
\frac1{V_{1,\gamma}(\boldsymbol{\ell})} [Q(\1_{\mathbb{Z}_+^q}\otimes V_{1,\gamma})](\mathbf{x}) \\
\leq {\left(\mu_0^0 + \sum_{i=1}^p \rme^{\gamma \beta (\mathbf{u}(i)-1)} \boldsymbol{\mu}_{0}(i)\right)}\frac1{V_{1,\gamma}(\boldsymbol{\ell})} \int_0^{\infty} \rme^{\gamma \mathbf{u}^T \boldsymbol{\ell} \rme^{-\beta u}-\IR(u)} \mathrm{d}u \\ + \left(\sum_{i=1}^p \rme^{\gamma \beta (\mathbf{u}(i)-1)} \boldsymbol{\ell}(i)\right)
\frac1{V_{1,\gamma}(\boldsymbol{\ell})}\int_0^{\infty}
\rme^{\gamma \mathbf{u}^T \boldsymbol{\ell} \rme^{-\beta u}-\IR(u) - u \beta}
\mathrm{d}u \;.
\end{multline}
Using (\ref{eq:definition-ir}), we have
\begin{align}
\frac1{V_{1,\gamma}(\boldsymbol{\ell})}
\int_0^{\infty}
\rme^{\gamma \mathbf{u}^T \boldsymbol{\ell} \rme^{-\beta
    u}-\IR(u)}
\mathrm{d}u & =
\int_0^{\infty} \rme^{(\gamma \mathbf{u}^T \boldsymbol{\ell} + \frac{\IdVect_p^T
    \boldsymbol{\ell} }{\beta} )(\rme^{-\beta u} - 1) - (\mu_0^0 + \IdVect_p^T
  \boldsymbol{\mu}_0) u} \mathrm{d}u \nonumber  \\
\label{eq:drift-limit-part-1}
&\to 0 \quad\text{as $\IdVect_p^T \boldsymbol{\ell}\to\infty$}\;,
\end{align}
where we used Lemma~\ref{lem:appendix-asymptotic-equivalence} with $a=\gamma \mathbf{u}^T \boldsymbol{\ell} + \frac{\IdVect_p^T
    \boldsymbol{\ell} }{\beta} $ and $\beta' = \mu_0^0 + \IdVect_p^T \boldsymbol{\mu}_0$.

Similarly, applying Lemma~\ref{lem:appendix-asymptotic-equivalence} with $\beta'= \beta + \mu_0^0 + \IdVect_p^T
  \boldsymbol{\mu}_0$, we obtain

\begin{align}
& \frac1{V_{1,\gamma}(\boldsymbol{\ell})}
\int_0^{\infty}
\rme^{\gamma \mathbf{u}^T \boldsymbol{\ell} \rme^{-\beta
    u}-\IR(u)-u\beta}
\mathrm{d}u \nonumber
\\
& =
\int_0^{\infty} \exp{\left((\gamma \mathbf{u}^T \boldsymbol{\ell} + \frac{\IdVect_p^T
    \boldsymbol{\ell} }{\beta})(\rme^{-\beta u} - 1) - (\beta + \mu_0^0 + \IdVect_p^T
  \boldsymbol{\mu}_0) u\right)} \mathrm{d}u \nonumber \\
\label{eq:drift-limit-part-2}
  & \sim \frac1{\gamma \beta \mathbf{u}^T \boldsymbol{\ell} + \IdVect_p^T
    \boldsymbol{\ell}} \quad\text{as $\IdVect_p^T \boldsymbol{\ell}\to\infty$} \;.
\end{align}
Therefore, inserting (\ref{eq:drift-limit-part-1}) and (\ref{eq:drift-limit-part-2}) in (\ref{eq:drift-uncond-1}), we get
\begin{align*}
\limsup_{\IdVect_p^T \boldsymbol{\ell}\to\infty}
\frac1{V_{1,\gamma}(\boldsymbol{\ell})} [Q(\1_{\mathbb{Z}_+^q}\otimes V_{1,\gamma})](\mathbf{x})
\leq \sup_{\boldsymbol{\ell}\in(0,\infty)^p}\frac{\sum_{i=1}^p \rme^{\gamma \beta (\mathbf{u}(i)-1)}
\boldsymbol{\ell}(i)}{\gamma \beta \mathbf{u}^T \boldsymbol{\ell} + \IdVect_p^T
    \boldsymbol{\ell}} \;.
\end{align*}
Now observe that, using~(\ref{eq:defubar}), we have
$$
0<\underline{u}\leq \mathbf{u}^T\boldsymbol{\ell}/\IdVect_p^T \boldsymbol{\ell}
\leq \overline{u} <\infty\;,
$$
and thus, for $\gamma>0$ small enough, using a Taylor expansion of the
exponential function at the origin, we have
\begin{align*}
\frac{\sum_{i=1}^p
   \mathrm{e}^{\gamma \beta (\mathbf{u}(i)-1)} \boldsymbol{\ell}(i)}
{\beta \gamma
  \mathbf{u}^T \boldsymbol{\ell} + {\IdVect_p^T \boldsymbol{\ell} }}
& \leq \frac{\IdVect_p^T \boldsymbol{\ell}+\gamma\beta(\mathbf{u}^T
  \boldsymbol{\ell} - \IdVect_p^T \boldsymbol{\ell})+ C\,\gamma^2\IdVect_p^T \boldsymbol{\ell}}
{\beta \gamma
  \mathbf{u}^T \boldsymbol{\ell} + {\IdVect_p^T \boldsymbol{\ell} }}
\\
&\leq 1- \beta\gamma\frac{\IdVect_p^T \boldsymbol{\ell}}{\beta \gamma
  \mathbf{u}^T \boldsymbol{\ell} + {\IdVect_p^T \boldsymbol{\ell} }} +
C'\,\gamma^2  \;,\\
&\leq 1- \beta\gamma\frac{1}{1+\beta \gamma \overline{u}} +
C'\,\gamma^2 =1-\beta\gamma+O(\gamma^2) \;,
\end{align*}
where $C, C'$ do not depend on $\boldsymbol{\ell}$. Hence taking
$\gamma_1 > 0$ small enough, we have for all $\gamma\in(0,\gamma_1]$,
$$
\theta(\gamma)=\sup_{\boldsymbol{\ell}\in(0,\infty)^p}\frac{\sum_{i=1}^p
  \mathrm{e}^{\gamma \beta (\mathbf{u}(i)-1)} \boldsymbol{\ell}(i)}
{\beta \gamma
  \mathbf{u}^T \boldsymbol{\ell} + {\IdVect_p^T \boldsymbol{\ell} }} <1\;.
$$
Hence choosing such a $\gamma$ and setting $\theta=\theta(\gamma)$, we
get~(\ref{eq:drift-step-1}) for $\boldsymbol{\ell}$ out of $(0,M]^p$ provided
that $M$ large enough. Since $V_{1,\gamma}$ is bounded on $(0,M]^p$, this
achieves the proof.
\subsection{Proof of Corollary~\ref{cor:moments-checkQvsN}}
\label{sec:proof-coroll-refc}
By \cite[Theorem~(1)]{bacry-delattre-hoffmann-muzy-2010}, using the definitions and notation of
  Section~\ref{sec:constr-hawk-proc}, we have
\begin{equation}
  \label{eq:lln-bacry}
\lim_{T\to\infty}\frac1T\mathbf{N}'(\1_{[0,T]}\otimes\mathbf{w}) = \overrightarrow{\mathbf{w}}^T(\IdMat_p-\aleph)^{-1}\boldsymbol{\mu}_0\quad\text{a.s.}
\end{equation}
On the other hand by the law of large numbers for Harris recurrent chains
(see~\cite[Theorem~17.1.7]{meyn-tweedie-2009}), we have
\begin{equation}
  \label{eq:lln-Qcheck}
  \lim_{n\to\infty}\frac1n\sum_{k=1}^n \mathbf{w}_o(\check{I}_k)=\check{\pi}(\mathbf{w}_o\otimes\1_{\mathbb{R}_+^p})\quad\text{a.s.}[P_*]\;.
\end{equation}
Here $[P_*]$ means that the result holds for any initial distribution on
$(\check{I}_0,\check{\boldsymbol{\lambda}}_0)$. In particular taking
$\check{\boldsymbol{\lambda}}_0=0$, we get that
\begin{equation}
  \label{eq:lln-Qcheck-bis}
\lim_{n\to\infty}\frac1n\mathbf{N}'(\1_{[0,\check{T}_n]}\otimes\mathbf{w}) =\check{\pi}(\mathbf{w}_o\otimes\1_{\mathbb{R}_+^p})\;,
\end{equation}
where $\check{T}_n$ is defined as $T_n$ in Section~\ref{sec:markov-assumption}
but for the unconstrained chain. Note that we used  that $\mathbf{w}_o(0)=0$,
so that only the marks in
$\{1,\dots,p\}$ are counted in the sum of the left-hand side of~(\ref{eq:lln-Qcheck}). Observe
now that under the event $\{\check{\boldsymbol{\lambda}}_0=0\}$, we have
$$
\mathrm{n}_0([0,\check{T}_n])+\mathbf{N}'(\1_{[0,\check{T}_n]}\otimes\1_{\{1,\dots,p\}})=n\;.
$$
(Recall that $\mathrm{n}_0$ is an independent homogeneous Poission point process which corresponds
to the events with marks equal to 0). Since $(\check{T}_n)$ is an increasing
sequence, we get that it diverges a.s. and thus, applying~(\ref{eq:lln-bacry})
with $\mathbf{w}=\1_{\{1,\dots,p\}}$, we
have,  under the event $\{\check{\boldsymbol{\lambda}}_0=0\}$,
$$
\lim_{n\to\infty}\frac1{\check{T}_n}\mathbf{N}'(\1_{[0,\check{T}_n]}\otimes\1_{\{1,\dots,p\}})= \IdVect_p^T(\IdMat_p-\aleph)^{-1}\boldsymbol{\mu}_0
\quad\text{a.s.}
$$
The last two displays yield, under the event $\{\check{\boldsymbol{\lambda}}_0=0\}$,
$$
\lim_{n\to\infty}\frac{\check{T}_n}{n}=\frac1{\mu_0^0+\IdVect_p^T(\IdMat_p-\aleph)^{-1}\boldsymbol{\mu}_0}\;.
$$
Applying~(\ref{eq:lln-bacry}),~(\ref{eq:lln-Qcheck-bis}) and the last  display,
we get~(\ref{eq:moments-checkQvsN}).

\subsection{Proof of Theorem~\ref{thm:nonergodic}}
\label{sec:proof-theor-refthm:n}
  The case~\ref{item:cas-toujours-recurrent} follows from Proposition
  \ref{prop:drift-step-1}.

  We next consider the case~\ref{item:pas-toujours-recurrent} so that
  $\psi(\{s_2\}\times(0,\infty)^p)>0$ for some
  $s_2\geq s^*$. By~(\ref{eq:transiant-cond}) there exists $i\in\{1,\dots,p\}$ such
  that $J(i)>0$. Thus, for any $m\geq1$, the constant sequence
  $(i,\dots,i)$ of length $m$ belongs to $\mathcal{A}_m(s_2)$ (see
  Definition~\ref{def:accessibility}). It follows that
  $\psi(\{s_2+m J(i)\}\times(0,\infty)^p)>0$ for all $m\geq1$.  We shall
  prove that for all $s'\geq1$ and $M>0$,
  \begin{equation}
    \label{eq:to-prove-for-transience}
    \liminf_{s\to\infty}\inf_{\boldsymbol{\ell}\in(0,M]^p}\mathbb{P}^{(s,\boldsymbol{\ell})}\left(\inf_{n\geq1}
    \mathbf{S}_n > s'\right) >0 \;.
  \end{equation}
  Take $s'$ such that $\psi(\{s'\}\times(0,\infty)^p)>0$. Define
  $s=s_2+m J(i)$ and take $m$ and $M>0$ large enough so that
  $\psi(\{s\}\times(0,M)^p)>0$ and
  $\mathbb{P}^{(s,\boldsymbol{\ell})}\left(\inf_{n\geq1} \mathbf{S}_n >
    s'\right) >0$ for all $\boldsymbol{\ell}\in(0,M]^p$. We thus have that the
  probability that the chain $\{\mathbf{X}_k,\,k\geq1\}$ avoids
  $\{s'\}\times(0,\infty)^p$ starting from anywhere in $\{s\}\times(0,M]^p$ is
  positive. Since these two sets have positive $\psi$-measure, we get by~\cite[Theorem~8.3.6]{meyn-tweedie-2009} that $Q$ is transient.

  Hence it only remains to show that~(\ref{eq:to-prove-for-transience}) holds.
  Applying Corollary~\ref{cor:moments-checkQvsN} and
  Condition~(\ref{eq:transiant-cond}), we have
  $$
\lim_{n\to\infty}M_n=\infty\quad \text{a.s.}[P_*], \quad\text{with}\quad M_n=\sum_{k=1}^n J_o(\check{I}_k)\;.
$$
Hence
$U=\inf_{n\geq1}M_n$ is valued in
$\mathbb{R}$ a.s.$[P_*]$ (it equals $-\infty$ with probability 0).  Let
$\epsilon\in(0, 1)$.  Now under the event $\{S_0=s\}$, we have that
$S_n = s+M_n$
as long as $S_k$ does not hit the constraint set
$\cup_{i=1}^p {A_i}$ for $k=1,\dots,n$.

For all  $s\geq s^*$, the
event $\{\inf_{n\geq0} S_n \geq  s^*\}$ coincides
with the event $\{U + s\geq s^*\}$. Hence, under
this latter event we have
$\inf_{n\geq1} S_n=s+U$.
Finally, we get that, for all $s\geq s^*$ and all
$\ell\in\mathbb{R}_+^p$,
\begin{equation}
  \label{eq:non-uniforme-en-ell}
\mathbb{P}^{(s,\boldsymbol{\ell})}\left(\inf_{n\geq1}
    S_n > s'\right) \geq \mathbb{P}^{\boldsymbol{\ell}}\left(U >
    (s'\vee s^*)-s \right) \;.
\end{equation}
(Recall that $U$ only depends on the unconstrained chain
$(\check{I}_n,\check{\boldsymbol{\lambda}}_n)$, $n\geq0$.)
Let us denote, for all $k\geq1$, $U_k=\inf_{n\geq k}M_n$ so that
$U=M_1\wedge\dots\wedge M_k\wedge U_{k+1}$. Since $M_1,\dots,M_k\geq -k\overline{J}$, we have, for all $s_1>k\overline{J}$,
$$
w_{s_1}(\boldsymbol{\ell}):=\mathbb{P}^{\boldsymbol{\ell}}\left(U > - s_1 \right)
=\mathbb{P}^{\boldsymbol{\ell}}\left(U_{k+1} > - s_1 \right)=[\check{Q}^k(w_{s_1})](\boldsymbol{\ell})\;.
$$
Observe that $w_{s_1}$ is a function valued in $[0,1]$. By Proposition~\ref{prop:embeddedunconstrained-geom-erg} and
\cite[Section~15]{meyn-tweedie-2009}), $\check{Q}$ admits a stationary
distribution $\check{\pi}$ and there are constants $\gamma,C>0$ and $\theta\in(0,1)$
not depending on $s_1$ such that, for all
$\boldsymbol{\ell},\boldsymbol{\ell}_o\in\mathbb{R}_+^p$ and $k\geq1$,
$$
[\check{Q}^k(w_{s_1})](\boldsymbol{\ell})\geq [\check{Q}^k(w_{s_1})](\boldsymbol{\ell}_o) - C\,\theta^k\,\left(V_{1,\gamma}(\boldsymbol{\ell})+V_{1,\gamma}(\boldsymbol{\ell}_o)\right)\;.
$$
Now, fix some $\boldsymbol{\ell}_o\in\mathbb{R}_+^p$ and $M>0$. Choose $k$
large enough so that
$$
c=C\,\theta^k\, \left(V_{1,\gamma}(M\IdVect_p)+V_{1,\gamma}(\boldsymbol{\ell}_o)\right)<1/2\;.
$$
Then we may choose $s_1$ large enough
so that $s_1>k\overline{J}$ and
$$
[\check{Q}^k(w_{s_1})](\boldsymbol{\ell}_o)
=\mathbb{P}^{\boldsymbol{\ell}_o}\left(U_{k+1} > -s_1 \right)
\geq \mathbb{P}^{\boldsymbol{\ell}_o}\left(U > -s_1 \right) \geq 1/2\;.
$$
The last four displays yield that, for all $\boldsymbol{\ell}\in(0,M]^p$, and
all $s_1$ large enough,
$$
\mathbb{P}^{\boldsymbol{\ell}}\left(U > - s_1 \right) \geq 1/2-c>0\;.
$$
This and~(\ref{eq:non-uniforme-en-ell})
yields~(\ref{eq:to-prove-for-transience}), which concludes the proof.

\subsection{Proof of Theorem~\ref{thm:geometrical-ergodicqGeQ2} : induction on $q$}
\label{sec:induction-q}

We have already  shown that the chains $Q$ and $\tilde{Q}$ are
$\psi$-irreducible and aperiodic and exhibited petite sets, see
Theorem~\ref{thm:aper-irred}.  To obtain a drift condition on the constrained
case $q\geq1$, we shall reason by induction on the number of constraints $q$.
We shall rely on the transition kernels $\tilde{Q}^{(-\mathcal{J})}$ introduced
in Section~\ref{sec:markov-assumption} for subsets $\mathcal{J}$ in
$\{1,\dots,q\}$. For such a given set $\mathcal{J}$, we shall further denote by
$\{(\check{I}_n, \check{\mathbf{S}}_n,
\check{\boldsymbol{\lambda}}_n),\,n\geq0\}$ the chain with transition kernel
$\tilde{Q}^{(-\mathcal{J})}$ which starts at the same state as $({I}_0,
{\mathbf{S}}_0, {\boldsymbol{\lambda}}_0)$ and, for all
$\mathbf{x}\in\mathbb{Z}_+^{q-\#\mathcal{J}}\times\mathbb{R}_+^p$, by
$\mathbb{E}^{(-\mathcal{J}),\mathbf{x}}[\dots]$ the expectation under
$\{(\check{\mathbf{S}}_0, \check{\boldsymbol{\lambda}}_0)=\mathbf{x}\}$.  (The
fact that this expectation only depends on $\mathbf{x}$ is explained in
Remark~\ref{rem:IvanishesInInitialCondition}.)

\begin{proposition}
\label{prop:drift-step-m-q-larger-1}
Let $p,q\geq1$ and define the
transition kernels $Q$ on the space
$\mathbb{Z}_+^q\times\mathbb{R}_+^p$ as in Section~\ref{sec:markov-assumption}.  Suppose that $\aleph$
is invertible and that Assumption~\ref{ass:Phi} holds.
For any non-empty $\mathcal{J}\subseteq\{1,\dots,q\}$, let $\{(\check{I}_n,
\check{\mathbf{S}}_n, \check{\boldsymbol{\lambda}}_n),\,n\geq0\}$ be a Markov
chain on the space
$\{0,\dots,p\}\times\mathbb{Z}_+^{q-\#\mathcal{J}}\times\mathbb{R}_+^p$ with
transition kernel $\tilde{Q}^{(-\mathcal{J})}$ defined in
Section~\ref{sec:markov-assumption}, and suppose that it satisfies the following
condition.
\begin{enumerate}[label=(C)]
\item\label{item:drift-step-q-larger-1}
For all $\gamma_1>0$ small enough, there exists $\gamma_0^*>0$ such that for all $\gamma_0\in(0,\gamma_0^*]$,
  \begin{equation}
    \label{eq:fromqeq0toqeq2-q-larger-1}
\lim_{m\to\infty}\quad\sup_{(\mathbf{s},\boldsymbol{\ell})\in\mathbb{Z}_+^{q-\#\mathcal{J}}\times\mathbb{R}_+^p}\quad\mathbb{E}^{(-\mathcal{J}),(\mathbf{s},\boldsymbol{\ell})}\left[\rme^{\gamma_0
    \sum_{k=1}^m
    \mathbf{w}(\check{I}_k)}V_{1,\gamma_1}(\check{\boldsymbol{\lambda}}_m-\boldsymbol{\ell})\right]
=0  \;,
  \end{equation}
where we denoted $\mathbf{w}$ is defined in~(\ref{eq:defw}).
\end{enumerate}
Then, for all $\gamma_1>0$ small enough, there exists $\gamma_0^*>0$ such that
for all $\gamma_0\in(0,\gamma_0^*]$, there exist $\theta \in (0, 1)$, $b>0$ and
$m \geq 1$, such that for any initial condition $\mathbf{x} \in\mathsf{X}$,
\begin{equation}
\label{eq:drift-step-m-q-larger-1}
Q^m(V_{0,\gamma_0}^{\otimes q}\otimes V_{1,\gamma_1})(\mathbf{x}) \leq \theta
\; (V_{0,\gamma_0}^{\otimes q}\otimes V_{1,\gamma_1})(\mathbf{x})+b \;,
\end{equation}
where $V_{1,\gamma}$ and $V_{0,\gamma}$ are defined in~(\ref{eq:V1}) and~(\ref{eq:V0}) respectively.
\end{proposition}

\begin{proof}

Take $\mathbf{x}=(\mathbf{s},\boldsymbol{\ell}) \in\mathsf{X}$. We already have the drift condition~(\ref{eq:drift-step-1}) of Proposition~\ref{prop:drift-step-1} (which holds for some $\gamma>0$ and
  thus, by Jensen inequality for any $\gamma_1\in(0,\gamma)$). By induction,
  for all small enough $\gamma_1>0$, there exists $b_1>0$ and $\theta_1>0$,
  such that, for all $m\geq1$,
\begin{align*}
Q^m (\1_{\mathbb{Z}_+^q} \otimes V_{1, \gamma_1})(\mathbf{x})
& \leq \theta_{1}^m V_{1, \gamma_1}(\boldsymbol{\ell})+ b_1 \left(1 + \theta_{1} + \dots + \theta_{1}^m\right) \nonumber \\
& \leq \theta_{1}^m V_{1, \gamma_1}(\boldsymbol{\ell})+ b_1 /\left(1 - \theta_{1}\right)\;.
\end{align*}
By~(\ref{eq:recurrence-of-S})
  and~(\ref{eq:notation-majorants-minorants-J}) we have, in the event $\{\mathbf{S}_0=\mathbf{s}\}$,
  $
  V_{0,\gamma_0}^{\otimes q}(\mathbf{S}_1)\leq \rme^{\gamma_0\overline{J}}\, V_{0,\gamma_0}^{\otimes q}(\mathbf{s}) \;,
$
then we get
$V_{0,\gamma_0}^{\otimes q}(\mathbf{S}_m)\leq\rme^{\gamma_0 \overline{J} m}V_{0,\gamma_0}^{\otimes q}(\mathbf{S}_0)\;.$
Hence, we obtain, for any $\gamma_0>0$,
  \begin{align}\nonumber
    Q^m(V_{0,\gamma_0}^{\otimes q}\otimes V_{1,\gamma_1})(\mathbf{x})
&\leq \rme^{\gamma_0 \overline{J} m}V_{0,\gamma_0}^{\otimes q}(\mathbf{s})
\left[\theta_{1}^m
V_{1,\gamma_1}(\boldsymbol{\ell})+ b_1 /\left(1 - \theta_{1}\right)\right] \\
    \label{eq:drift-ell-large}
&\hspace{-2cm}\leq  \left(\rme^{\gamma_0 \overline{J}}\theta_1\right)^m
(V_{0,\gamma_0}^{\otimes q}\otimes V_{1,\gamma_1})(\mathbf{x})+ \frac{b_1\rme^{\gamma_0 \overline{J}m}}{1-\theta_1}V_{0,\gamma_0}^{\otimes q}(\mathbf{s})\;.
\end{align}
Observe that, for a given $\theta_1<1$, we can choose $\gamma_0>0$ small enough
so that $\rme^{\gamma_0 \overline{J}}\theta_1<1$. It follows that, choosing
such a $\gamma_0$, for all $m\geq1$ and $K\geq1$, the
bound~(\ref{eq:drift-ell-large}) implies~(\ref{eq:drift-step-m-q-larger-1})
for all $\mathbf{x}=(\mathbf{s},\boldsymbol{\ell}) \in\mathsf{X}$ such that
$V_{0,\gamma_0}^{\otimes q}(\mathbf{s})\leq K$.  To address the case $V_{0,\gamma_0}^{\otimes
  q}(\mathbf{s})> K$, we shall prove the following assertion.
\begin{enumerate}[label=(D)]
\item\label{item:remains2} For all $\gamma_1>0$ small enough, there exists
  $\gamma_0^*>0$ such that for all $\gamma_0\in(0,\gamma_0^*]$,
  there exist $m\geq1$, $K>0$, $b\geq0$ and $\theta\in(0,1)$ such
  that~(\ref{eq:drift-step-m-q-larger-1}) holds for all
  $\mathbf{x}=(\mathbf{s},\boldsymbol{\ell}) \in\mathsf{X}$ satisfying
  $V_{0,\gamma_0}^{\otimes q}(\mathbf{s})> K$.
\end{enumerate}
To conclude the proof, we now show that Assertion~\ref{item:remains2}
holds. Let $m\geq1$ and $\mathbf{x}=(\mathbf{s},\boldsymbol{\ell})
\in\mathsf{X}$. Observe that, by~(\ref{eq:recurrence-of-S}), for all
$\gamma_0>0$, we have, in the event $\{\mathbf{X}_0=\mathbf{x}\}$,
\begin{align*}
V_{0,\gamma_0}^{\otimes q}(\mathbf{S}_m)& =V_{0,\gamma_0}^{\otimes
  q}(\mathbf{s})\,\rme^{\gamma_0 \sum_{k=1}^m \sum_{j=1}^q \mathbf{J}_j(I_k)}
 = V_{0,\gamma_0}^{\otimes
  q}(\mathbf{s})\,\rme^{\gamma_0 \sum_{k=1}^m \mathbf{w}(I_k)} \;.
\end{align*}
Hence we have, for all  $\mathbf{x}=(\mathbf{s},\boldsymbol{\ell})
\in\mathsf{X}$,
\begin{align}\nonumber
  Q^m(V_{0,\gamma_0}^{\otimes q}\otimes V_{1,\gamma_1})(\mathbf{x})&=
V_{0,\gamma_0}^{\otimes q}(\mathbf{s}) \,
\mathbb{E}^\mathbf{x}\left[\rme^{\gamma_0 \sum_{k=1}^m \mathbf{w}(I_k)}V_{1,\gamma_1}(\boldsymbol{\lambda}_m)\right]\\
  \label{eq:QmsGrand1}
&=
(V_{0,\gamma_0}^{\otimes q}\otimes V_{1,\gamma_1})(\mathbf{x}) \,
\mathbb{E}^\mathbf{x}\left[\rme^{\gamma_0 \sum_{k=1}^m \mathbf{w}(I_k)}V_{1,\gamma_1}(\boldsymbol{\lambda}_m-\boldsymbol{\ell})\right]\;.
\end{align}
Now the idea of the proof relies on the fact
that if $V_{0,\gamma_0}^{\otimes q}(\mathbf{S}_0)$ is large enough, then it
will take some time for all the components $j$ in a set $\mathcal{J}$
(with cardinal $\#\mathcal{J}$ at least $1$) of the multivariate spread to reach their
constraint set
\begin{align}
\label{eq:def-AofJ}
\mathbf{A}(j)=\bigcup_{i=1,\dots,p} {\mathbf{A}}_i(j)\;.
\end{align}
Up to this time the process defined by
$\{\mathbf{Y}_n^{(-\mathcal{J})}=(I_n,\mathbf{S}_n^{(-\mathcal{J})},\boldsymbol{\lambda}_n),\,n\geq0\}$,
where $\mathbf{S}_n^{(-\mathcal{J})}$ is obtained by removing the $j$-th component of
$\mathbf{S}_n$ for all $j\in\mathcal{J}$, behaves as a Markov chain with kernel $\tilde{Q}^{(-\mathcal{J})}$
(defined as $\tilde{Q}$ but without the $j$-th constraints with
$j\in\mathcal{J}$, see above).  More precisely, we
use the following coupling argument. For a positive integer $m$, we denote
\begin{align}
\label{eq:def-sm}
s^*_{m}(j) = 1+m \overline{J} + \max(\mathbf{A}(j))\;.
\end{align}
We can then define a Markov chain $\{(\check{I}_n, \check{\mathbf{S}}_n,
\check{\boldsymbol{\lambda}}_n),\,n\geq0\}$ with transition kernel
$\tilde{Q}^{(-\mathcal{J})}$ and such that in the event $\{\mathbf{S}_0(j)\geq
s^*_{m}(j),\,j\in\mathcal{J}\}$, the sequence $\{(\check{I}_n, \check{\mathbf{S}}_n,
\check{\boldsymbol{\lambda}}_n),\,n=0, 1, \dots, m\}$ coincides with
$\{(I_n,\mathbf{S}_n^{(-\mathcal{J})}, \boldsymbol{\lambda}_n),\, n=0, 1, \dots, m\}$.
Then, by construction, for all $\mathbf{x}=(\mathbf{s},\boldsymbol{\ell})$
such that $\mathbf{s}(j)\geq s^*_{m}(j)$ for all $j\in\mathcal{J}$, denoting
$\mathbf{x}'=(\mathbf{s}^{(-\mathcal{J})},\boldsymbol{\ell})$, we have
\begin{align}\label{eq:alpha0V1eta1}
 \mathbb{E}^{\mathbf{x}}\left[\rme^{\gamma_0 \sum_{k=1}^m \mathbf{w}(I_k)}V_{1,\gamma_1}(\boldsymbol{\lambda}_m-\boldsymbol{\ell})\right]
= \mathbb{E}^{{(-\mathcal{J})},\mathbf{x}'}\left[\rme^{\gamma_0 \sum_{k=1}^m
    \mathbf{w}(\check{I}_k)}V_{1,\gamma_1}(\boldsymbol{\lambda}_m-\boldsymbol{\ell})\right]\;.
\end{align}
Now, we take $K>0$ such that $\log K\geq
p\gamma_0\max_{j=1,\dots,q}s^*_m(j)$. We let $\mathcal{J}$ denote the set of all indices $j\in\{1,\dots,q\}$
such that $\mathbf{s}(j)\geq s^*_m(j)$. The set $\mathcal{J}$ is not empty since $V_{0,\gamma_0}^{\otimes
  q}(\mathbf{s})> K$ implies $\max(\mathbf{s})\geq \max(s^*_m)$.
Hence, Equality~(\ref{eq:alpha0V1eta1}) applies for  all $\mathbf{x}=(\mathbf{s},\boldsymbol{\ell})
\in\mathsf{X}$ satisfying $V_{0,\gamma_0}^{\otimes
  q}(\mathbf{s})> K$.
For $\gamma^*_0$ as in Condition~\ref{item:drift-step-q-larger-1} (we can
choose it independently of $\mathcal{J}$ since there is a finite number of
such subsets), we obtain Assertion~\ref{item:remains2}.
\end{proof}

Condition~\ref{item:drift-step-q-larger-1} of Proposition~\ref{prop:drift-step-m-q-larger-1} needs
to be simplified. A first step in this direction is given by the following lemma.
\begin{lemma}\label{lem:new-fromqeqtoqPlus1}
  Let $\{(I_n, \mathbf{S}_n,{\boldsymbol{\lambda}}_n),\,n\geq0\}$ be the
  Markov chain on the space $\{0,1,\dots,p\}\times\mathbb{Z}_+^q\times\mathbb{R}_+^p$ with
  transition kernel $\tilde{Q}$ defined in
  Section~\ref{sec:markov-assumption}. Let
  $\mathbf{w}:\{0,1,\dots,p\}\to\mathbb{R}$. Suppose that the following
  assertions hold.
\begin{enumerate}[label=(E-\arabic*)]
\item \label{item:q-larger-1-ergo} For all $\gamma'>0$ small enough, there
  exists $\gamma^*>0$ such that, for all
  $\gamma\in(0,\gamma^*]$,
$\tilde{Q}$ is $(\1_{\{0,1,\dots,p\}}\otimes V_{0,\gamma}^{\otimes q}\otimes
V_{1,\gamma'})$-geometrically ergodic, where $V_{0,\gamma}$ and $V_{1,\gamma'}$ are defined in~(\ref{eq:V0}) and~(\ref{eq:V1}) respectively.
\item \label{item:q-larger-1-decreasing-en-m}
The
  stationary distribution $\tilde{\pi}$ of $\tilde{Q}$ satisfies $\tilde{\pi}\left[\mathbf{w}\otimes\1_{\mathbb{Z}_+^q\times\mathbb{R}_+^p}\right]<0$.
\end{enumerate}
Then, for all $\gamma_1>0$ small enough, there are positive constants $\gamma$, $\rho$ and $\gamma_0^*$ such that, for all
$\gamma_0\in(0,\gamma_0^*]$,
\begin{equation}
    \label{eq:q-larger-1fromqeq0toqeq1}
\sup_{\mathbf{x}\in\mathsf{X}}\;\sup_{m\geq1}\;\frac
{\mathbb{E}^{\mathbf{x}}\left[\rme^{\gamma_0\sum_{k=1}^m
      \{\mathbf{w}(I_k)+\rho\}}\,(V_{0,\gamma}^{\otimes q}\otimes V_{1,\gamma_1})(\mathbf{X}_{m})\right]}
{(V_{0,\gamma}^{\otimes q}\otimes V_{1,\gamma_1})(\mathbf{x})}<\infty   \;.
  \end{equation}
\end{lemma}
\begin{proof}
Using Assumption~\ref{item:q-larger-1-ergo}, for any $\gamma'>0$ small enough, we may choose
$\gamma,\gamma_1\in(0,\gamma^*]$ such that $\tilde{Q}$ is $\1_{\{0, \dots, p\}} \otimes V_{0,\gamma}^{\otimes q} \otimes V_{1,
  {\gamma_1}}$-geometrically ergodic.
Let $\overline{w}=\max_{0\leq i\leq p}\mathbf{w}(i)$. Then, for any
$\gamma^*>0$, there
$\gamma, \gamma_1 > 0$ small enough, we have  for
all $j\geq1$ and $\mathbf{x}=(\mathbf{s},\boldsymbol{\ell})\in\mathsf{X}$,
\begin{align}
\nonumber
  \mathbb{E}^{\mathbf{x}}\left[\rme^{\gamma_0
    \sum_{k=1}^{j}
    \mathbf{w}(I_k)}\,(V_{0,\gamma}^{\otimes q}\otimes V_{1,\gamma_1})(\mathbf{X}_{j})\right]
&\leq \rme^{\gamma_0\overline{w}j}
  \mathbb{E}^{\mathbf{x}}\left[(V_{0,\gamma}^{\otimes q}\otimes V_{1,\gamma_1})(\mathbf{X}_{j})\right]\\
  \label{eq:crude-bound}
&\leq C\,\rme^{\gamma_0\overline{w}j} (V_{0,\gamma}^{\otimes q}\otimes V_{1,\gamma_1})(\mathbf{x}) \;,
\end{align}
where, in the last inequality, we used~\ref{item:q-larger-1-ergo} and thus
the constant $C$ only depends on $\gamma,\gamma_1$.
Let now $m,l\geq1$ and write  $m=nl+j$ with $j\in\{0,\dots,l\}$. Then,
using~(\ref{eq:crude-bound}), the H\"{o}lder inequality and~(\ref{eq:crude-bound})
again, we get
\begin{align}
\nonumber
&\mathbb{E}^{\mathbf{x}}\left[\rme^{\gamma_0
    \sum_{k=1}^{m}
    \mathbf{w}(I_k)}\,\,(V_{0,\gamma}^{\otimes q}\otimes V_{1,\gamma_1})(\mathbf{X}_{m})\right]\\
&\leq  C\,\rme^{\gamma_0\overline{w}j} \, \mathbb{E}^{\mathbf{x}}\left[\rme^{\gamma_0
    \sum_{k=1}^{nl}
    \mathbf{w}(I_k)}\,(V_{0,\gamma}^{\otimes q}\otimes V_{1,\gamma_1})(\mathbf{X}_{nl})\right]\\
\nonumber
&=  C\,\rme^{\gamma_0\overline{w}j} \, \mathbb{E}^{\mathbf{x}}\left[\left\{\prod_{i=1}^l\rme^{\gamma_0
    \sum_{k=0}^{n-1}
    \mathbf{w}(I_{kl+i})}\right\}\,(V_{0,\gamma}^{\otimes q}\otimes V_{1,\gamma_1})(\mathbf{X}_{nl})\right]\\
\nonumber
&\leq C\,\rme^{\gamma_0\overline{w}j} \,
\left(\prod_{i=1}^l\mathbb{E}^{\mathbf{x}}\left[\rme^{l\,\gamma_0
      \sum_{k=0}^{n-1} \mathbf{w}(I_{kl+i})}\,(V_{0,\gamma}^{\otimes q}\otimes
    V_{1,\gamma_1})(\mathbf{X}_{nl})\right]\right)^{1/l}\\
\nonumber
&\leq  C\,\rme^{\gamma_0\overline{w}j} \,
\left(\prod_{i=1}^l C\,\rme^{\gamma_0\overline{w}(l-i)} \,
\mathbb{E}^{\mathbf{x}}\left[\rme^{l\,\gamma_0
      \sum_{k=0}^{n-1} \mathbf{w}(I_{kl+i})}\,(V_{0,\gamma}^{\otimes q}\otimes
    V_{1,\gamma_1})(\mathbf{X}_{(n-1)l+i})\right]\right)^{1/l}\\
  \label{eq:holder}
&\leq C^2 \rme^{2\gamma_0\overline{w}l} \,
\left(\prod_{i=1}^l  \,
\mathbb{E}^{\mathbf{x}}\left[\rme^{l\,\gamma_0
      \sum_{k=0}^{n-1} \mathbf{w}(I_{kl+i})}\,(V_{0,\gamma}^{\otimes q}\otimes
    V_{1,\gamma_1})(\mathbf{X}_{(n-1)l+i})\right]\right)^{1/l}\;.
\end{align}
Applying \cite[Theorem~15.0.1]{meyn-tweedie-2009}, $\tilde{Q}$ has a unique invariant
probability measure $\tilde{\pi}$ and
there exists $R>0$ and $\kappa>0$ such that, for all
 $i\in\{0, \dots, p\}$, $\mathbf{s} \in \mathbb{Z}_+^{q}$ and
$\boldsymbol{\ell}\in\mathbb{R}_+^p$ and all positive integer $m$,
\begin{align}\label{eq:q-larger-1tildeQ-gem-ergodic}
\|\tilde{Q}^{m}((i, \mathbf{s}, \boldsymbol{\ell}),\cdot) - \tilde{\pi}\|_{(\gamma, \gamma_1)} & \leq R\,
\mathrm{e}^{-\kappa m} \;  V_{0,\gamma}^{\otimes q}(\mathbf{s}) \; V_{1, \gamma_1}(\boldsymbol{\ell})\;,
\end{align}
where for any signed measure $\xi$ on $\{0,\dots, p\}\times\mathbb{Z}_+^{q}\times\mathbb{R}_+^p$, we
set
$$
\|\xi\|_{\gamma,\gamma_1}=\sup\left\{\xi(g)~:~ \|g/(\1_{\{0, \dots, p\}} \otimes V_{0,\gamma}^{\otimes q} \otimes V_{1, {\gamma_1}})\|_\infty\leq1\right\}\;.
$$
Define
$$
w^*(i, \mathbf{s}, \boldsymbol{\ell})=\rme^{l\gamma_0 \mathbf{w}(i)}V_{0,\gamma}(\mathbf{s}) V_{1,\gamma_1}(\boldsymbol{\ell})\;,
$$
for all $(i, \mathbf{s}, \boldsymbol{\ell})\in\{0,\dots, p\}\times
\mathbb{Z}_+^q \times\mathbb{R}_+^p$. Then
$$
\left\|w^*/(\1_{\{0, \dots, p\}} \otimes V_{0,\gamma}^{\otimes q} \otimes V_{1,
  {\gamma_1}})\right\|_\infty\leq \rme^{l\gamma_0 \overline{w}} \;.
$$
Thus,~(\ref{eq:q-larger-1tildeQ-gem-ergodic}) yields that, for all $(i, \mathbf{s}, \boldsymbol{\ell})\in\{0,\dots, p\}\times \mathbb{Z}_+^q \times\mathbb{R}_+^p$,
\begin{equation}
  \label{eq:oubli}
\tilde{Q}^{l}((i, \mathbf{s}, \boldsymbol{\ell}),w^*)\leq \tilde{\pi}(w^*)+ R\,
\mathrm{e}^{-\kappa l} \; \rme^{l\gamma_0 \overline{w}} \;  V_{0,\gamma}^{\otimes q}(\mathbf{s}) \; V_{1, \gamma_1}(\boldsymbol{\ell})
\end{equation}
Now, by dominated convergence, we have, as $a\to0$,
$$
\tilde{\pi}\left(a^{-1}\left[\exp(a \mathbf{w}\otimes\1_{\mathbb{Z}_+^q\times\mathbb{R}_+^p})-1\right]\right)
\to\tilde{\pi}\left(\mathbf{w}\otimes\1_{\mathbb{Z}_+^q\times\mathbb{R}_+^p}\right)\;,
$$
which is negative by
Condition~\ref{item:q-larger-1-decreasing-en-m}. Hence there exists $a>0$
such that
\begin{equation}
  \label{eq:ab-def}
b:=\tilde{\pi}\left(\exp(a\mathbf{w}\otimes\1_{\mathbb{Z}_+^q\times\mathbb{R}_+^p})\right)<1 \;.
\end{equation}
Now, by dominated convergence, we may decrease $\gamma^*$
so that $\gamma,\gamma_1\in(0,\gamma^*]$ implies $\tilde{\pi}(w^*)\leq b^{1/2}$
for all $l$ and $\gamma_0$ satisfying $l\gamma_0=a$. Then
we have for these $a$ and $\gamma,\gamma_1$, which set the values of $R$ and $\kappa$, and for
all  $l$ and $\gamma_0$ satisfying $l\gamma_0=a$,
\begin{align*}
\tilde{Q}^{l}((i, \mathbf{s}, \boldsymbol{\ell}),w^*) \leq
\left(b^{1/2}+ R\,
\mathrm{e}^{-\kappa l} \; \rme^{a\overline{w}}\right) \;  V_{0,\gamma}^{\otimes
q}(\mathbf{s}) \; V_{1, \gamma_1}(\boldsymbol{\ell}) \;.
\end{align*}
We now denote by $L$ some integer such that $b^{1/2}+ R\,
\mathrm{e}^{-\kappa l} \; \rme^{a\overline{w}}\leq b^{1/4}$ for all $l\geq
L$. Then, provided that $l\gamma_0=a$ and $l\geq L$, we get
\begin{align*}
\tilde{Q}^{l}((i, \mathbf{s}, \boldsymbol{\ell}),w^*) \leq b^{1/4}V_{0,\gamma}^{\otimes
q}(\mathbf{s}) \; V_{1, \gamma_1}(\boldsymbol{\ell}) \;.
\end{align*}
Using this bound, by successive conditioning, we get that, for all $n\geq1$ and
all $1\leq j\leq l$, and $\mathbf{x}=(\mathbf{s},\boldsymbol{\ell})$, if
$l\gamma_0=a$,
\begin{align*}
&\mathbb{E}^{\mathbf{x}}\left[\rme^{l\,\gamma_0 \sum_{k=0}^{n-1}
    \mathbf{w}(I_{kl+i})} \,(V_{0,\gamma}^{\otimes q}\otimes
    V_{1,\gamma_1})(\mathbf{X}_{(n-1)l+i}))\right]\\
&\leq
b^{(n-1)/4}\mathbb{E}^{\mathbf{x}}\left[\rme^{l\,\gamma_0\mathbf{w}(I_{i})}
 \,(V_{0,\gamma}^{\otimes q}\otimes
    V_{1,\gamma_1})(\mathbf{X}_{i}))\right]\\
&\leq  C\,\rme^{a\,\overline{w}i}\,b^{(n-1)/4}\,(V_{0,\gamma}^{\otimes q}\otimes V_{1, \gamma_1})(\mathbf{x}) \;,
\end{align*}
where, in the last inequality, we used
$\rme^{l\,\gamma_0\mathbf{w}(I_{j})}\leq\rme^{a\overline{w}}$ and
Condition~\ref{item:q-larger-1-ergo} and thus $C$ is some constant not
depending on $n$.
Inserting this in~(\ref{eq:holder}), we get, for $m=nl+j$ with
$j\in\{0,\dots,l\}$ and if $l\gamma_0=a$ with $l\geq L$,
\begin{align}
\nonumber
&\mathbb{E}^{\mathbf{x}}\left[\rme^{\gamma_0
    \sum_{k=1}^{m}
    \mathbf{w}(I_k)}\,\,(V_{0,\gamma}^{\otimes q}\otimes
  V_{1,\gamma_1})(\mathbf{X}_{m})\right]\\
\nonumber
&\leq  C^3 \rme^{3\,a\,\overline{w}} \,b^{(n-1)/4}\,(V_{0,\gamma}^{\otimes q}\otimes V_{1, \gamma_1})(\mathbf{x}) \\
\nonumber
&\leq  C^3 \rme^{3\,a\,\overline{w}}b^{-1/2}
\,b^{m/(4l)}\,(V_{0,\gamma}^{\otimes q}\otimes V_{1, \gamma_1})(\mathbf{x}) \\
&\leq  C^3 \rme^{3\,a\,\overline{w}}b^{-1/2} \,b^{\gamma_0\,m/(4a)}\,(V_{0,\gamma}^{\otimes q}\otimes V_{1, \gamma_1})(\mathbf{x})
\end{align}
Now, this bound holds for any $\gamma_0=a/l$ with $l\geq L$, where the
constants $a$ and $b$ are chosen so that~(\ref{eq:ab-def}) holds and $L$ only
depends on $a,b,\gamma$ and $\gamma_1$.  This
yields~(\ref{eq:q-larger-1fromqeq0toqeq1}) by taking $\rho>0$ such that
$\rme^\rho<b^{-1/(4a)}$. Note that we have shown
that~(\ref{eq:q-larger-1fromqeq0toqeq1}) holds for all
$\gamma_0\in(0,a/L]$ but only if $a/\gamma_0$ is an integer. Nevertheless, it
is easy to show that this can be extended to all $\gamma_0\in(0,\gamma_0^*]$
by choosing $\rho$ and $\gamma_0^*$ adequately.
\end{proof}

If $q=1$ in Proposition~\ref{prop:drift-step-m-q-larger-1}, then we only need to prove
Condition~\ref{item:drift-step-q-larger-1} for the chain $\tilde{Q}^{-\{1\}}$, which
then corresponds to $q=0$, that is $V_{0,\gamma_0}^{\otimes q}\equiv1$.
Hence an immediate consequence of Lemma~\ref{lem:new-fromqeqtoqPlus1} is that
if $q=1$ in  Proposition~\ref{prop:drift-step-m-q-larger-1},
Condition~\ref{item:drift-step-q-larger-1} is implied by the geometric ergodicity of
the unconstrained chain $\check{Q}$ (see Section~\ref{sec:markov-assumption}) and a
simple moment condition on the stationary distribution $\check{\pi}$ of this
chain, see~\ref{item:q-larger-1-decreasing-en-m}.

However, for $q\geq2$, Lemma~\ref{lem:new-fromqeqtoqPlus1} is no longer sufficient to prove
Condition~\ref{item:drift-step-q-larger-1} of Proposition~\ref{prop:drift-step-m-q-larger-1}
because of the presence of $V_{0,\gamma_0}^{\otimes q}$ in the denominator
of~(\ref{eq:q-larger-1fromqeq0toqeq1}). Again we shall rely on an
inductive reasoning to obtain the following result.

\begin{lemma}
  \label{lem:recurrence-controle}
  Let $p,q\geq1$. Let $\{(I_n,
  \mathbf{S}_n,{\boldsymbol{\lambda}}_n),\,n\geq0\}$ be the Markov chain on the
  space $\{0,1,\dots,p\}\times\mathbb{Z}_+^q\times\mathbb{R}_+^p$ with
  transition kernel $\tilde{Q}$ defined in Section~\ref{sec:markov-assumption}.
  For any $\mathcal{J}\subseteq\{1,\dots,q\}$, further define the transition
  kernel $\tilde{Q}^{(-\mathcal{J})}$ on the space
  $\{0,\dots,p\}\times\mathbb{Z}_+^{q-\#\mathcal{J}}\times\mathbb{R}_+^p$, as
  in Section~\ref{sec:markov-assumption}. Let
  $\mathbf{w}:\{0,1,\dots,p\}\to\mathbb{R}$. Suppose that for any subset
  $\mathcal{J}$ in $\{1,\dots,q\}$, the transition kernel
  $\tilde{Q}^{(-\mathcal{J})}$ satisfies
  Assumptions~\ref{item:q-larger-1-ergo}
  and~\ref{item:q-larger-1-decreasing-en-m}. Then for all $\gamma_1>0$
  small enough, there are positive constants $\rho_0$ and $\gamma_0^*>0$ such
  that, for all $\gamma_0\in(0,\gamma_0^*]$,
  \begin{align}
  \label{eq:recurrence-controle}
\sup_{(\mathbf{s},\boldsymbol{\ell})\in\mathsf{X}}\;\sup_{m\geq1}\;\;
\mathbb{E}^{\mathbf{x}}\left[\rme^{\gamma_0
      \sum_{k=1}^{m}{\{\mathbf{w}(I_k)+\rho_0\}}}\,V_{1,\gamma_1}(\boldsymbol{\lambda}_m-\boldsymbol{\ell})\right]
<\infty\;.
 \end{align}
\end{lemma}
\begin{proof}
We prove this lemma by induction on $q$.

For $q=0$, Lemma~\ref{lem:new-fromqeqtoqPlus1} holds with
$V_{0,\gamma}^{\otimes q}\equiv1$, which implies~(\ref{eq:recurrence-controle})
 for all $\gamma_1>0$ small enough with some $\rho_0,\gamma_0^*$ only depending on
$\gamma_1$. This provides the case $q=0$ which initiates the induction.

Next we show the result in the case $q\geq1$ using the induction hypothesis
that the result holds for lower values of $q$.

Define for any $j=1,\dots,q$, $\mathbf{A}(j)$ as in~(\ref{eq:def-AofJ}) and
$$
s^*(j) = 1 + \max(\mathbf{A}(j))\quad\text{and}\quad\overline{s^*}=\max_{j=1,\dots,q}s^*(j)\;.
$$
If   $\mathbf{x}=(\mathbf{s},\boldsymbol{\ell})\in\mathsf{X}$ is such that
$$
\mathbf{s}\in\overline{\mathbf{A}}= \{1,\dots,s^*(1)\}\times\dots\times\{1,\dots,s^*(q)\}\;,
$$
then, applying Lemma~\ref{lem:new-fromqeqtoqPlus1}, for all
$\gamma_1>0$ small enough, we can choose $\gamma,\rho,\gamma^*>0$ such that,
for all $\gamma_0\in(0,\gamma^*]$, there is some constant $C_0>0$ such that, for all $m\geq1$,
\begin{align}\nonumber
  \mathbb{E}^{\mathbf{x}}\left[\rme^{\gamma_0
    \sum_{k=1}^m \mathbf{w}(I_k)}V_{1,
    \gamma_1}(\boldsymbol{\lambda}_m)\right]
&\leq  \mathbb{E}^{\mathbf{x}}\left[\rme^{\gamma_0
    \sum_{k=1}^m \mathbf{w}(I_k)}
(V_{0,\gamma}^{\otimes q}\otimes V_{1,\gamma_1})(\mathbf{X}_{m})\right]\\
\nonumber
&\leq C_0 \,  \rme^{- \rho \gamma_0 m} \,(V_{0,\gamma}^{\otimes q}\otimes V_{1,\gamma_1})(\mathbf{x})\\
\label{eq:bound-insideAbar}
&\leq C_0\,\rme^{q\overline{s^*}\gamma} \,  \rme^{- \rho \gamma_0 m}\,V_{1, \gamma_1}(\boldsymbol{\ell}) \;.
\end{align}
We denote by $\tau$ the first hitting time of the set $\overline{A}$,
$$
\tau=\inf\left\{k\geq0,\, \mathbf{S}_k\in\overline{A} \right\}
$$
and we set $\tau_m=\tau\wedge m$, which is a stopping time with respect to the
filtration $\mathcal{F}_k=\sigma(\mathbf{Y}_n,\,0\leq n\leq k)$, $k\geq0$.
We define by $\mathcal{J}_m$   the set of indices $j$ such
that $\mathbf{S}_k(j)$ stay away of $\mathbf{A}(j)$ for all
$k=0,1,\dots,\tau_m$, which is a random set measurable with respect to
$\mathcal{F}_{\tau_m}$ and is non-empty by definition
of $\tau_m$.
By definition of $\tau_m$, we have, in the event $\{\tau_m<m\}$,
$\mathbf{S}_{\tau_m}\in\overline{A}$ and thus, by~(\ref{eq:bound-insideAbar}),
setting $C_1=C_0\,\rme^{q\overline{s^*}\gamma}$,
\begin{align}
  \label{eq:stopping-time-bound}
\mathbb{E}\left[\rme^{\gamma_0
    \sum_{k=\tau_{m}+1}^m \mathbf{w}(I_k)}V_{1, \gamma_1}(\boldsymbol{\lambda}_m)\,\vert\,\mathcal{F}_{\tau_m}\right]
&\leq C_1 \,  \rme^{- \rho \gamma_0(m-\tau_m)}\,V_{1,\gamma_1}(\boldsymbol{\lambda}_{\tau_m})\quad \text{a.s.}
\end{align}
On the other hand, up to $n=\tau_m$ the chain $\mathbf{X}_n$ (with transition
kernel $\tilde{Q}$) coincides with the chain with transition kernel
$\tilde{Q}^{(-\mathcal{J}_m)}$ that starts at the same state, where
$\mathcal{J}_m$ is defined as above.  Hence we may write, for all integer
$l=1,2,\dots,m$, all non-empty subsets $\mathcal{J}$ of $\{1,\dots,q\}$ and
$\mathbf{x}=(\mathbf{s},\boldsymbol{\ell})\in\mathsf{X}$ satisfying
$\mathbf{s}\notin\overline{\mathbf{A}}$,
\begin{align*}
&\mathbb{E}^{\mathbf{x}}\left[\rme^{\gamma_0
    \sum_{k=1}^{\tau_m} \mathbf{w}(I_k)}\rme^{- \rho \gamma_0(m-\tau_m)}\,V_{1,
  \gamma_1}(\boldsymbol{\lambda}_{\tau_m})\,\1_{\{\tau_m=l\}\cap\{\mathcal{J}_m=\mathcal{J}\}}\right]\\
&\phantom{ blabla} = \mathbb{E}^{(-\mathcal{J}),\mathbf{x}}\left[\rme^{\gamma_0
    \sum_{k=1}^{l} \mathbf{w}(\check{I}_k)}\rme^{- \rho\gamma_0 (m-l)}\,V_{1,
  \gamma_1}(\check{\boldsymbol{\lambda}}_{l})\,\1_{\{\tau_m=l\}\cap\{\mathcal{J}_m=\mathcal{J}\}}\right]  \\
&\phantom{ blabla}\leq \rme^{- \rho \gamma_0(m-l)}\,\mathbb{E}^{(-\mathcal{J}),\mathbf{x}}\left[\rme^{\gamma_0
    \sum_{k=1}^{l} \mathbf{w}(\check{I}_k)}\,V_{1,
  \gamma_1}(\check{\boldsymbol{\lambda}}_{l})\right]\,,
\end{align*}
where $\{(\check{I}_n, \check{\mathbf{S}}_n,
\check{\boldsymbol{\lambda}}_n),\,n\geq0\}$ here denotes the chain with transition kernel
$\tilde{Q}^{(-\mathcal{J})}$ which starts at the same state as $({I}_0,
{\mathbf{S}}_0, {\boldsymbol{\lambda}}_0)$.
Using  the induction hypothesis,  if $\gamma_1>0$ is small enough, there are positive constants
$\gamma'$ and $\rho'$ such that for all $\gamma_0\in(0,\gamma']$, $\mathbf{x}\in\mathsf{X}$ and $l\geq1$,
$$
\mathbb{E}^{(-\mathcal{J}),\mathbf{x}}\left[\rme^{\gamma_0
    \sum_{k=1}^{l} \mathbf{w}(\check{I}_k)}\,V_{1,
  \gamma_1}(\check{\boldsymbol{\lambda}}_{l})\right]\leq C_2\,\rme^{-\rho'\gamma_0\,l}\,
V_{1, {\gamma_1}} (\boldsymbol{\ell}) \;,
$$
for some constant $C_2>0$ independent of $l$ and $\mathbf{x}$.
Inserting this bound in the previous display and summing over all $l=0,1,\dots,m$
and $\mathcal{J}$, we get, for all $\gamma_0\in(0,\gamma']$ and
$\mathbf{x}=(\mathbf{s},\boldsymbol{\ell})\in\mathsf{X}$ such that
$\mathbf{s}\notin\overline{\mathbf{A}}$,
\begin{align*}
\mathbb{E}^{\mathbf{x}}\left[\rme^{\gamma_0
    \sum_{k=1}^{\tau_m} \mathbf{w}(I_k)}\rme^{- \rho \gamma_0(m-\tau_m)}\,V_{1,
  \gamma_1}(\boldsymbol{\lambda}_{\tau_m})\right]
&\leq C_3\sum_{l=0}^m\rme^{- \rho\gamma_0 (m-l)-\rho'\gamma_0\,l}V_{1,\gamma_1}(\boldsymbol{\ell})\\
&\leq C_3\,m\,\rme^{- (\rho\wedge\rho') \gamma_0 m}V_{1,\gamma_1}(\boldsymbol{\ell}) \;,
\end{align*}
where $C_3>0$ is independent of $m$ and  $\mathbf{x}$.
This last bound,
with~(\ref{eq:stopping-time-bound}) yields, for all $\gamma_0\in(0,\gamma'\wedge\gamma^*]$ and
$\mathbf{x}=(\mathbf{s},\boldsymbol{\ell})\in\mathsf{X}$,
\begin{align*}
 \mathbb{E}^{\mathbf{x}}\left[\rme^{\gamma_0 \sum_{k=1}^{m}{\mathbf{w}(I_k)}}V_{1,\gamma_1}(\boldsymbol{\lambda}_{m})
 \right] & \leq  \mathbb{E}^{\mathbf{x}}\left[\rme^{\gamma_0
    \sum_{k=1}^{\tau_m} \mathbf{w}(I_k)}
\mathbb{E}\left[\rme^{\gamma_0 \sum_{k=\tau_m+1}^{m}{\mathbf{w}(I_k)}}V_{1,\gamma_1}(\boldsymbol{\lambda}_{m})\,\vert\,\mathcal{F}_{\tau_m}\right]\right]\\
&\leq  C_1\,  \mathbb{E}^{\mathbf{x}}\left[\rme^{\gamma_0
    \sum_{k=1}^{\tau_m} \mathbf{w}(I_k)} \rme^{- \rho (m-\tau_m)}\,V_{1,
  \gamma_1}(\boldsymbol{\lambda}_{\tau_m})\right]\\
&\leq C_4\,m\,\rme^{- (\rho\wedge\rho') m}\,V_{1,\gamma_1}(\boldsymbol{\ell})\;,
\end{align*}
where $C_4$ is
some positive constant independent of $m$ and  $\mathbf{x}$. Now, taking $\gamma_0^*=\gamma'\wedge\gamma^*$
$\rho_0\in(0,\rho\wedge\rho')$, we get~(\ref{eq:recurrence-controle}).
\end{proof}

\begin{theorem}
\label{thm:geometrical-ergodic}
Let $p,q\geq1$ and suppose that Assumptions~\ref{ass:Phi} and~\ref{ass:access}
hold. Let $\{(I_n, \mathbf{S}_n,{\boldsymbol{\lambda}}_n),\,n\geq0\}$ be the
Markov chain on the space
$\{0,1,\dots,p\}\times\mathbb{Z}_+^q\times\mathbb{R}_+^p$ with transition
kernel $\tilde{Q}$ defined in Section~\ref{sec:markov-assumption}.  For any
$\mathcal{J}\subseteq\{1,\dots,q\}$, further define the transition kernel
$\tilde{Q}^{(-\mathcal{J})}$ on the space
$\{0,\dots,p\}\times\mathbb{Z}_+^{q-\#\mathcal{J}}\times\mathbb{R}_+^p$, as in
Section~\ref{sec:markov-assumption}. Define
\begin{equation}
  \label{eq:defw}
\mathbf{w}(i)=
\begin{cases}
\sum_{j=1}^q \overrightarrow{\mathbf{J}}_{i,j} &\text{ for all
$i=1,\dots,p$}\\
0 &\text{ for $i=0$}.
\end{cases}
\end{equation}
Suppose that for any subset
$\mathcal{J}$ in $\{1,\dots,q\}$, the transition kernel
$\tilde{Q}^{(-\mathcal{J})}$ satisfies
\begin{enumerate}[label=(B-\arabic*)]
\item  \label{item:condition-induction-1-geom-ergodic}
For all $\gamma'>0$ small enough, there exists $\gamma^*>0$ such that
for all $\gamma\in(0,\gamma^*]$, $\tilde{Q}^{(-\mathcal{J})}$ is $(\1_{\{0,1,\dots,p\}}\otimes V_{0,\gamma}^{\otimes (q-\#\mathcal{J})}\otimes
V_{1,\gamma'})$-geometrically ergodic.
\item  \label{item:condition-induction-2-geom-ergodic} The stationary distribution $\tilde{\pi}^{(-\mathcal{J})}$ of $\tilde{Q}^{(-\mathcal{J})}$ satisfies $\tilde{\pi}^{(-\mathcal{J})}\left[\mathbf{w}\otimes\1_{\mathbb{Z}_+^q\times\mathbb{R}_+^p}\right]<0$.
\end{enumerate}
Then, for all $\gamma_1>0$ small enough, there exists $\gamma_0^*>0$ such that
for all $\gamma_0\in(0,\gamma_0^*]$,  the transition kernel $\tilde{Q}$ is  $(\1_{\{0,1,\dots,p\}}\otimes V_{0,\gamma_0}^{\otimes q}\otimes
V_{1,\gamma_1})$-geometrically ergodic.
\end{theorem}

In practice, Theorem~\ref{thm:geometrical-ergodic} should be applied by
induction, so that only Conditions of the
form~\ref{item:condition-induction-2-geom-ergodic} have to be checked out. The
case $q=0$ is treated in
Proposition~\ref{prop:embeddedunconstrained-geom-erg}. We treat the case $q=1$
in the following section. Then we explain in Section~\ref{sec:case-q-lager-1} how to perform this
induction for $q\geq2$.

\begin{proof}[Proof of Theorem~\ref{thm:geometrical-ergodic}]
  By Conditions~\ref{item:condition-induction-1-geom-ergodic}
  and~\ref{item:condition-induction-2-geom-ergodic}, we may apply
  Lemma~\ref{lem:recurrence-controle} to show that
  Condition~\ref{item:drift-step-q-larger-1} of
  Proposition~\ref{prop:drift-step-m-q-larger-1} holds. Hence, applying this
  proposition, we get the drift condition~(\ref{eq:drift-step-m-q-larger-1}).

  The small sets of Theorem~\ref{thm:aper-irred} are of the form
  $\{1,\dots,K\}^q\times(0,M]^p$. Observe that the sublevel sets of
  $V_{0,\gamma_0}^q\otimes V_{1,\gamma_1}$
$$
C_{\gamma_0,\gamma_1}(r)=\left\{\mathbf{x}\in\mathbb{Z}_+\times\mathbb{R}_+^p~:~V_{0,\gamma_0}^q\otimes
  V_{1,\gamma_1}(\mathbf{x})\leq r\right\},\quad r>0\;,
$$
are included in such sets. Hence the drift
condition~(\ref{eq:drift-step-m-q-larger-1}) applies to show that the chain $Q$
is $(V_{0,\gamma_0}\otimes V_{1,\gamma_1})$-geometrically ergodic,
see~\cite[Chapter~15]{meyn-tweedie-2009}.  The fact that $\tilde{Q}$ is
$(\1_{\{0,\dots,p\}} \otimes V_{0,\gamma_0}\otimes
V_{1,\gamma_1})$-geometrically ergodic follows as in the proof of
Proposition~\ref{prop:embeddedunconstrained-geom-erg}.
\end{proof}

\begin{proof}[Proof of Theorem~\ref{thm:geometrical-ergodicqGeQ2}]
  This result is obtained by induction on $q$ using
  Theorem~\ref{thm:geometrical-ergodic} and
  Theorem~\ref{thm:geometrical-ergodicqIs1} to initiate the induction at $q=1$.
\end{proof}




\subsection{Proof of Corollary~\ref{cor:geo-ergo-overlineQ}}
\label{sec:proof-coroll-refc-1}
  We first show that for all $D > 0$, $K \geq 1$ and $M > 0$, the set $C = (0,
  D] \times \{0,\dots,p\}\times\{1,\dots,K\}^q \otimes(0, M]^p$ is a petite
  set.  By Corollary~\ref{cor:smallset-Qtilde}, there exits $m\geq1$, $\epsilon
  > 0$ and a probability measure $\tilde{nu}$ on $\mathsf{Y}$ such that
  $\{0,\dots,p\}\times\{1,\dots,K\}^q \otimes(0, M]^p$ is an $(m,\epsilon,\tilde\nu)$
  and $(m+1,\epsilon,\tilde\nu)$-small set for $\tilde{Q}$.
   For all $\mathbf{z} = (\delta,i, \mathbf{s},\boldsymbol{\ell})\in C$, under the initial condition
  $\mathbf{Z}_0 = \mathbf{z}$, we have $\boldsymbol{\lambda}_{m-1} \in \left(0,
    M + (m-1) \overline{\alpha}\beta\right]^p$. It follows that the density of
  $\Delta_m$ is bounded from below by $c_0 \rme^{-(m-1)c_1t}$ over $t \in
  \mathbb{R}_+$ for some positive constants $c_0,c_1$ (see the proof of
  Lemma~\ref{lem:lower-bound-small-set-step-1}). We get that for all Borel sets $A\subset\mathbb{R}_+$ and $B\subset\mathsf{Y}$,
\begin{align*}
\overline{Q}^{m}(\mathbf{z}, A\times B) & \geq c_0  \int_{A} \rme^{-c_1 t}\mathrm{d} t \; \tilde{Q}^{m}(\mathbf{y}, B)
\\ & \geq \epsilon (c_0/c_1) \nu (A\times B)\;,
 \end{align*}
 where we used that $\{0,\dots,p\}\times\{1,\dots,K\}^q \otimes(0, M]^p$ is
 an $(m,\epsilon,\tilde\nu)$-small set for $\tilde{Q}$ and $\nu$ is the probability
 measure defined by $\nu (A\times B)=\int_Ac_1\rme^{-c_1 t}\mathrm{d} t\,\tilde{\nu}(B)$.
Hence $C$ is an $(m,\epsilon,\nu)$-small set for $\overline{Q}$. Similarly, we have that
is also an $(m+1,\epsilon,\tilde\nu)$-small set. Therefore, we conclude that
$\overline{Q}$ is $\psi$-irreducible and aperiodic and that all bounded subsets
of $\mathsf{Z}$ are petite sets.

Next we show that the drift condition  obtained in
Proposition~\ref{prop:drift-step-m-q-larger-1} for $\tilde{Q}$ extends to
$\overline{Q}$. As explained in Section~\ref{sec:case-q-lager-1}, the
assumptions of Theorem~\ref{thm:geometrical-ergodicqGeQ2} imply that of
Theorem~\ref{thm:geometrical-ergodic} and thus of Proposition~\ref{prop:drift-step-m-q-larger-1}
(see the proof of Theorem~\ref{thm:geometrical-ergodicqGeQ2}).  Hence, for all
$\gamma_1>0$ small enough, there exists $\gamma_0^*>0$, $\theta \in (0, 1)$,
$b>0$ and $m \geq 1$ such that for all $\gamma_0\in(0,\gamma_0^*]$ such that
for any initial condition $\mathbf{x}\in\mathsf{X}$, we have
\begin{equation*}
\left[\tilde{Q}^m(V_{0,\gamma_0}\otimes V_{1,\gamma_1})\right](\mathbf{x}) \leq \theta \; (V_{0,\gamma_0}\otimes V_{1,\gamma_1})(\mathbf{x})+b \;.
\end{equation*}

By the Cauchy-Schwartz inequality and using that for all $\mathbf{x} \in
\mathsf{X}$, $\Delta_m$ is stochastically smaller than an exponential
distribution with mean $(\IdVect_p^T \boldsymbol{\mu}_0)^{-1}$ (see the proof
of Lemma~\ref{lem:lower-bound-small-set-step-1}), we have
\begin{align*}
& \overline{Q}^m(V_{2,\gamma_2}\otimes \1_{\{0,\dots,p\}} \otimes V_{0,\gamma_0/2}\otimes V_{1,\gamma_1/2})(\Delta, i, \mathbf{x})
\\ & \leq
\left(\overline{\mathbb{E}}^{\mathbf{x}}\left[V_{2,2\gamma_2}(\Delta_m)\right]\right)^{1/2} \left(\left[\;\overline{Q}^m(\1_{\{0,\dots,p\}} \otimes V_{0,\gamma_0}\otimes V_{1,\gamma_1})\right](\Delta, i, \mathbf{x})\right)^{1/2}
\\
&
\leq \left(\int_{0}^{+\infty} \left(\IdVect_p^T \boldsymbol{\mu}_0\right)\rme^{2\gamma_2 r} \rme^{-\IdVect_p^T \boldsymbol{\mu}_0 r} \rmd r\right)^{1/2}\;  \left(\theta \; (V_{0,\gamma_0}\otimes V_{1,\gamma_1})(\mathbf{x})+b\right)^{1/2} \;.
\end{align*}
Provided that $\gamma_2 < \IdVect_p^T \boldsymbol{\mu}_0/2$, we thus get
\begin{multline*}
\overline{Q}^m(V_{2,\gamma_2}\otimes \1_{\{0,\dots,p\}} \otimes V_{0,\gamma_0/2}\otimes V_{1,\gamma_1/2})(\Delta, i, \mathbf{x})
\\
\leq \left(1-2\gamma_2/\IdVect_p^T \boldsymbol{\mu}_0\right)^{-1/2}\;  \left[\theta^{1/2} \; (V_{0,\gamma_0/2}\otimes V_{1,\gamma_1/2})(\mathbf{x})+b^{1/2}\right]
\;.
\end{multline*}
Hence, for $\gamma_2 > 0$ small enough, we have, for some $\theta' \in (0, 1)$ and $b' > 0$,
$$ \overline{Q}^m(V_{2,\gamma_2}\otimes \1_{\{0,\dots,p\}} \otimes V_{0,\gamma_0/2}\otimes V_{1,\gamma_1/2})(\Delta, i, \mathbf{x})
 \leq \theta' (V_{0,\gamma_0/2}\otimes V_{1,\gamma_1/2})(\mathbf{x})+b' \;.
$$
Hence we obtain a drift condition with an unbounded off petite set function
$V_{2,\gamma_2}\otimes \1_{\{0,\dots,p\}} \otimes V_{0,\gamma_0/2}\otimes
V_{1,\gamma_1/2}$. The geometric ergodicity follows (see \cite[Section~15]{meyn-tweedie-2009}).
\subsection{Proof of Theorem~\ref{thm:fclt-cont-time}}
\label{sec:proof-theor-refeq:fc}
We need the following lemmas to complete the proof for this theorem. Denote the number of arrivals on $[0, u]$ by $\mathbf{N}_u = \mathbf{N} \left(\1_{(0,u]}\otimes \1_{\{0,\dots,p\}}\right)$ and the time interval between the last arrival before $u$ and $u$ by $\Gamma_{u} = u - \sum_{k=1}^{\mathbf{N}_u} \Delta_k$.

\begin{lemma}
\label{lem:Nu-sur-u}
We have
$$
\lim_{u \to \infty}\frac{\mathbf{N}_u}{u} \to \frac{1}{\overline{\mathbb{E}}(\Delta_1)}\;, \textit{ a.s. } [P_*] \;.
$$
\end{lemma}
\begin{proof}
By definition of $\mathbf{N}_u$, we have
$$
\sum_{k=1}^{\mathbf{N}_u} \Delta_k \leq u \leq \sum_{k=1}^{\mathbf{N}_u+1} \Delta_k \;.
$$

Since $\overline{Q}$ is positive Harris, using~\cite[Theorem~17.3.2]{meyn-tweedie-2009}, we get
$$
\lim_{n\to\infty}\frac{1}{n}\sum_{k=1}^{n} \Delta_k = \overline{\mathbb{E}}(\Delta_1)\;, \textit{ a.s. } [P_*] \;.
$$
We deduce that
$$
\lim_{u\to+\infty}\frac{u}{{\mathbf{N}_u}} = \lim_{u\to+\infty}\frac{\sum_{k=1}^{\mathbf{N}_u} \Delta_k}{{\mathbf{N}_u}} = \lim_{u\to+\infty}\frac{\sum_{k=1}^{\mathbf{N}_u+1} \Delta_k}{{\mathbf{N}_u}} = \overline{\mathbb{E}}(\Delta_1) \;, \textit{ a.s. } [P_*] \;.
$$
Hence the result.
\end{proof}

\begin{lemma}
\label{lem:gamma-tau-T}
For all $T \in \mathbb{R}_+$, we have
$$
\sup_{\tau\in[0,1]} \Gamma_{\tau T}^{\mathbf{x}} = o_{P_*} \left(\sqrt{T}\right) \;,
$$
where as in the notation a.s.$[P_*]$, the subscript $P_*$ in $o_{P_*}$
indicates that the result holds for any initial distribution.
\end{lemma}

\begin{proof}
Let us define
\begin{equation}
\label{eq:def-nT}
n_T = \left[\frac{T}{\overline{\mathbb{E}}(\Delta_1)}\right] \;.
\end{equation}
By definition of $\Gamma_{u}$, for all $\tau\in[0,1]$, we have that $\Gamma_{\tau T} \leq \max_{k=1,\dots,\mathbf{N}_T+1} {\Delta_k}$.

For any $\tau\in[0,1]$, $\epsilon > 0$ and $T \in \mathbb{R}_+$, we get, for all $\mathbf{x} \in \mathsf{X}$,
$$
\mathbb{P}\left(\Gamma_{\tau T} \geq \epsilon \sqrt{T}\right) \leq \mathbb{P}^{\mathbf{x}}\left(\max_{k=1,\dots,2 n_T} {\Delta_k} \geq \epsilon \sqrt{T}\right) + \mathbb{P}^{\mathbf{x}}\left(\mathbf{N}_T + 1 > 2 n_T \right)\;.
$$
Using Lemma~\ref{lem:Nu-sur-u}, we have $\lim_{T\to\infty} \mathbb{P}^{\mathbf{x}}\left(\mathbf{N}_T + 1 > 2 n_T \right) = 0$. Observe that we have, for all $k\geq 1$, $\Delta_k \overset{st}{\leq} E_k$ where $\{E_i\}_{i \geq 1}$ are the inter-arrivals of a Poisson point process with intensity $\1_p^T \boldsymbol{\ell}$. We get, for all $\mathbf{x}\in\mathsf{X}$,
$$
\lim_{T\to\infty} \mathbb{P}^{\mathbf{x}}\left(\max_{k=1,\dots,2 n_T} {\Delta_k} \geq \epsilon \sqrt{T}\right) \leq \lim_{T\to\infty} \mathbb{P}^{\mathbf{x}}\left(\max_{k=1,\dots,2 n_T} {E_k} \geq \epsilon \sqrt{T}\right) = 0\;,
$$
which concludes the proof.
\end{proof}

We now provide a proof of Theorem~\ref{thm:fclt-cont-time} using Lemmas~\ref{lem:Nu-sur-u} and~\ref{lem:gamma-tau-T}.

\begin{proof}[Proof of Theorem~\ref{thm:fclt-cont-time}]
For all $t \in [0, 1]$, $T \in \mathbb{R}_+$ and $\mathbf{x} \in \mathsf{X}$, we have
\begin{align*}
\left(\mathbf{N}^{\mathbf{x}}(\1_{[0,tT]}\otimes\mathbf{w})-t T \; E(\mathbf{w})\right)
 =
 \sum_{k=1}^{\mathbf{N}_{tT}^{\mathbf{x}}} \left(\mathbf{w}(I_k) - E(\mathbf{w}) \Delta_k\right) - E(\mathbf{w}) \Gamma_{tT}^{\mathbf{x}} \;.
\end{align*}
We also observe that, defining $n_T$ as in~(\ref{eq:def-nT}),
$$
 \sum_{k=1}^{\mathbf{N}_{tT}^{\mathbf{x}}} \left(\mathbf{w}(I_k) - E(\mathbf{w}) \Delta_k\right)
=  \Sigma_{n_T}\left(\frac{\mathbf{N}_{tT}^{\mathbf{x}}}{n_T}, g\right)=s_{n_T}\left(\frac{\mathbf{N}_{tT}^{\mathbf{x}}}{n_T}, g\right)\;,
$$
where we used that, by definition of $E(\mathbf{w})$, $\overline{\pi} (g) = 0$
and that $\mathbf{N}_{tT}$ is an integer.
Hence by Lemma~\ref{lem:gamma-tau-T}, we get
\begin{equation}
  \label{eq:fclt-beforecomposition}
\sup_{t\in[0,1]}\left|\left(\mathbf{N}^{\mathbf{x}}(\1_{[0,tT]}\otimes\mathbf{w})-t T \;
  E(\mathbf{w})\right) - s_{n_T}\left(\frac{\mathbf{N}_{tT}^{\mathbf{x}}}{n_T},
  g\right)\right|= o_{P_*}(\sqrt{T})\;,
\end{equation}
Now, by Lemma~\ref{lem:Nu-sur-u}, we have, for all $t\in[0,1]$,
$\mathbf{N}_{tT}^{\mathbf{x}}/n_T\to t$ as $T\to\infty$ a.s.$[P_*]$. By the Dini theorem,
this convergence holds uniformly over $t\in[0,1]$. Applying
\cite[Theorem~13.2.1]{whitt-2002}, this convergence and the one
in~(\ref{eq:fclt}) implies that, for any initial distribution,
$$
(n_T\sigma_g^2)^{-1/2} s_{n_T} \left(\frac{\mathbf{N}_{tT}^{\mathbf{x}}}{n_T}, g\right) \overset{d}{\to} B_t \;,
$$
where the weak convergence holds in $D([0,1])$.
This,~(\ref{eq:fclt-beforecomposition}), $\sigma_g^2=v(\mathbf{w})$ and the fact that $n_T/T\to
1/\overline{\mathbb{E}}[\Delta_1]$ yield~(\ref{eq:fclt-realtime}), which
achieves the proof.
\end{proof}

\section{Conclusion}

In this paper we have introduced and studied constrained multivariate Hawkes
processes. The constraints are expressed using a multidimensional variable, whose evolution is driven by the point
process. Under the Markov setting (exponential fertility functions), we have
proven that the underlying Markov chain is $V$-geometrically ergodic under some
conditions on the parameters. A converse result leading to the transience of
the chain in the case where an univariate spread variable is used ($q=1$)
illustrates the sharpness of our conditions.  Moreover, in the general case, we
used a functional central limit theorem applying to the chain to derive the
scaling limit of the integrated point process in physical time.
Finally we have briefly explained how
the constrained multivariate Hawkes process can be applied to model the
dynamics of a limit order book. The scaling limit of the mid-price can be
deduced from our findings.

We only presented the case of the order book of one asset with only two
limits. This is clearly not a restriction.  Multi-assets limit order books can
be considered, yielding multivariate boundary conditions (hence $q\geq2$).  In
a forthcoming paper, we will use this model to provide an empirical study from
real data and discuss the potential application of this model to the dynamics
of limit order books.  We believe that the applicability of the constrained
multivariate Hawkes process extends well beyond financial applications, as it
could model the motion of an object on a discrete net with some boundary
conditions.

\section{Acknowledgements}
This research is supported by NATIXIS quantitative research department.

\appendix

\section{Additional technical lemmas}
\label{app:technical}

The following result is used in the proof of
Proposition~\ref{prop:general-smallset}.

\begin{lemma}
\label{lem:general-smallset}
Let $p\geq1$,  and $\Gamma$ is a $p\times p$ invertible
matrix. Then there exists a probability measure $\nu$ such that,
for all $M>0$, $\gamma > 0$, there exists
$\epsilon > 0$ such
that for all $g: \mathbb{R}^p\rightarrow{\mathbb{R}_{+}}$ and
$\boldsymbol{\ell} \in (0, M]^p$,
\[
\int_{D}g(u\boldsymbol{\ell} + \Gamma \boldsymbol{\vartheta}) u^{\gamma}
\mathrm{d}u \mathrm{d}\boldsymbol{\vartheta} \geq \epsilon \int g \; \mathrm{d}\nu \;,
\]
where $D = \{(u, \boldsymbol{\vartheta}) \in (0, 1] \times (0, 1]^p, 0< u < \boldsymbol{\vartheta}(1) < \cdots < \boldsymbol{\vartheta}(p)\}$.
\end{lemma}

\begin{proof}
Setting $\boldsymbol{\omega} = u\boldsymbol{\ell} + \Gamma \boldsymbol{\vartheta}$, we have
\[
\int_{D}g(u\boldsymbol{\ell} + \Gamma \boldsymbol{\vartheta}) u^{\gamma}
\mathrm{d}u \mathrm{d}\boldsymbol{\vartheta} = \frac{1}{\det{\Gamma}} \int
g(\boldsymbol{\omega}) u^{\gamma} {\1}_{D}{((u, \Gamma^{-1} \boldsymbol{\omega}
  - u \Gamma^{-1} \boldsymbol{\ell}))} \mathrm{d}u
\mathrm{d}\boldsymbol{\omega} \;.
\]

We choose any $0 < \eta < 1/(p+1)$ so that the open set $B_{\eta}
= \{ \boldsymbol{\vartheta} \in (0,1]^p : \eta < \boldsymbol{\vartheta}(1),
\eta + \boldsymbol{\vartheta}(1) < \boldsymbol{\vartheta}(2), \cdots, \eta +
\boldsymbol{\vartheta}({p-1}) < \boldsymbol{\vartheta}({p}),
\boldsymbol{\vartheta}({p}) < 1-\eta\}$ is not empty.

Denoting $\Gamma B_{\eta} = \{\Gamma \boldsymbol{\vartheta} :
\boldsymbol{\vartheta} \in B_{\eta}\}$ and setting
$$
\eta' = \frac12\eta \min\left(1,\left(\max_{i=1,\dots,p}\sup_{\boldsymbol{\ell}\in(0,M]^p}\left|[\Gamma^{-1}\boldsymbol{\ell}](i)\right| \right)^{-1}\right)\,
$$
we have that for all $u \in (0, \eta']$, $\boldsymbol{\ell} \in (0,
M]^p$ and $\boldsymbol{\omega}\in\Gamma B_{\eta}$, $[u, \Gamma^{-1} \boldsymbol{\omega} - u \Gamma^{-1}\boldsymbol{\ell}]
\in D$. We obtain that
\begin{align*}
\int_{D}g(u\boldsymbol{\ell} + \Gamma \boldsymbol{\vartheta}) u^{\gamma} \mathrm{d}u \mathrm{d}\boldsymbol{\vartheta} &
\geq \frac{1}{\det{\Gamma}}
\int_{\Gamma B_{\eta}}g(\boldsymbol{\omega})
\left(\int_{0}^{\eta'}u^{\gamma} \mathrm{d}u\right) \mathrm{d}\boldsymbol{\omega} \\
& = \frac{\mu^{Leb}(\Gamma B_{\eta})(\eta')^{\gamma + 1}}{(\gamma +
  1)\det{\Gamma}} \int g \mathrm{d}\nu \;,
\end{align*}
where $\mu^{Leb}$ is Lebesgue measure, $\nu$ is the uniform probability measure on $\Gamma B_{\eta}$.
\end{proof}

The following result is used in the proof of Proposition~\ref{prop:drift-step-1}.

\begin{lemma}
\label{lem:appendix-asymptotic-equivalence}
Let $\beta > 0$ and $\beta' > 0$.
Then, as $a\to\infty$,
\begin{equation}
  \label{eq:limit-lem2}
\int_0^{\infty} \rme^{a(\rme^{-\beta t} - 1) - \beta'
  t}\mathrm{d}t \to 0 \;.
\end{equation}
If moreover $\beta' > \beta$, then we have the following asymptotic equivalence
as $a\to\infty$,
\begin{equation}
  \label{eq:asymp-lem2}
\int_0^{\infty} \rme^{a(\rme^{-\beta t} - 1) -
  \beta' t} \mathrm{d}t \sim \frac{1}{a\beta} \;.
\end{equation}
\end{lemma}
\begin{proof}
Setting $\vartheta = \rme^{ - \beta t}$, we get
$$
\int_0^{\infty} \rme^{a(e^{-\beta t} - 1)} (\rme^{ - \beta t})^{\frac{\beta'}{\beta}} \mathrm{d}t =
\frac{1}{\beta} \int_0^{1} \rme^{a(\vartheta - 1)} \vartheta^{\frac{\beta'}{\beta}-1} \mathrm{d}\vartheta \;.
$$
Letting $a\to\infty$, by dominated convergence, we get~(\ref{eq:limit-lem2}).

Now we set $\omega = -a(\vartheta - 1)$ and obtain that
\begin{align*}
\frac{1}{\beta} \int_0^{1} \rme^{a(\vartheta - 1)} \vartheta^{\frac{\beta'}{\beta}-1} \mathrm{d}\vartheta
& = \frac{1}{a\beta} \int_0^{a} \rme^{-\omega} (1-\frac{\omega}{a})^{\frac{\beta'}{\beta}-1} \mathrm{d}\omega \;.
\end{align*}
Letting $a \rightarrow \infty$, we obtain~(\ref{eq:asymp-lem2}) by dominated convergence.
\end{proof}

\bibliographystyle{plain}
\bibliography{hawkes}

\end{document}